\documentclass[11pt]{article}

\usepackage[margin=1in]{geometry}
\usepackage[T1]{fontenc}
\usepackage[utf8]{inputenc}
\usepackage{amsmath,amssymb,amsthm,mathtools}
\usepackage{newpxtext,newpxmath}
\usepackage{microtype}
\usepackage{booktabs}
\usepackage{array,tabularx}
\usepackage{enumitem}
\usepackage{graphicx}
\usepackage{float}
\usepackage{placeins}
\usepackage{subcaption}
\usepackage{xcolor}

\usepackage[backend=biber,style=authoryear,natbib=true,maxcitenames=3,maxbibnames=99,giveninits=true,uniquename=init]{biblatex}
\usepackage[colorlinks=true,linkcolor=blue!55!black,citecolor=blue!55!black,urlcolor=blue!55!black]{hyperref}

\allowdisplaybreaks
\newtheorem{assumption}{Assumption}
\newtheorem{theorem}{Theorem}
\newtheorem{proposition}{Proposition}
\newtheorem{lemma}{Lemma}[section]

\newtheorem{example}{Example}
\theoremstyle{definition}
\newtheorem{definition}{Definition}

\newcommand{\R}{\mathbb R}
\newcommand{\E}{\mathbb E}
\newcommand{\CS}{\mathrm{CS}}
\newcommand{\TS}{\mathrm{TS}}
\newcommand{\supp}{\operatorname{supp}}
\newcommand{\cav}{\operatorname{cav}}
\newcommand{\co}{\operatorname{co}}

\hypersetup{
  pdftitle={Designing Markets for Screening: Optimal Composition and the Geometry of Surplus},
  pdfauthor={Panagiotis Kyriazis},
  pdfkeywords={screening, market composition,
               nonlinear pricing, segmentation, surplus set, adverse selection}
}

\title{Designing Markets for Screening: Optimal Composition and the Geometry of Surplus\thanks{This paper was previously circulated under the titles ``Optimal Market Composition in Monopoly Screening'' and ``The Surplus Possibility Set of Monopoly Screening''. I am grateful to Zeinab Aboutalebi, \"Ozlem Bedre-Defolie, Giacomo Calzolari, Pietro Dall'Ara, Piotr Dworczak, Vasundhara Mallick, Matthew Mitchell, and Alessandro Pavan for valuable comments and suggestions. I also thank the participants of the 4th Italian Junior Workshop in Economic Theory and the 2026 European Economic Theory Conference, as well as seminar audiences at Athens University of Economics and Business and the Micro Working Group at the EUI. I gratefully acknowledge financial support received during this research from the Digital Platforms: Pricing, Variety, and Quality Provision (DIPVAR) project, funded by the European Research Council (ERC) under the European Union’s Horizon 2020 research and innovation programme (grant agreement No. 853123), and from the Deutsche Forschungsgemeinschaft (DFG, German Research Foundation; project number 504870592).}}
\author{Panagiotis Kyriazis\thanks{European University Institute. Email: \href{mailto:panagiotis.kyriazis@eui.eu}{panagiotis.kyriazis@eui.eu}.}}
\date{November 24, 2026}

\begin{document}

\maketitle

\begin{abstract}

We characterize the structure of market compositions that generate Pareto efficient outcomes and the geometry of the surplus set in a screening model with type-dependent costs. Under mild cost-curvature conditions, agent-oriented optima use rather than suppress screening to sustain information rents: all agents are served, allocations strictly separate types, and an atom at the highest type receives the efficient allocation. Each Pareto efficient payoff is implemented by a unique market. When costs are type independent or a type with no gains from trade is feasible, we also characterize the entire surplus set and establish its convexity. Otherwise, type-dependent costs can generate nonconvexity. Market segmentation can expand attainable outcomes only on the rent-suppression side of the set, without increasing maximum agent surplus at any level of principal surplus. The characterization and comparative statics link market formation to the distribution of surplus and identify the limits on payoffs attainable under screening.


\end{abstract}

\newpage
\section{Introduction}
\label{sec:introduction}
Screening lies at the heart of contracting under private information. Firms screen consumers through product lines and nonlinear prices, governments through nonlinear tax schedules, and insurers through menus of coverage and premia. The distribution of agents' private characteristics, the market composition, affects both the gains from trade and how much of those gains the principal must leave to agents as information rents. Screening models typically take market composition as given and ask how the principal should design its menu. We instead study the surplus possibilities of screening as the distribution itself varies. We characterize the market compositions that generate Pareto efficient combinations of agent and principal surplus and study the geometry of the attainable surplus set.

We consider a one-dimensional screening model with quasilinear payoffs and type-dependent costs. We use seller–buyer language throughout. A buyer's type captures her willingness to pay for quality and may also affect the seller's cost of providing it. We assume that willingness to pay rises faster with type than the type-dependent marginal cost, so higher types generate greater gains from trade even when they are more costly to serve. The seller knows the distribution of types, but does not observe the realization, and offers a profit-maximizing menu of qualities and prices. Initially, we assume that any distribution on a fixed interval can be chosen. Each distribution induces a pair of consumer surplus and profit, and we study the resulting set of surplus outcomes by maximizing a weighted sum of the two payoffs. By identifying the optimal compositions for different weights, our analysis provides the limits of surplus outcomes attainable under screening and thus a benchmark for evaluating technologies that shape the markets that sellers screen.

Our first main result characterizes the structure of optimal markets and the menus they induce. When seller profit receives at least as much weight as consumer surplus, the market concentrates entirely at the highest type. The allocation is first-best, and the seller extracts all surplus. When consumer surplus receives greater weight, the optimal market instead has a heterogeneous structure. All buyers in the market are served, distinct types receive different qualities, and a positive share of buyers have the highest type and receive first-best quality. The seller's induced menu thus offers a continuum of vertically differentiated products, extending to efficient quality at the top. Under a mild curvature condition on the seller's cost, the optimal market is unique. As the weight on seller profit increases from zero to parity with consumer surplus, these markets trace the entire Pareto frontier of the attainable surplus set, so every Pareto efficient payoff is implemented by a unique market composition. Thus, within our unrestricted benchmark, the familiar pattern of buyers purchasing only a few discrete quality tiers, such as basic, standard, and premium,  indicates a market that falls short of Pareto efficiency.

 Heterogeneity of a nondegenerate frontier market is sustained because it generates information rents. Contracts offered to lower types limit what the seller can extract from higher types. Bunching distinct types leaves further gains unexploited. When distinct types receive the same quality, shifting quality from lower to higher types increases the gains from trade. Higher types value the additional quality more, even after accounting for any additional cost of serving them. To realize these gains, the market must be reshaped so that product differentiation becomes the seller’s optimal response. We show that such a change in market composition can raise total surplus while leaving seller profit unchanged to first order, so consumers capture these additional gains. Differentiation, therefore, remains valuable even when consumer surplus alone is maximized.

The atom at the top reflects a complementary force. Suppose we rank buyers with respect to their valuations. Without an atom at the highest feasible valuation, composition can be adjusted to raise quality for a small group of buyers ranked the highest while leaving quality unchanged at lower ranks. Below this group, this adjustment induces valuations to rise by more at higher ranks, increasing the rents needed to prevent buyers from choosing contracts intended for ones with lower types. When the group whose quality changes is sufficiently small, these gains in information rents outweigh any losses, and seller profit also rises. Optimality therefore requires a positive share of buyers to have the highest feasible valuation and receive first-best quality. These buyers still earn rents because  contracts intended for lower types remain available to them.

Under a stronger curvature condition, we characterize maximum consumer surplus at every feasible profit level, including levels below those on the Pareto frontier. We then obtain a complete characterization of the surplus set in two cases. If a type with no gains from trade is feasible, adding buyers of this type scales both payoffs down in the same proportion. Together with the boundary characterization, this yields a compact, convex set. The set is also compact and convex when costs are type independent, even if all feasible types have gains from trade. In that case, consumer surplus can be varied continuously down to zero while profit is held fixed. With type-dependent costs and no type with zero gains from trade, convexity can fail. We construct an example with two attainable surplus pairs for which no interior convex combination is attainable in a single market.

This nonconvexity creates a role for market segmentation. If the seller observes the segment label and screens each selected market separately, that is, second- and third-degree price discrimination can be performed simultaneously, the resulting surplus set is the convex hull of the single-market set. Segmentation adds nothing when the single-market set is already convex. When it is nonconvex, it convexifies it by generating the missing outcomes on the minimum consumer surplus (or rent-suppression) side of the set. Crucially, it does not shift the maximum consumer surplus boundary. At every level of seller profit, the greatest attainable consumer surplus is still generated by a single market. Segmentation, therefore, does not create new Pareto-efficient outcomes, so it cannot improve the efficient trade-off between consumer surplus and profit.

Neither the frontier characterization nor the convexity results follow from convexity of the space of distributions. Mixing two distributions changes the seller's optimal menu and need not average their payoffs. Moreover, the objective of the composition problem is not concave. Our approach first shows that the optimization problem can be restricted to Myerson-regular markets and reformulates choice of composition as a choice of monotone virtual values in quantile space. Feasible perturbations establish the structure of every optimum. Within this class, the first-order condition reduces to a boundary-value problem whose solution is unique under the relevant curvature condition. Together with existence, this yields a global characterization despite nonconcavity of the objective. It also makes a rich comparative statics analysis tractable.

Our comparative statics results distinguish changes in a given type's allocation and rent from changes in the types represented in the market. Under private values, increasing the weight on consumers results in the
optimal distribution shifting downward in the first-order stochastic dominance sense, and having a longer separating region and a smaller atom at the top. At a given valuation in the common separating region, buyers receive higher quality, so the downward distortions due to screening are mitigated, and earn higher rents. Total surplus nevertheless falls whenever the optimal market changes. Type-dependent costs can produce a different pattern. If costs rise sufficiently steeply near the top and the common feasible minimum type is sufficiently high, greater consumer orientation lowers rents at every type above that minimum. Expected consumer surplus still rises because the market shifts toward higher types, whose rents are larger within either market. Thus, average consumer surplus can increase even though no type receives a higher rent.

Changing the lower endpoint of the feasible type interval provides a different comparative static. At a fixed consumer-oriented welfare weight, raising the minimum feasible type leaves the optimum unchanged until the restriction binds. Thereafter, the optimal market is the original one conditional on types above the new minimum. Retained types receive the same qualities but all pay more by the same amount, equal to the former rent of the new lowest type. Consumer surplus falls, while profit and total surplus rise.

Type-dependent costs also let us compare environments in which gross willingness to pay and cost co-move differently. We hold fixed the type-independent cost of quality and the feasible range of net values, defined as willingness to pay less type-dependent marginal cost. Stronger co-movement means that gross willingness to pay and type-dependent marginal cost both rise more for a given increase in net value. Under the mild curvature condition, reoptimizing composition at the same consumer-oriented welfare weights yields weakly lower quality at each net value in the common separating support. Under the stronger condition, however, maximum consumer surplus is weakly higher at every previously attainable profit level. Thus, stronger value–cost co-movement can reduce quality at a fixed net value while expanding the maximum consumer surplus frontier at a given profit level.

Two extensions examine restrictions on composition.
The first allows the choice of market composition to be costly. A natural interpretation of this problem is a hold-up one. Agents themselves or an upstream party incur a cost to change the distribution of buyers' valuations before the seller designs its menu, with the aim of limiting subsequent surplus extraction. We capture this through a broad class of costs, including technologies in which a first-order stochastically higher type distribution is weakly cheaper to induce. The resulting trade-off is between paying to create heterogeneity that generates information rents and leaving the market concentrated at the highest type, where the seller extracts all surplus. Sufficiently costly attenuation can make the latter optimal even when consumer surplus receives greater weight than profit. Whenever a nondegenerate market is optimal, however, it retains a strictly separating interior, an efficiently served atom at the highest type, and no positive mass of excluded buyers. Moreover, when attenuation costs are separable across buyers and vanish continuously at the highest type, some attenuation remains optimal whenever consumer surplus receives greater weight than profit.

The second extension fixes the mean type under private values, giving the information-design problem associated with a binary prior, studied more generally in \citet{Kyriazis2026Learning}. For welfare weights that strictly favor consumers, the connected upper separating region and efficiently served top atom survive under the relevant curvature conditions in that analysis. The mean constraint can nevertheless make exclusion optimal despite gains from trade for every buyer, and allow market segmentation to improve the consumer surplus frontier.

\subsection{Related literature}

This paper builds on the theory of one-dimensional screening. In the nonlinear-pricing models of \citet{mussarosen1978} and \citet{maskinriley1984}, the distribution of buyer types is taken as primitive and the central object is the seller's optimal menu. \citet{Hellwig2010} develops a unified treatment of screening with arbitrary type distributions, while \citet{ChadeSchlee2012} and \citet{ChadeSchlee2020} study monopoly insurance with common values and arbitrary type distributions. The paper is also related to \citet{SharkeySibley1993}, who characterize Pareto-optimal nonlinear tariffs when a regulator places different welfare weights on customer types.

The paper also relates to work that endogenizes the environment faced by a seller before it chooses a price or mechanism. \citet{JohnsonMyatt2006} study demand transformations induced by advertising, marketing, and product design, \citet{BarIsaacCaruanaCunat2023} study targeted product design, and \citet{DaiKoh2025} study costly manipulation of consumer valuations. Closer to us, \citet{CondorelliSzentes2020} allow a buyer to choose the distribution of her valuation before the seller observes it and posts a price. Additionally, \citet{LewisSappington1994} study a seller's incentive to provide buyers with private information about their tastes, while \citet{OttavianiPrat2001} analyze the value of public information to a monopolist. More recently,  \citet{roessler_szentes_2017} study buyer-optimal learning before monopoly pricing, while \citet{DebRoesler2024} study consumer-optimal learning before multidimensional screening. By contrast, \citet{XuYang2026} study a consumer-minded intermediary that provides information after a monopolist sets its price
and show that greater weight on consumer surplus can generate Pareto-worse outcomes. Other recent work studies information design in conjunction with pricing and screening; see \citet{WeiGreen2024}, \citet{ScreeningwithPersuasion}, \citet{ShiZhangChen2026}, and \citet{BergemannHeumannMorris2026Integrated}. Our contribution is to study the characteristics an optimal market should have prior to screening and the implications for the division of surplus. Most directly, in \citet{Kyriazis2026Learning}, an information designer provides information to agents before the principal chooses the menu. When the prior distribution in \citet{Kyriazis2026Learning} is binary, that model is a constrained version of our private-value problem. We instead allow unrestricted choice of market composition and type-dependent costs.

Finally, the paper relates to work on market segmentation. \citet{BergemannBrooksMorris2015} characterize consumer- and producer-surplus outcomes under third-degree price discrimination, while \citet{HaghpanahSiegel2022} study when the multiproduct surplus frontier can be attained. \citet{KartikZhong2026} characterize attainable payoff sets in a posted-price setting with interdependent values by varying information structures for a fixed prior, whereas we vary market composition itself before the seller offers a screening menu. \citet{Yang2022ConsumerData} studies segmentations sold by a data broker to a producer with private cost.  \citet{bergemannheumannwang2026} fix an aggregate market and characterize the consumer-optimal segmentation when the monopolist can offer a separate menu in each segment. We instead characterize optimal markets directly and then consider the role of segmentation.

\section{Model}
\label{sec:model}

Fix $\bar\theta>0$ and $\bar x>0$. For each
$\underline\theta\in[0,\bar\theta)$, let
$\Theta_0:=[0,\bar\theta]$ and $\Theta_{\underline\theta}:=[\underline\theta,\bar\theta]
\subseteq\Theta_0$. A principal offers a one-dimensional allocation $x\in[0,\bar x]$ and a transfer $t\in\R$ to an agent with private type $\theta\in\Theta_{\underline\theta}$. We interpret $x$ as product or service quality. Although we use seller--buyer language in most of the paper, the same environment also covers simple models of insurance, taxation, procurement and regulation.

\subsection{Preferences, cost, and total surplus}

The buyer's utility is
\begin{equation}
 u_B(\theta,x,t)=\theta x-t.
 \label{eq:model-buyer-utility}
\end{equation}
The type $\theta$ is the buyer's marginal willingness to pay for quality $x$. The buyer's outside option is normalized to zero.

The seller's cost of assigning $x$ to type $\theta$ is
\begin{equation*}
 c_S(\theta,x)=m(\theta)x+C(x),
\end{equation*}
and her payoff is
\begin{equation}
 u_S(\theta,x,t)=t-m(\theta)x-C(x).
 \label{eq:model-seller-payoff}
\end{equation}
The function $m(\theta)$ is the type-dependent component of marginal cost. It permits high valuation buyers to be more expensive to serve. The function $C(x)$ is a type-independent cost of quality provision.

The total surplus is
\begin{align*}
 S(\theta,x)
 &=u_B(\theta,x,t)+u_S(\theta,x,t)=[\theta-m(\theta)]x-C(x).
\end{align*}
Define marginal net social value as the difference between willingness to pay $\theta$ and type-dependent marginal cost $m(\theta)$,
\begin{equation*}
 \nu(\theta):=\theta-m(\theta),
 \qquad
 \bar\nu:=\nu(\bar\theta).
\end{equation*}
Thus, $\nu(\theta)$ is an index that determines efficient trade. We call this a common-index model because the same private type governs demand, cost, and net social value.

We maintain the following assumptions throughout the paper.\footnote{The type-dependent primitives are defined on the domain $\Theta_0$. In the environment indexed by
$\underline\theta$, the buyer's realized type lies in
$\Theta_{\underline\theta}$.}
\begin{assumption}
\label{ass:main-nu}
The function $m:\Theta_0\to\R_+$ is twice continuously differentiable and satisfies
\begin{equation*}
 m(0)=0,
 \qquad
 0\leq m'(\theta)\leq1-\underline r
 \quad\text{for some }\underline r\in(0,1].
\end{equation*}
\end{assumption}

\begin{assumption}
\label{ass:main-C}
The function $C:[0,\bar x]\to\R_+$ is continuously differentiable on the closed interval, twice continuously differentiable on $(0,\bar x]$, and satisfies
$C(0)=C'(0)=0$, $C''(x)>0$ for $x>0$ and $C'(\bar x)=\bar\nu$.
\end{assumption}

The inequality $m'\geq0$ in Assumption~\ref{ass:main-nu} implies that cost rises with willingness to pay. The upper bound yields $m'<1$ and, thus, requires gross willingness to pay to rise faster than the type-dependent component of the marginal cost. Equivalently,
\[\nu'(\theta)=1-m'(\theta)\in[\underline r,1],\]
so higher types have strictly greater marginal net social value. Assumption~\ref{ass:main-C} gives a unique pointwise allocation for every effective marginal value.

For $z\leq\bar\nu$, define
\begin{equation*}
 q(z):=\arg\max_{x\in[0,\bar x]}\{zx-C(x)\},
 \qquad
 \pi(z):=\max_{x\in[0,\bar x]}\{zx-C(x)\}.
\end{equation*}
Strict convexity makes $q$ single-valued; $q(z)=\pi(z)=0$ for $z\leq0$, $q$ is increasing on the positive range, and $\pi'(z)=q(z)$. The first-best quality and surplus are
$x^{FB}(\theta)=q(\nu(\theta))$ and $S^{FB}(\theta)=\pi(\nu(\theta))$.
Because $\nu$ is strictly increasing, higher types have weakly higher first-best qualities.

Finally, we say that a type $\tilde \theta \in\Theta_{\underline\theta}$ is payoff-neutral if $x^{FB}(\tilde\theta)=0$, that is, if that type would not be served even under the first-best.

\subsection{Timing, market composition, and the support problems}

Before the seller offers a menu, an upstream institution chooses a
market composition, i.e., a distribution of types, $G\in\Delta(\Theta_{\underline\theta})$.
The seller observes $G$ but not the buyer's realized type and then offers a profit-maximizing incentive-compatible and individually rational menu. The buyer observes the realized type and selects a menu item. The optimal market need not be supported on the
entire feasible interval. By \emph{exclusion}, we mean that the seller
assigns zero quality to buyers represented in the chosen market.

For a market $G$, let $(x_G,t_G)$ be a seller-optimal direct mechanism\footnote{As is standard, it is without loss to focus on direct mechanisms by the revelation principle.}. Consumer surplus, profit, and total surplus are
\begin{align*}
 \CS(G)&:=\E_G[\theta x_G(\theta)-t_G(\theta)],
 \\
 \Pi(G)&:=\E_G[t_G(\theta)-m(\theta)x_G(\theta)-C(x_G(\theta))],
 \\
 \TS(G)&:=\CS(G)+\Pi(G)
 =\E_G[\nu(\theta)x_G(\theta)-C(x_G(\theta))].
\end{align*}
The set of outcomes induced by varying the market composition is
\begin{equation*}
 \mathcal V_{\underline\theta}
 :=\{(\CS(G),\Pi(G)):G\in\Delta(\Theta_{\underline\theta})\}.
\end{equation*}
We use $\mathcal F^P(\mathcal V_{\underline\theta})$ for its Pareto
frontier and $\mathcal F^L(\mathcal V_{\underline\theta})$ for the
continuation of the maximum-consumer surplus boundary toward lower
profits:
\begin{align*}
 \mathcal F^P(\mathcal V_{\underline\theta})
 &:={\Bigl\{}(c,p)\in\mathcal V_{\underline\theta}:
 \nexists(c',p')\in\mathcal V_{\underline\theta}
 \text{ with }c'\geq c,\ p'\geq p,
 \text{ and at least one strict}{\Bigr\}},\\
 \mathcal F^L(\mathcal V_{\underline\theta})
 &:={\Bigl\{}(c,p)\in\mathcal V_{\underline\theta}:
 \nexists(c',p')\in\mathcal V_{\underline\theta}
 \text{ with }c'\geq c,\ p'\leq p,
 \text{ and at least one strict}{\Bigr\}}.
\end{align*}
For each $\beta\in\R$, define
\begin{equation}
 U_{\underline\theta}(\beta)
 :=\max_{G\in\Delta(\Theta_{\underline\theta})}
 \{\CS(G)+\beta\Pi(G)\}.
 \label{eq:model-support-problem}
\end{equation}
The coefficient $\beta$ can be interpreted as the institution's relative weight on profit. When $\beta<0$, profit is treated as a ``bad''.

\subsection{Normalizations and applications}
\label{subsec:normalization-applications}

The payoff specification we maintain can be obtained from more general primitive screening environments. Suppose primitive agent utility is $v(a,y)-\tau$, where
$v_{ay}(a,y)=z(a)k(y)$
for continuous and strictly positive $z$ and $k$, and the type-only
component of $v$ is absorbed by the agent's outside option. If the
principal's primitive payoff is
$\tau-\mu(a)X(y)-K(y)$, where $X'(y)=k(y)$, then monotone relabelings
of type and allocation, together with a transfer adjustment, yield
\eqref{eq:model-buyer-utility} and \eqref{eq:model-seller-payoff} for
suitable $m$ and $C$.

The conditions $m(0)=0$ and $C'(0)=0$ are also normalizations. Suppose $\widetilde C(0)=0$,
$\widetilde C$ is differentiable and strictly convex, and
$\widetilde\nu(\vartheta):=
\vartheta-\widetilde m(\vartheta)$ is strictly increasing. If the
primitive type domain contains a type $\vartheta^0$, weakly
below all admissible types, satisfying
$\widetilde\nu(\vartheta^0)=\widetilde C'(0)$,
an affine change of variables maps $\vartheta^0$ to zero and yields
$m(0)=C(0)=C'(0)=0$, while rescaling agent and principal surplus by the same factor.
Strict convexity makes $\vartheta^0$ the highest type that does not
trade under the first best. Hence, $\underline\theta=0$ when
this payoff-neutral type is admissible, whereas
$\underline\theta>0$ when every admissible type trades under the first best.

Our model nests various screening environments as we next illustrate.

\paragraph{Monopoly Screening}
Under private values, the seller's cost is independent of type. In this case, our model nests the standard monopoly screening model of \citet{mussarosen1978} by setting
\begin{equation*}
 m(v)\equiv0,
 \qquad
 \nu(v)=v,
 \qquad
 \bar\theta=\bar\nu=1,
 \qquad
 v\in[\underline v,1],
\end{equation*}
where $\underline v\in[0,1)$ is the normalized lowest admissible
valuation. Under the normalization in
Section \ref{subsec:normalization-applications}, the primitive feasible interval
starts weakly above the highest first-best no-trade valuation.
Its lower endpoint maps to $\underline v=0$ exactly when it equals that
valuation, and maps to $\underline v>0$ otherwise.

\paragraph{Insurance and healthcare.}
Let $p$ be a privately known loss probability or an index of insurance needs, let $q$ be indemnity or coverage, and let $\tau$ be the premium. A simple specification is
\begin{equation*}
 \widetilde u_A(p,q,\tau)=\chi p q-\frac{\gamma}{2}q^2-\tau,
 \qquad
 \widetilde u_P(p,q,\tau)=\tau-p q-K(q),
 \qquad \chi>1,\ \gamma>0.
\end{equation*}
The parameter $\chi$ is the value of one unit of indemnity in the loss state relative to its actuarial cost, while the quadratic term captures diminishing marginal insurance value. Setting
\begin{equation*}
 \theta=\chi p,
 \qquad x=q,
 \qquad t=\tau+\frac{\gamma}{2}q^2
\end{equation*}
gives
\begin{equation*}
 m(\theta)=\frac{\theta}{\chi},
 \qquad
 \nu(\theta)=\left(1-\frac1\chi\right)\theta,
 \qquad
 C(x)=K(x)+\frac{\gamma}{2}x^2.
\end{equation*}
Buyers with higher needs therefore value coverage more and generate higher expected claims, while remaining more valuable to insure.
This model is a simplified version of the insurance environments studied, for example, by \citet{stiglitz1977} and \citet{ChadeSchlee2012,ChadeSchlee2020}.

\paragraph{Labor taxation.}
Let a worker of type $a$ choose labor or output $\ell$ and receive consumption $c$. Suppose
\begin{equation*}
 \widetilde u_A(a,\ell,c)=c+aX(\ell)-A(\ell),
 \qquad
 \widetilde u_P(a,\ell,c)=V(\ell)-c.
 \label{eq:model-tax-primitive}
\end{equation*}
Define $x=X(\ell)$ and $t=A(\ell)-c$. The worker obtains $ax-t$, while the fiscal authority obtains
\begin{equation*}
 t-C(x),
 \qquad
 C(x)=A(X^{-1}(x))-V(X^{-1}(x)).
 \label{eq:model-tax-map}
\end{equation*}
This example captures the incentive-provision side of nonlinear taxation, in which higher types are more productive or have a lower effective disutility of work, and a nonlinear schedule trades off output against information rents \citep{mirrlees1971}.

\subsection{Screening a Given Market}
\label{sec:reduction}

The downstream screening problem is standard. We start by reducing the upstream problem to a more tractable problem over virtual value profiles.

\subsubsection{Virtual net values and adjusted revenue}

For a smooth market with density $g$, well-known incentive-compatibility and envelope arguments give the seller's raw virtual net value
\begin{equation}
 \widehat\varphi_G(\theta)
 =\nu(\theta)-\frac{1-G(\theta)}{g(\theta)}.
 \label{eq:reduction-smooth-type-virtual-value}
\end{equation}
Relative to the familiar private-value formula, the type-dependent cost changes the willingness to pay term $\theta$ to net value $\nu(\theta)$.

To cover arbitrary markets, let $Q$ be the lower quantile and let $Q^+(u):=\inf\{\theta:G(\theta)>u\}$ be the strict-upper quantile. Define the adjusted raw revenue curve
\begin{equation}
 \widehat R_G(u)
 :=(1-u)Q^+(u)-\int_u^1m(Q(s))\,ds.
 \label{eq:reduction-adjusted-raw}
\end{equation}
To interpret this curve, consider a small increase in quality for
the upper tail of the market starting at rank $u$, choosing the cutoff so that
buyers with the same type are treated alike. The cutoff type is
$Q^+(u)$, so the seller can charge each buyer in this tail $Q^+(u)$
per additional unit of quality. The first term is the resulting
revenue per unit of the increase. The second subtracts the
corresponding type-dependent cost.

Let $R_G:=\cav(\widehat R_G)$ be the least concave majorant of this
curve, and define the ironed virtual net value\footnote{Since we allow for nonsmooth concave revenue curves,
we use $R'_{G,+}$ to denote the right derivative of $R_G$ at
quantile $u$. At the endpoints, we use the conventions
$\phi_G(0):=\lim_{u\downarrow0}\phi_G(u)\in[-\infty,\bar\nu]$
and
$\phi_G(1):=\lim_{u\uparrow1}\phi_G(u)\in\mathbb R$.
The profile is integrable on $(0,1)$ and bounded above by
$\bar\nu$, but need not be bounded below.
We also extend $q$ and $\pi$ by $q(-\infty)=\pi(-\infty)=0$.}
\begin{equation*}
 \phi_G(u):=-R'_{G,+}(u),\qquad 0<u<1.
\end{equation*}
Concavifying the revenue curve makes $\phi_G$ nondecreasing,
so the induced quality schedule is monotone.

When $G$ is smooth,
$-\widehat R_G'(u)
 =\nu(Q(u))-(1-u)Q'(u)$,
which is \eqref{eq:reduction-smooth-type-virtual-value} evaluated at $\theta=Q(u)$.

\paragraph{Seller's problem.}
For every market $G$, the seller's problem can be written in quantile space as
\begin{equation*}
 \max_{\substack{x:[0,1]\to[0,\bar x]\\ x\ \mathrm{nondecreasing}}}
 \int_0^1\{\phi_G(u)x(u)-C(x(u))\}\,du.
\end{equation*}
Because $\phi_G$ is nondecreasing, the unique maximizing quantile allocation is $x_G=q(\phi_G)$ almost everywhere, and the associated payoffs are
\begin{align*}
 \Pi(G)&=\int_0^1\pi(\phi_G(u))\,du,
 \\
 \TS(G)&=\int_0^1\{\nu(Q(u))q(\phi_G(u))-C(q(\phi_G(u)))\}\,du,
 \\
 \CS(G)&=\int_0^1\{\nu(Q(u))-\phi_G(u)\}q(\phi_G(u))\,du.
\end{align*}
The first formula shows that profit depends only on ironed virtual net values. Consumer surplus also depends on the gap between actual and virtual net values.

\subsection{Regularization and truncation}
\label{subsec:main-regularization}

The main technical simplification is that for the support problems \eqref{eq:model-support-problem}, it is
without loss to restrict attention to Myerson-regular markets.
The following lemma replaces any market by a regular one with
the same profit and weakly higher consumer surplus.

\begin{lemma}[Regularization]
\label{lem:main-regularization}
Maintain Assumptions~\ref{ass:main-nu}--\ref{ass:main-C}. Let $G$ be any market, with quantile $Q$, concavified adjusted revenue curve $R_G$, and ironed virtual-net-value profile $\phi_G$. There is a unique quantile $\widetilde Q$ satisfying
\begin{equation*}
 (1-u)\widetilde Q(u)
 =R_G(u)+\int_u^1m(\widetilde Q(s))\,ds,
 \qquad u<1.
\end{equation*}
It is nondecreasing and obeys
$Q(u)\leq\widetilde Q(u)\leq\bar\theta$, $u<1$. The associated market $\widetilde G$ has adjusted raw revenue curve $R_G$. It is therefore regular, has the same ironed virtual net values $\phi_G$ and the same seller-optimal rank allocation $q(\phi_G)$ as $G$, and satisfies $\Pi(\widetilde G)=\Pi(G)$ and
\[\CS(\widetilde G)-\CS(G)=\int_0^1
 \bigl[\nu(\widetilde Q(u))-\nu(Q(u))\bigr]
 q(\phi_G(u))\,du\geq0,\]
with equality if and only if
$\bigl[\widetilde Q(u)-Q(u)\bigr]q(\phi_G(u))=0$ for almost every $u$. Regularization also preserves every lower-support restriction.
\end{lemma}

Regularization preserves the ironed virtual-value profile and
quality at each rank while weakly raising actual net values.
Profit is unchanged, so any resulting gain in total surplus
accrues to consumers. Consequently, the reduction preserves the
maximum of $\CS+\beta\Pi$ for every $\beta\in\mathbb R$, as well
as maximum consumer surplus at each fixed profit. Together with the no-exclusion result below, the equality condition
will allow us to transfer the structural conclusions from regular
profiles to every optimizer of the original market problem.

Moreover, for any regular profile $\phi$, replacing it by $\phi^+:=\max\{\phi,0\}$ weakly raises the induced quantile and leaves both consumer surplus and profit unchanged. Ranks at which $\phi\leq0$ receive zero quality and generate no payoffs, while the profile and induced quantile are unchanged at every rank with $\phi>0$. Thus, we may restrict attention to nonnegative regular profiles.

\subsection{The reduced problem}
\label{sec:reduced-problem}

Let
\[
 \Phi:=\left\{\phi:[0,1]\to[0,\bar\nu]:
 \phi\text{ is nondecreasing and right-continuous on }[0,1), \ \phi(1)=\displaystyle\lim_{u\uparrow1}\phi(u)\right\}.\]
For $\phi\in\Phi$, let $Q_\phi$ be the unique solution of
\begin{equation}\label{eq:reconstruction}
 (1-u)Q_\phi(u)
 =\int_u^1\{\phi(s)+Q_\phi(s)-\nu(Q_\phi(s))\}\,ds, \qquad 0\leq u<1.
\end{equation}
We refer to this equation as the \emph{reconstruction equation} because it recovers the type distribution that induces the chosen
virtual-value profile. Under private values, $m\equiv0$, so it reduces to
\[
Q_\phi(u)=\frac{1}{1-u}\int_u^1\phi(s)\,ds,
\qquad 0\leq u<1.
\]
The valuation at rank $u$ is therefore the average virtual value in
the upper tail of the market starting at that rank. A feasible increase in $\phi$ on an interval of ranks raises the induced types below that
interval. Quality at those lower ranks remains unchanged, since the
seller's allocation $q(\phi(u))$ depends only on the virtual value at
that rank. With type-dependent costs, the upper-tail average also includes
$m(Q_\phi(s))$. The types being recovered therefore determine part
of the cost of serving them.
Since $m$ is nondecreasing, an increase in virtual values still raises
the induced types below the affected interval, while leaving quality
at those lower ranks unchanged.

Put $M(\phi):=Q_\phi(0)$. After regularization and truncation, the problem can be restricted to profiles in $\Phi$. Such a profile satisfies the minimum-type restriction when
$M(\phi)\geq\underline\theta$. Thus, define
\[
 \Phi_{\underline\theta}
 :=\{\phi\in\Phi:M(\phi)\geq\underline\theta\}.
\]
When a payoff-neutral type is feasible and normalized to zero, this
constraint reduces to \(M(\phi)\geq0\) and is slack at every
nondegenerate optimum.

The objective is given by
\begin{equation*}
 J_\beta(\phi)
 :=\int_0^1\Big([\nu(Q_\phi(u))-\phi(u)]q(\phi(u))+\beta\pi(\phi(u))\Big)\,du.
\end{equation*}
Thus, the problem can be written as
\begin{equation*}
 \max_{\phi\in\Phi_{\underline\theta}}J_\beta(\phi).
\end{equation*}
The next section solves this problem and provides the structure and the characterization of optimal markets.

\section{Optimal Market Composition and the Pareto Frontier}
\label{sec:frontier}

Assumptions \ref{ass:main-nu} and \ref{ass:main-C} are enough to establish the structure of the markets that generate outcomes on the frontier of the surplus set. To fully characterize these markets and prove their uniqueness for each weight, we need the following assumption.

\begin{assumption}[Curvature]
\label{ass:main-curvature}
The function $C$ is three times continuously differentiable on $(0,\bar x]$ and
\begin{equation*}
 -\frac{xC'''(x)}{C''(x)}<2, \quad x\in(0,\bar x].
\end{equation*}
\end{assumption}

This assumption implies that the second derivative of the cost function cannot decrease faster than $1/x^2$. It holds, for instance, for constant-elasticity costs or whenever $C'''\geq0$.

For $\beta\in[0,1]$, let $G_{\beta,\underline\theta}$ denote the optimizer and write $C_{\beta,\underline\theta}=\CS(G_{\beta,\underline\theta})$ and $P_{\beta,\underline\theta}=\Pi(G_{\beta,\underline\theta})$. Moreover, let
$\theta_\beta^0:=\min\supp G_{\beta,0}$
denote the lowest type selected when the feasible interval is
$[0,\bar\theta]$.
For a regular market $G$, write
$x_{\phi}(u)=q(\phi_G(u))$
for the seller-optimal allocation in quantile space. We say that $G$ is strictly separated on a quantile interval $I$ if $u\mapsto x_{\phi}(u)$ is strictly increasing on $I$. In type space, a seller-optimal mechanism is strictly separated on a set of realized types $S\subseteq\supp G$ if
\[
 \theta<\theta'\text{ in }S
 \quad\Longrightarrow\quad
 x_G(\theta)<x_G(\theta').
\]
Whenever $Q_G$ is strictly increasing on $I$, the two formulations agree for $S=Q_G(I)$.
We are now ready to state the first main result of the paper.

\begin{theorem}[Optimal markets and the Pareto frontier]
\label{thm:main-structure}
Maintain Assumptions~\ref{ass:main-nu}--\ref{ass:main-C}.
\begin{enumerate}[label=(\roman*)]
 \item If $\beta\geq1$, the unique optimizer is the degenerate market $\delta_{\bar\theta}$.
 \item If $\beta\in[0,1)$, the optimizer has connected support, no positive mass of realized types is assigned zero quality, the active interior is strictly separated, and a positive atom at the highest type is served at the first-best quality. There is no atom below the highest type.

 If, in addition, Assumption~\ref{ass:main-curvature} holds, the
optimizer is unique and its support is
$\supp G_{\beta,\underline\theta}
 =[\ell_{\beta,\underline\theta},\bar\theta]$, where $\ell_{\beta,\underline\theta}
 =\max\{\underline\theta,\theta_\beta^0\}$.
Moreover, $\phi(0)>0$, so the lowest realized type also receives
positive quality.

 \item Under Assumption~\ref{ass:main-curvature}, the weighted problems trace the entire Pareto frontier:
 \begin{equation*}
  \mathcal F^P(\mathcal V_{\underline\theta})
  =\{(C_{\beta,\underline\theta},P_{\beta,\underline\theta}):\beta\in[0,1]\}.
 \end{equation*}
 Every Pareto efficient payoff pair is implemented by a unique market with the structure identified in (i) and (ii). Different weights can select the same market only when the minimum-type constraint binds.
\end{enumerate}
\end{theorem}

 Figure \ref{fig:private-value-benchmark} illustrates the optimal market in value and quantile space and the corresponding optimal virtual value in quantile space in the private values specification of the model when the cost is quadratic and $\beta=0$, hence the consumer-optimal market.

\begin{figure}[t]
 \centering
 \captionsetup{font=small,skip=4pt}
 \begin{subfigure}{0.5\textwidth}
  \centering
  \includegraphics[width=\textwidth]{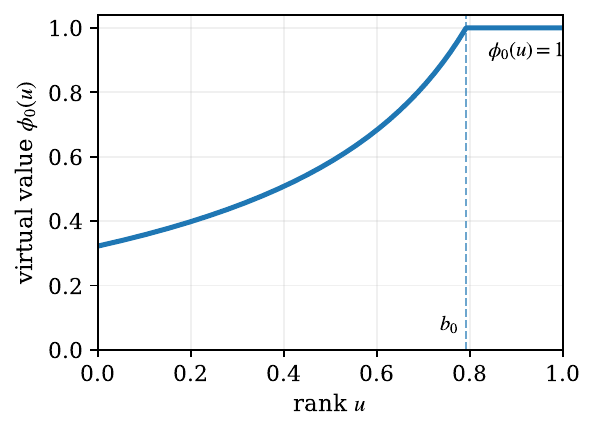}
  \caption{Virtual values in quantile space}
 \end{subfigure}\hfill
 \begin{subfigure}{0.5\textwidth}
  \centering
  \includegraphics[width=\textwidth]{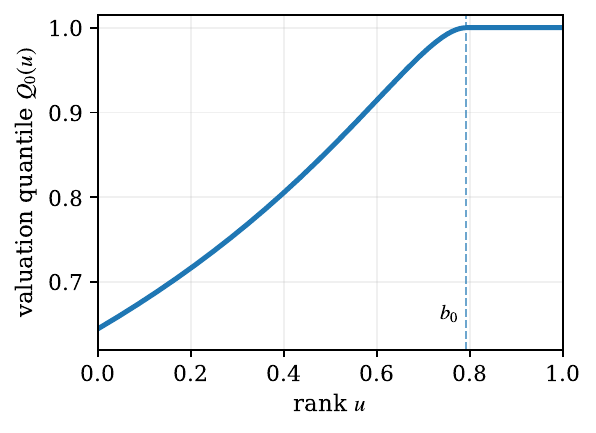}
  \caption{Type quantiles}
 \end{subfigure}\hfill
 \begin{subfigure}{0.5\textwidth}
  \centering
  \includegraphics[width=\textwidth]{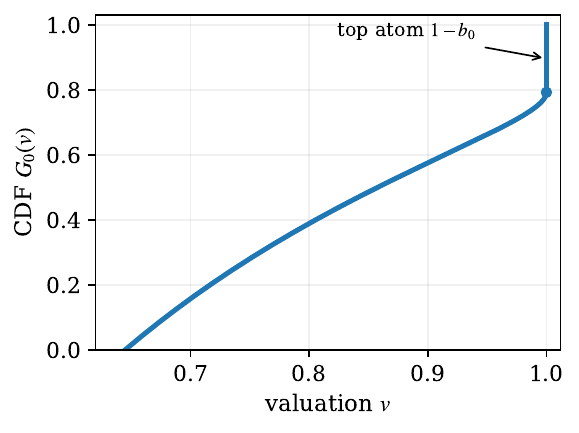}
  \caption{Market composition}
 \end{subfigure}
 \caption{The market that maximizes consumer surplus in the quadratic private-value benchmark. Cost is $C(x)=x^2/2$, valuations lie in $[0,1]$, and the profit weight is $\beta=0$.}
 \label{fig:private-value-benchmark}
\end{figure}

\subsection{Proof Outline and Discussion of Structural Properties}

\paragraph{Degenerate solution with profit-biased objective.}
The first part of the theorem is intuitive. For $\beta=1$, the problem
is equivalent to maximizing total surplus. For every type and quality,
\[
\nu(\theta)x-C(x)\leq \bar\nu x-C(x)\leq \pi(\bar\nu),
\]
with equality only at $\theta=\bar\theta$ and $x=\bar x$. Thus, any
positive mass on lower types reduces average total surplus, even when
those types receive their own first-best qualities. The degenerate
top market uniquely attains the bound, maximizes profit, and
leaves consumers no surplus. For $\beta>1$, the objective places still
greater weight on profit, so the same market remains uniquely optimal.

\paragraph{Why there is no positive-mass exclusion.}
There is nothing to gain from buyers whom the seller would exclude. Suppose mass $a>0$ of selected buyers receive zero quality. Monotonicity implies that these buyers form an initial block and do not generate consumer surplus or profit. Removing this block and renormalizing the market leaves the seller's problem and thus the menu unchanged, apart from rescaling. Because payoffs are averages over the selected market, both consumer surplus and profit are divided by $1-a$. The objective therefore strictly increases whenever the original market has positive value, contradicting optimality.

\paragraph{Perturbing the virtual value.}
A change in virtual values affects both quality and the types represented
in the market. An increase confined to a small interval of ranks raises
quality on that interval and, through the reconstruction equation \eqref{eq:reconstruction}, raises
the induced types at lower ranks. Quality at those lower ranks remains
unchanged. Their contribution to total surplus therefore rises because
higher net values are assigned the same allocations.

The \emph{propagation state} $A_\phi(u)$ captures the marginal surplus
effect of the type changes induced by increasing virtual values near
rank $u$, holding qualities at their original levels. It satisfies
\[
 (1-u)A_\phi'(u)
 =[1-\nu'(Q_\phi(u))]A_\phi(u)
 +\nu'(Q_\phi(u))q(\phi(u)),
 \qquad
 A_\phi(0)=0.
\]
Under private values, this expression reduces to
\[A_\phi(u)=\int_0^u\frac{q(\phi(s))}{1-s}\,ds.
\]
A marginal increase in virtual values near rank $u$ raises valuations
at lower ranks $s$ proportionally to $1/(1-s)$. Multiplying each change in
valuation by the quality supplied there yields its contribution
to total surplus. The integral adds these contributions across lower
ranks.

For a bounded perturbation $h$ supported where $\phi$ is bounded away
from zero, and sufficiently small $\varepsilon$ preserving the pointwise
bounds on the profile, the first-order change in the objective is
\begin{equation*}
 J_\beta(\phi+\varepsilon h)-J_\beta(\phi)
 =\varepsilon\int_0^1 H_\beta[\phi](u)h(u)\,du+o(\varepsilon),
\end{equation*}
where
\begin{equation*}
 H_\beta[\phi](u)
 =A_\phi(u)
 +[\nu(Q_\phi(u))-\phi(u)]q'(\phi(u))
 +(\beta-1)q(\phi(u)).
\end{equation*}
The first term is the surplus effect of the induced changes in types.
The second is the surplus effect of the local quality response.
Since $C'(q(\phi(u)))=\phi(u)$, an additional unit of quality generates
$\nu(Q_\phi(u))-\phi(u)$ units of total surplus, while quality responds
to virtual value at rate $q'(\phi(u))$.

The last term accounts for the change in profit. To understand its
coefficient, write the objective as $\TS+(\beta-1)\Pi$.
Since $\pi'=q$, a marginal increase in virtual value raises profit
at rate $q(\phi(u))$. For a given total surplus, this increase shifts
surplus from consumers to the seller and therefore contributes
$(\beta-1)q(\phi(u))$ to the objective.

A perturbation can also move the endogenous lower support
$M(\phi):=Q_\phi(0)$. In the Appendix, we show that
\[
 DM(\phi)[h]
 =\int_0^1 w_\phi(u)h(u)\,du,
\]
where the \emph{lower-support sensitivity} $w_\phi$ is positive and
nondecreasing. Thus, $w_\phi(u)$ measures how strongly a perturbation
near rank $u$ moves the minimum induced type. A perturbation placed at
a higher rank has a weakly larger effect on the lower support. Under
private values, $w_\phi\equiv1$; any amplification across ranks is due
to feedback through type-dependent cost. Accordingly, $w_\phi$ enters
the constrained first-order condition only through the multiplier on
$M(\phi)\geq\underline\theta$.

\paragraph{Why there is no bunching.}
Suppose $\phi=\gamma$ on an interval $(\ell,r)$, where
$0<\gamma<\bar\nu$. The seller assigns all buyers in this interval
the same quality $q(\gamma)$. We consider a change in composition
that induces the seller to lower quality in the lower part of
the group and raise it in the upper part. The marginal value of
increasing virtual value is strictly higher at higher ranks
within the block:
\begin{align*}
 (1-u)\frac{d}{du}H_\beta[\phi](u)
 &=[1-\nu'(Q_\phi(u))]A_\phi(u)+\nu'(Q_\phi(u))q(\gamma)\notag +\ \nu'(Q_\phi(u))[\nu(Q_\phi(u))-\gamma]q'(\gamma)>0.
\end{align*}
Thus, a unit increase in virtual value is strictly more valuable at a higher rank in the block. If $u_m=(\ell+r)/2$, the perturbation along direction
$h=-\mathbf 1_{(\ell,u_m)}+\mathbf 1_{(u_m,r)}$
has first-order value
\[
 D J_\beta(\phi)[h]
 =\int_{u_m}^{r}H_\beta[\phi](u)\,du
  -\int_{\ell}^{u_m}H_\beta[\phi](u)\,du>0.
\]
To see the inequality, let $I_L=(\ell,u_m)$ and $I_H=(u_m,r)$. We obtain the decomposition

\begin{samepage}
\begin{align*}
 D J_\beta(\phi)[h]
 &={\underbrace{\left[\int_{I_H}A_\phi(u)\,du-\int_{I_L}A_\phi(u)\,du\right]}_{\text{propagation effect}}} + q'(\gamma){\underbrace{\left[\int_{I_H}(\nu(Q_\phi(u))-\gamma)\,du-\int_{I_L}(\nu(Q_\phi(u))-\gamma)\,du\right]}_{\text{local allocation effect}}}\\ &+\quad(\beta-1)q(\gamma){\underbrace{\bigl(|I_H|-|I_L|\bigr)}_{=0}}.
\end{align*}
\end{samepage}
The first term is the surplus gain from the change in composition.
Moving virtual value from the lower half to the upper half weakly
raises every upper-tail integral of the profile. The reconstruction
equation then implies that induced types rise, with the increase being strict on part of the active block. Since net value increases
with type, these changes generate a strictly positive first-order
surplus gain at the original qualities.

The second term is the gain from reallocating quality. At the common
initial quality, additional quality generates more surplus for a
buyer with a higher net value. Shifting quality toward the higher
ranks therefore weakly raises total surplus to first order.

Profit is unchanged to first order. On the flat block,
$\pi'(\gamma)=q(\gamma)$ is constant, so the profit gain from raising
virtual values on the upper half offsets the loss from lowering
them on the equal-sized lower half. The first-order increase in
total surplus therefore accrues to consumers. This is why
differentiation remains valuable even when the objective is to maximize consumer surplus.
Figure~\ref{fig:no-bunching-illustration} illustrates the perturbation.

\begin{figure}[t]
 \centering
 \captionsetup{font=small,skip=4pt}
 \includegraphics[width=.98\textwidth]{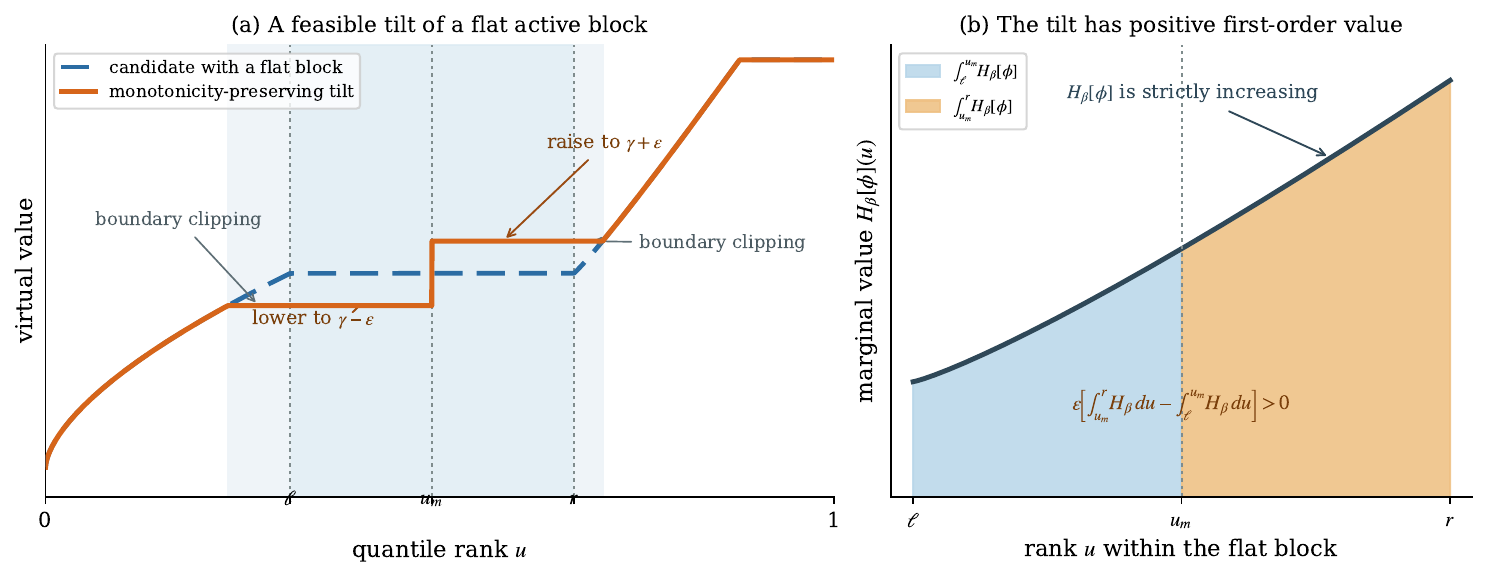}
 \caption{Why bunching is not optimal. Panel (a) starts from a monotone virtual-value profile with a flat block $\phi(u)=\gamma$ on $(\ell,r)$. The solid profile is the perturbation used in the proof. It lowers the profile to $\gamma-\varepsilon$ below the midpoint and raises it to $\gamma+\varepsilon$ above the midpoint, with small boundary adjustments that preserve monotonicity. Panel (b) shows that the marginal value $H_\beta[\phi](u)$ is strictly increasing across the block, so the gain on $(u_m,r)$ strictly exceeds the loss on $(\ell,u_m)$.}
 \label{fig:no-bunching-illustration}
\end{figure}

\paragraph{Why the highest type has positive mass.}
The optimal market must balance two forces. Higher valuations can generate greater information rents, but concentrating the entire market at $\bar\theta$ eliminates private information and allows the seller to extract all surplus. Heterogeneity is therefore needed to sustain rents, while concentration at the top preserves efficient trade.

This tension also appears in the fixed-good monopoly problem of \citet{CondorelliSzentes2020}, where the valuation distribution is chosen before the seller posts a price. Raising valuations increases consumer surplus at a given price, but may induce a higher price. To sustain full trade at price $p$, every higher price $v$ must satisfy $v[1-G(v-)]\leq p$. The consumer-optimal distribution saturates these constraints:
\[
 G(v)=1-\frac{p}{v},\qquad p\leq v<\bar\theta,
 \qquad G(\{\bar\theta\})=\frac{p}{\bar\theta}.
\]
The resulting equal-revenue distribution makes the seller indifferent among prices in the support, with the lowest price implementing full trade. Its lower component sustains information rents, while an atom is placed at the highest valuation.

Screening adds an endogenous allocation margin to this trade-off.
Incentive compatibility gives
$U(\theta)=\int_\ell^\theta x(s)\,ds$,
where $\ell$ is the lowest realized type. Under posted pricing, rents
increase at rate one for purchasing buyers. Here, the rate at which rents
accumulate is the quality chosen by the seller in response to market
composition. The designer must therefore select a distribution that
not only limits surplus extraction but also induces the quality schedule
through which higher types earn rents. This changes the optimal structure.
Under posted pricing, equal revenue makes the seller indifferent among
a range of prices. Under screening, quality is determined by
the ironed virtual net value, $x=q(\phi)$. Nonpositive virtual net values
induce zero quality, while a constant positive virtual net value across
distinct types induces bunching. Our no-exclusion and no-bunching
arguments rule out both. Rather than an equal-revenue lower component,
every nondegenerate optimal market has a fully served, strictly
separating interior, even with type-dependent costs.

Now, consider increasing virtual values on a small upper tail. Quality
changes only on that tail, but the reconstruction equation also raises
the types represented at lower ranks. These type changes increase
total surplus at unchanged qualities. Their marginal contribution
becomes increasingly large as the perturbed ranks approach one.
The private-value expression for $A_\phi$ shows why: valuations at
higher ranks depend on averages over smaller upper tails, so they
respond more to a given increase in the integral of virtual values
above them. Positive quality makes these valuation changes valuable.

The same argument applies with type-dependent costs. Suppose that
$\phi(u)<\bar\nu$ for every $u<1$. On a sufficiently high active tail,
quality is bounded below by some $\underline x>0$. Since
$A_\phi\geq0$ and
$\nu'(Q_\phi(u))\in[\underline r,1]$, the equation for $A_\phi$ gives
\[
 A_\phi'(u)\geq\frac{\underline r\,\underline x}{1-u},
 \qquad
 A_\phi(u)\longrightarrow+\infty
 \quad\text{as }u\uparrow1.
\]
The other terms in $H_\beta[\phi]$, capturing the local quality
response and the change in profit, remain bounded.
Thus, near the top, the surplus gain from the induced
type changes makes an increase in virtual values profitable.
A small upward shift of the tail strictly
improves the objective. It also preserves monotonicity and respects the
minimum-type restriction.

The optimal profile must therefore reach $\bar\nu$ at some rank
$b<1$. Since it is nondecreasing, it remains at $\bar\nu$ on
$[b,1]$. The reconstruction equation then gives
$Q_\phi(u)=\bar\theta$ throughout that interval. Hence a positive
mass $1-b$ of buyers has the highest type and receives first-best
quality $q(\bar\nu)=\bar x$.

The reconstruction equation also makes $Q_\phi$ continuous for $u<1$. Thus, there can be no
gap between types in the market because that would require a jump
in the quantile. Hence the support is connected. To see why there are no atoms below the top, recall that $Q_\phi(u)$
is the upper-tail average of $\phi+m(Q_\phi)$. The function being
averaged is nondecreasing, strictly below $\bar\theta$ at every rank
below $b$, and equal to $\bar\theta$ on $[b,1]$. Its upper-tail
average therefore strictly exceeds its value at each $u<b$.
Differentiating the reconstruction equation gives
\[
 (1-u)Q_\phi'(u)
 =\nu(Q_\phi(u))-\phi(u)>0
 \qquad\text{for almost every }u<b.
\]
Thus, $Q_\phi$ is strictly increasing below $b$ so no lower type
occupies an interval of ranks with positive mass. The only atom
is at $\bar\theta$. Buyers in this atom have the same type, so
purchasing the same quality does not contradict strict
separation of distinct types.

In a nondegenerate optimum, these buyers still earn information
rents. The seller must leave them enough surplus to prevent them
from choosing contracts intended for lower types. The next
subsection characterizes the separating schedule and the size of
the top atom, and establishes uniqueness under the
curvature condition.

\subsection{Characterization}

Fix $\beta\in[0,1)$ and write $\alpha=1-\beta$. On the separated region, let $y$ denote actual net value and $x$ quality. In the Appendix we show that the lower-support constraint has a multiplier $\xi\geq0$ and that the constrained first-order condition is
\begin{equation*}
 H_\beta[\phi](u)+\xi w_\phi(u)=0
 \quad\text{for a.e. rank in the separated interior},
 \qquad
 \xi\bigl(M(\phi)-\underline\theta\bigr)=0.
\end{equation*}
When the minimum-type constraint is slack, $\xi=0$ and the marginal
welfare effect vanishes at each interior rank. When it binds, an
increase in virtual values in one part of the profile raises the
minimum induced type. This creates room for reductions elsewhere
that keep the minimum fixed. The term $\xi w_\phi(u)$ accounts for
the value of this.

Let $B(u)=A_\phi(u)+\xi w_\phi(u)$. It combines the marginal surplus
effect of the induced changes in types with the constraint term.
Using $\phi=C'(x)$ and $q'(\phi)=1/C''(x)$, the first-order condition
becomes
\[
 B+\frac{y-C'(x)}{C''(x)}=\alpha x.
\]
This equation balances the effects of increasing virtual value.
The local quality response generates
$[y-C'(x)]/C''(x)$ additional total surplus, while $B$ accounts for
the induced changes in types and the minimum-type restriction.
Profit rises at rate $x$. Since the objective is $\TS-\alpha\Pi$,
this increase contributes $-\alpha x$. Equivalently, quality satisfies the state--control relation
$y=C'(x)+C''(x)(\alpha x-B)$.
Under the curvature condition above, a larger $B$ at a given net
value $y$ implies higher quality and a smaller gap between actual
and virtual net value.

Let $b$ be the first rank in the top atom. Smooth pasting and the
state--control relation give
\[
 y(b-)=\phi(b-)=\bar\nu,
 \qquad
 x(b-)=\bar x,
 \qquad
 B(b-)=\alpha\bar x.
\]
At first-best quality, $y=C'(x)$, so the marginal surplus gain from
changing quality is zero. The first-order condition then requires
$B=\alpha\bar x$: the composition and constraint terms balance the
profit term. These conditions describe where the separating region
joins the efficiently served top atom.

The role of Assumption~\ref{ass:main-curvature} can be seen by
differentiating the right-hand side of the state--control relation
with respect to quality:
\begin{equation*}
 D_\alpha(x,B)
 :=(1+\alpha)C''(x)+(\alpha x-B)C'''(x)>0
 \qquad\text{whenever }0\leq B\leq\alpha x.
\end{equation*}
If $C'''(x)\geq0$, the inequality is immediate. If $C'''(x)<0$, then $0\leq\alpha x-B\leq\alpha x$ and Assumption~\ref{ass:main-curvature} give
\[
 \frac{D_\alpha(x,B)}{C''(x)}
 \geq 1+\alpha-\alpha\left(-\frac{xC'''(x)}{C''(x)}\right)
 >1-\alpha\geq0.
\]
Hence each state $(B,y)$ determines a unique quality. Write that quality as
$x=X_\alpha(B,y)$.
The first-order condition and the nonlinear tail equation then reduce the separating branch to the two-state system
\begin{align}
 \dot B
 &=[1-\rho(y)]B+\rho(y)X_\alpha(B,y),
 \label{eq:main-two-state-B}\\
 \dot y
 &=\rho(y)\big[y-C'(X_\alpha(B,y))\big],
 \notag
\end{align}
where $t=-\log(1-u)$ is a transformation of rank and
$\rho(y):=\nu'(\nu^{-1}(y))$ is the rate at which net value rises
with type. The first equation combines the equations for $A_\phi$
and $w_\phi$. The second is the reconstruction equation expressed
in net values: below the top atom, the positive gap between actual
and virtual net value makes $y$ rise with rank.

Once $(B,y)$ is known, quality is $X_\alpha(B,y)$, virtual value
is $C'(X_\alpha(B,y))$, gross type is $\nu^{-1}(y)$, and rank is
$u=1-e^{-t}$. Differentiating the state--control relation shows
that quality also rises along the separating branch. Thus,
quality, net value, and $B$ all increase as we move toward
the top atom.

Let $T=-\log(1-b)$ be the transformed rank at which the top atom
begins. The terminal state is
$(B(T),y(T))=(\alpha\bar x,\bar\nu)$, and the atom has mass $e^{-T}$.
At the lower endpoint,
\[
 B(0)\geq0,\qquad
 y(0)\geq\nu(\underline\theta),\qquad
 B(0)[y(0)-\nu(\underline\theta)]=0.
\]
Since $A_\phi(0)=0$ and $w_\phi(0)=1$, we have $B(0)=\xi$.
These conditions therefore express the minimum-type restriction
and complementary slackness.

The optimal market is obtained by following the unique path backward
from $(\alpha\bar x,\bar\nu)$ until either $B=0$ or
$y=\nu(\underline\theta)$. If $B=0$ is reached first, the lower
endpoint is freely chosen and the minimum-type restriction has
zero shadow cost. If the minimum feasible type is reached first, this
determines the lowest selected type and stops the path with $B>0$.
Both boundaries are reached together when the minimum feasible type
coincides with the freely selected endpoint. In this case, the constraint binds with a zero multiplier.

\subsection{The Pareto frontier}

Under Assumption~\ref{ass:main-curvature}, the support problems exhaust the Pareto frontier. This is despite the fact that weighted sums need not recover every efficient payoff when the attainable set is nonconvex, and Section~\ref{sec:geometry} shows that nonconvexity can indeed arise here. Nevertheless, as \(\beta\) varies from zero to one, the unique selected payoff moves continuously from the consumer surplus maximum to the degenerate top market, which maximizes both total surplus and seller profit. Thus, on the efficient frontier, the infinite-dimensional choice of a type distribution reduces to a one-parameter family of markets that uniquely implement it. Therefore, whenever an objective over consumer surplus and seller profit selects a Pareto-efficient outcome, the market that implements it belongs to this family.

At points where the frontier is differentiable and profit varies with $\beta$, its slope satisfies \(d\CS/dP=-\beta\). Thus, \(\beta\)  can be interpreted as both an institution's relative weight on seller profit and as the marginal loss of consumer surplus required to increase profit along the frontier. When the constraint on the minimum feasible type binds, different weights can select the same market. That is, the map from the weights to markets need not be one-to-one. Nevertheless, any given efficient outcome is never implemented by more than one market.

\section{Geometry of the Surplus Set}
\label{sec:geometry}

The preceding analysis characterizes the Pareto frontier. We now
study the rest of the surplus set and ask whether every convex
combination of attainable payoffs can itself be generated by a
single market. This determines whether allowing separate menus
across observable market segments expands the set of outcomes.

Under the stronger curvature condition introduced below, we first
identify the maximum consumer surplus at every feasible profit.
We then characterize the entire surplus set when a payoff-neutral
type is feasible or costs are type independent. In both cases,
the set is convex. With type-dependent costs and no payoff-neutral
type, nonconvexity can arise. Segmentation can then add payoff
pairs, but it cannot increase the maximum consumer surplus at any
given profit.

\subsection{Completing the consumer surplus frontier}
\label{sec:lower-boundary}

The Pareto frontier covers profits from the consumer-optimal
market to the profit-maximizing market. Every outcome at a lower
profit is Pareto dominated by the consumer-optimal market, which
gives the seller more profit and buyers at least as much surplus.
To describe the full surplus set, however, we also need the
greatest consumer surplus attainable at each of these lower
profits. We obtain this by allowing negative values of $\beta$ in the
objective.

\begin{assumption}[Strong curvature]
\label{ass:lf-strong-curvature}
The function $C$ is three times continuously differentiable on $(0,\bar x]$ and
\begin{equation*}
 -\frac{xC'''(x)}{C''(x)}<1,
 \qquad x\in(0,\bar x].
\end{equation*}
\end{assumption}

The stronger condition is used only because $\beta$ can become arbitrarily small. To see the connection with the Pareto problem, put $\alpha=1-\beta\geq1$ for a nonpositive profit coefficient. The same objective, first variation, and state--control relation continue to apply. What must now hold uniformly over arbitrarily large $\alpha$ is
\begin{equation*}
 D_\alpha(x,B)
 =(1+\alpha)C''(x)+(\alpha x-B)C'''(x)>0
 \quad\text{for}\quad \alpha\geq1 \ \text{and} \ \ 0\leq B\leq\alpha x.
\end{equation*}
If $C'''(x)\geq0$, positivity is immediate. If $C'''(x)<0$, then
\[
 \frac{D_\alpha(x,B)}{C''(x)}
 \geq 1+\alpha-\alpha\left(-\frac{xC'''(x)}{C''(x)}\right)>1,
 \qquad \alpha\geq1,
\]
so Assumption~\ref{ass:lf-strong-curvature} guarantees this inequality uniformly over all nonpositive profit weights. The no-exclusion, no-bunching, and top-atom perturbations are otherwise the same as in the case of markets that generate outcomes on the Pareto frontier.

Let
\(P^{\min}_{\underline\theta}:=\min_G\Pi(G)\) and
\(P^{\max}:=\pi(\bar\nu)\). If
\(\underline\theta=0\), define
$\underline\beta
 :=\inf\{\beta\leq0:U_0(\beta)>0\}\in[-\infty,0)$.

\begin{proposition}[The maximum-consumer surplus frontier]
\label{prop:main-right-frontier}
Maintain Assumptions~\ref{ass:main-nu},~\ref{ass:main-C}, and
\ref{ass:lf-strong-curvature}. Every finite nonpositive-weight problem
has a unique optimal market. Whenever its payoff pair is nonzero,
that market is regular and fully active, has connected support,
strictly separates below an efficiently served highest-type atom,
and has no other atoms. Together with the Pareto path, the nonpositive-weight payoff path and
its minimum-profit endpoint exhaust
\begin{equation*}
\{(\Gamma_{\underline\theta}(p),p):
       P^{\min}_{\underline\theta}\leq p\leq P^{\max}\}, \qquad \Gamma_{\underline\theta}(p):=\max_{\Pi(G)=p}\CS(G).
\end{equation*}
Every point on this boundary with
\(p>P^{\min}_{\underline\theta}\) is implemented by a unique market.
Moreover:
\begin{enumerate}[label=(\roman*)]
 \item If \(\underline\theta>0\), the optimal market payoffs converge as
 \(\beta\to-\infty\) to
 $\bigl(\Gamma_{\underline\theta}
           (P^{\min}_{\underline\theta}),
        P^{\min}_{\underline\theta}\bigr)$,
 and every weak accumulation point of the optimal markets implements
 that payoff. Under private values, the minimum-profit endpoint is
 also uniquely implemented, by \(\delta_{\underline\theta}\).

 \item If \(\underline\theta=0\), the optimizer has nonzero payoff pair
 for \(\underline\beta<\beta\leq0\). At every finite
 \(\beta\leq\underline\beta\), the unique optimizer is \(\delta_0\).
 As \(\beta\downarrow\underline\beta\), interpreted as
 \(\beta\to-\infty\) when \(\underline\beta=-\infty\), the optimal
 markets converge weakly to \(\delta_0\). Thus, every point on the
 maximum-consumer surplus boundary is uniquely implemented.
\end{enumerate}
\end{proposition}
As $\beta$ increases toward zero, both payoffs rise whenever the
optimal market changes. Along this part of the boundary, allowing
more profit permits higher consumer surplus. Reducing
profit below its level in the consumer-optimal market requires
a sacrifice of consumer surplus as well.

If a payoff-neutral type is feasible, the boundary ends at
$(0,0)$, implemented by $\delta_0$. This market may become optimal
at a finite penalty on profit or be approached only as the penalty
tends to infinity. If every feasible type has positive gains from
trade, minimum profit is strictly positive. As $\beta\to-\infty$,
the selected payoff approaches the minimum attainable profit
together with the greatest consumer surplus compatible with
that profit.

The characterization now identifies the best outcome for consumers
at every feasible profit which we need for the
geometry and segmentation results. Whether segmentation can increase consumer
surplus at a given profit matters also at profits below that of
the consumer-optimal market.

\subsection{The geometry of surplus}

The preceding proposition identifies the greatest consumer surplus
attainable at every feasible profit. To describe the rest of the
set, we ask whether every convex combination of attainable payoff
pairs can also be implemented by one market. Mixtures of
distributions do not automatically achieve this because the seller
optimizes its menu for the mixture, so the resulting
payoffs need not equal the corresponding average of the original
payoffs.

\begin{theorem}[The surplus set]
\label{thm:main-geometry}
\begin{enumerate}[label=(\roman*)]
 \item \emph{A feasible payoff-neutral type.} Normalize the highest first-best no-trade type to zero. Under Assumptions~\ref{ass:main-nu}--\ref{ass:main-C}, adding mass $1-\lambda$ at that type to any market scales both payoff coordinates by $\lambda$. Under Assumption~\ref{ass:lf-strong-curvature},
 \begin{equation}
  \mathcal V_0=\{\lambda z:\lambda\in[0,1],\ z\in\mathcal F^P(\mathcal V_0)\cup\mathcal F^L(\mathcal V_0)\}.
  \label{eq:main-radial-set}
 \end{equation}
 Hence the single-market set is compact and convex and every payoff has a single-market implementation.
 \item \emph{Private values without a payoff-neutral type.} Under private values and the strong curvature condition, the single-market payoff set remains compact and convex even when no payoff-neutral type is admissible. At every fixed profit, each consumer surplus level between zero
and the maximum is attained by a single market.
 \item \emph{Type-dependent costs without a payoff-neutral type.} There exist economies satisfying the maintained assumptions for which two markets minimize the same positive weighted sum of consumer surplus and profit at different profit levels, while no single market generates the interior of the line segment between their payoffs. Thus, the single-market payoff set may be nonconvex.
\end{enumerate}
\end{theorem}

\subsubsection{A payoff-neutral type as a scaling instrument and convexity of the surplus set}
\label{subsec:general-radial-scaling}

A payoff-neutral type can be assigned the outside option and cannot contribute positively to the seller's payoff.  After the normalization, that type is zero.  It therefore separates the composition of the active market from the scale at which that market is used. Formally, for any active market $G$ and $\lambda\in[0,1]$, define $G^\lambda=(1-\lambda)\delta_0+\lambda G$.
Extending an optimal mechanism for $G$ with the outside option at the normalized payoff-neutral type earns $\lambda\Pi(G)$. Conversely, individual rationality prevents that type from contributing positive profit, and restricting any mechanism for the diluted market to the active component cannot earn more than $\Pi(G)$. Hence
\begin{equation*}
 (\CS(G^\lambda),\Pi(G^\lambda))
 =\lambda(\CS(G),\Pi(G)).
\end{equation*}
Dilution implements every point on the line segment joining an
attainable payoff to the origin. To see why these segments fill
the convex hull of $\mathcal V_0$, take any nonzero point in the
hull and extend the ray from the origin through it to the boundary. If the endpoint has positive consumer surplus,
it maximizes $\CS+\beta\Pi$ for some $\beta$. Every market in a
convex combination attaining this bound must maximize the same
objective. Under the strong curvature condition, that maximizing
market is unique, so the endpoint is itself the payoff of a
single market. If the endpoint has zero consumer surplus, it is
the maximum-profit payoff, which also has a single-market
implementation. Diluting the market at the endpoint then
implements the original point. Thus every point in the convex
hull is attainable by a single market, establishing convexity.
Figure~\ref{fig:mr-full-payoff-set} displays the surplus set in the quadratic-cost Mussa-Rosen monopoly screening model.

\begin{figure}[t]
  \centering
  \captionsetup{font=small,skip=4pt}
  \includegraphics[width=\linewidth]{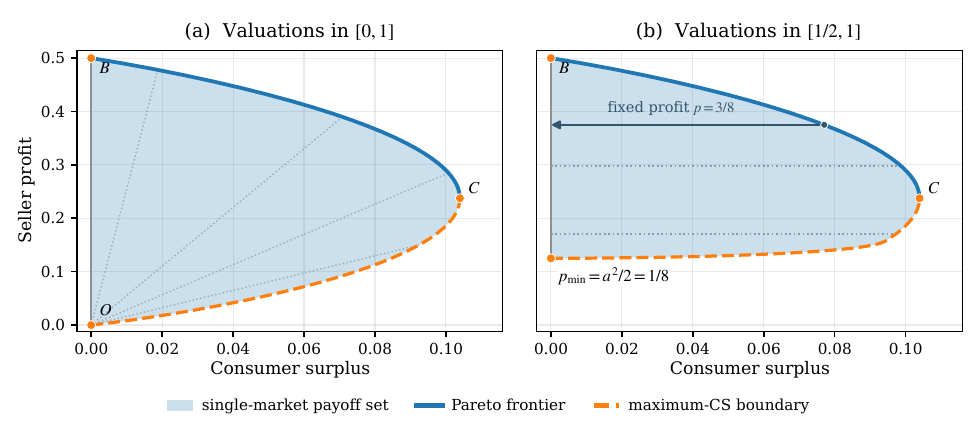}
  \caption{Single-market surplus sets in the Mussa--Rosen
  specification with $C(x)=x^2/2$. The solid curve traces the Pareto frontier
  and the dashed curve continues the maximum consumer surplus boundary
  to the minimum-profit endpoint. In panel~(a), adding zero-valuation
  buyers scales payoffs along a ray to the origin. In panel~(b),
  the lower valuation is $a=1/2$ and minimum profit is $a^2/2=1/8$.
  The arrow illustrates the adjustment at $p=3/8$. It connects the
  maximum consumer surplus at that profit to zero consumer surplus
  in the market concentrated at $\sqrt{2p}=\sqrt{3}/2$. Every point
  on the segment is attainable by a single market.}
  \label{fig:mr-full-payoff-set}
\end{figure}

\subsubsection{Convexity under private values without a payoff-neutral type}

When no payoff-neutral type is available, private values still
allow us to reduce consumer surplus while holding profit fixed.
The key is that every attainable profit can also be earned in
a homogeneous market, where buyers receive no surplus.

Normalize valuations to $[\underline v,1]$, with $\underline v>0$.
In any market, the seller can guarantee $\pi(\underline v)$ by
offering quality $q(\underline v)$ at price
$\underline v q(\underline v)$, and profit cannot exceed $\pi(1)$.
Thus, for every attainable profit $p$, there is a type
$z_p\in[\underline v,1]$ satisfying $\pi(z_p)=p$.
The constant virtual-value profile $z_p$ implements the homogeneous
market at that type: the seller provides the efficient quality
and extracts all surplus.

We connect a market with positive consumer surplus to this
homogeneous market. Let a nonnegative regular profile $\phi$
generate consumer surplus $c>0$ and profit $p$. First, move its
virtual values toward $z_p$:
\begin{equation*}
 \phi_t=(1-t)\phi+t z_p,\qquad t\in[0,1].
\end{equation*}
Convexity of $\pi$ implies that profit along this path is at most
$p$. To restore profit, move each profile toward the constant
profile $1$:
\begin{equation*}
 \psi_t=(1-s(t))\phi_t+s(t)\cdot1,
 \qquad
 \int_0^1\pi(\psi_t(u))\,du=p.
\end{equation*}
Profit increases continuously and strictly with $s$, reaching
$\pi(1)>p$ at $s=1$. Hence there is a unique continuous $s(t)$ that restores profit to $p$.

Both adjustments preserve monotonicity and feasibility.
Under private values, the reconstruction equation is linear:
\[Q_{\phi_t}=(1-t)Q_\phi+t z_p \qquad\text{and}\qquad
Q_{\psi_t}=(1-s(t))Q_{\phi_t}+s(t)\cdot1.\]
The first adjustment therefore moves each type quantile toward
$z_p$, and the second toward $1$. All types remain within the
feasible interval, and each resulting profile induces a single
market.

Since both endpoint profiles already earn $p$, we have
$s(0)=s(1)=0$. The path therefore begins at $\phi$ and ends at
the constant profile $z_p$. Consumer surplus varies continuously
from $c$ to zero while profit remains $p$, so every intermediate
consumer surplus is attained.

Proposition \ref{prop:main-right-frontier} shows that a single market
attains the maximum consumer surplus at each profit in the convex hull
of the set of attainable payoffs. Applying this construction to those markets
implements every smaller nonnegative consumer surplus at the same
profit. Every point in the convex hull therefore has a single-market
implementation, establishing convexity.

\subsection{Why type-dependent costs can break convexity}
\label{sec:nonconvexity}

The private-value construction relies on the fact that every
attainable profit can also be earned in a homogeneous market,
where buyers receive no surplus. Type-dependent costs can violate
this property. Under private values, offering the efficient quality corresponding to the lowest feasible
type at that type's willingness to pay
guarantees profit equal to that of the lowest-type homogeneous market.
With type-dependent costs, higher types will accept the same offer
but may cost more to serve at the same price, so this guarantee
can fail. In the example below, a heterogeneous market earns less
than every feasible homogeneous market, and the surplus set is
nonconvex.

\begin{example}[A nonconvex single-market surplus set]
\label{ex:main-nonconvexity}
Let the feasible type interval be \([1,2]\), and set
\begin{equation*}
 m(\theta)=\frac35\theta+\frac{\theta^2}{80},
 \qquad
 C(x)=\frac23x^{3/2},
 \qquad 0\leq x\leq\frac9{16}.
\end{equation*}
The type-dependent component of the cost is defined by the same formula on \([0,2]\).
Both cost components are increasing and strictly convex. Higher types
are more expensive to serve but still generate greater gains from trade. These primitives satisfy
Assumptions~\ref{ass:main-nu},~\ref{ass:main-C}, and
\ref{ass:lf-strong-curvature}.

There are two competing ways to reduce a weighted sum of consumer
surplus and profit. The homogeneous market \(G^A=\delta_1\) eliminates
information rents: the seller serves the lowest feasible type efficiently
and extracts all surplus. A second market \(G^B\) has a positive density
throughout \((1,2)\) and an atom at \(2\). It induces a strictly separating
menu below the efficiently served top atom. This heterogeneous market
leaves positive information rents but generates less profit than \(G^A\).

For a coefficient \(\lambda_*\in(1/2,1)\), these are the only two
markets minimizing \(\CS+\lambda_*\Pi\). Specifically, every market
satisfies
\begin{equation}
 \CS(G)+\lambda_*\Pi(G)\geq\lambda_*\pi(\nu(1)),
 \qquad G\in\Delta([1,2]),
 \label{eq:main-nonconvex-support-equality}
\end{equation}
with equality if and only if \(G=G^A\) or \(G=G^B\).
The Appendix constructs \(G^B\) and proves the bound for arbitrary
distributions.\footnote{
Numerically, \(\lambda_*\approx0.8140\). The two
\((\CS,\Pi)\) pairs are approximately
\((0,0.0193952)\) and \((0.0031817,0.0154866)\).}
Every interior point of the line segment joining the two payoff
pairs also satisfies
\eqref{eq:main-nonconvex-support-equality} with equality.
A market implementing such a point would therefore be another
minimizer of $\CS+\lambda_*\Pi$. Since $G^A$ and $G^B$ are the
only minimizers, no such market exists. The interior of the
segment is missing from $\mathcal V_1$, proving nonconvexity.
\end{example}

The private value construction cannot
be used here because any market earning less than $\pi(\nu(1))$ must
leave buyers positive surplus. In particular, at profit
$\Pi(G^B)$, no change in composition can eliminate rents while
preserving profit. At any profit strictly between $\Pi(G^B)$ and $\Pi(G^A)$, a
single market must leave buyers strictly more surplus than the
corresponding convex combination of the two endpoint payoffs.
Segmentation, as we will see in the next subsection, attains that combination by keeping $G^A$ and $G^B$
as separate segments, for each of which the seller chooses a separate menu.

\begin{figure}[t]
 \centering
 \captionsetup{font=small,skip=4pt}
 \includegraphics[width=.7\textwidth]{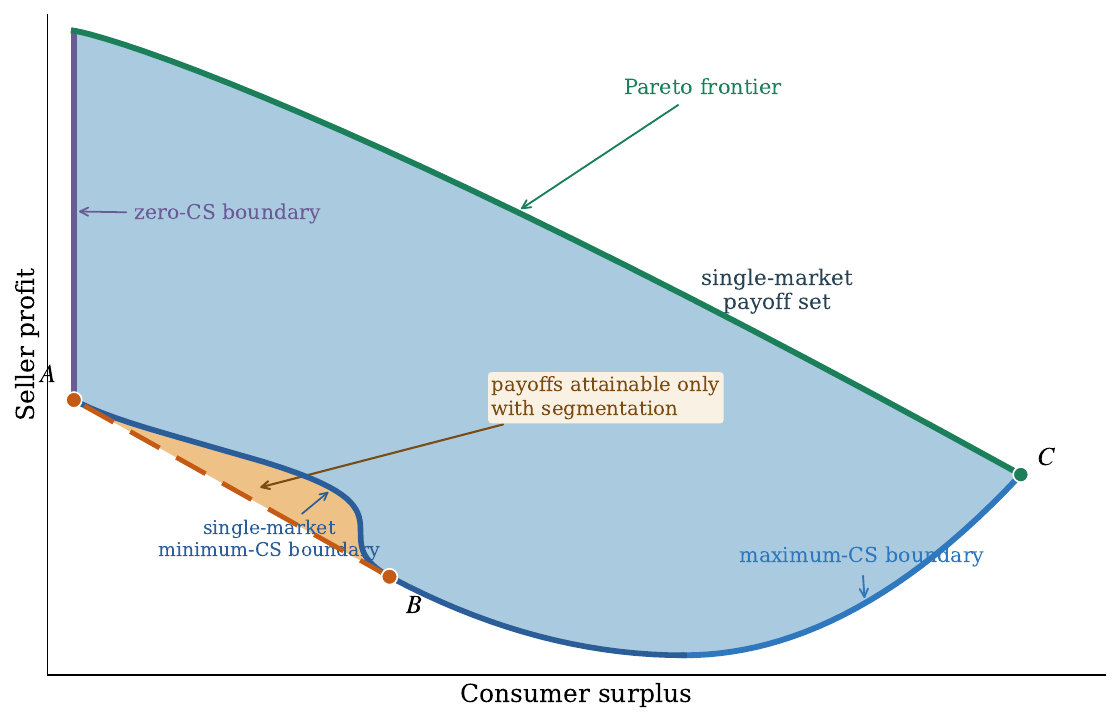}
 \caption{A nonconvex surplus set. Points \(A\) and \(B\) are the
 two equality payoffs in \eqref{eq:main-nonconvex-support-equality}.
 Their dashed chord lies on a supporting line for
 \(\CS+\lambda_*\Pi\). Its endpoints are attainable by single markets,
 but every interior point requires segmentation. The shaded gap
 illustrates the distinction between the single-market set and its
 convex hull.}
 \label{fig:nonconvex-surplus-set}
\end{figure}

\subsection{Market Segmentation}
\label{sec:segmentation}

We now allow choosing both the composition of
the market and its division into observable segments. That is, we choose the type distribution and the share of buyers in each segment,
thereby determining the aggregate distribution as well. The seller
observes the segment label and offers a separate profit-maximizing
menu in each segment. We ask whether this expands the set of
attainable payoffs and, in particular, whether it raises maximum
consumer surplus at a given profit.

\begin{definition}[Observable segmentation]
\label{def:observable-segmentation}
An observable segmentation is a finite collection
\[
 \{(\lambda_j,G_j)\}_{j=1}^J,
 \qquad
 \lambda_j\geq0,
 \qquad
 \sum_{j=1}^J\lambda_j=1,
 \qquad
 G_j\in\Delta(\Theta_{\underline\theta}),
\]
where the seller observes the segment label $j$ and offers a separate profit-maximizing menu in each segment. The induced aggregate payoff is
\[
 \sum_{j=1}^J\lambda_j\bigl(\CS(G_j),\Pi(G_j)\bigr).
\]
The set of payoffs attainable by observable segmentation is therefore
\begin{equation}
 \mathcal K_{\underline\theta}
 :=\left\{
 \sum_{j=1}^J\lambda_j\bigl(\CS(G_j),\Pi(G_j)\bigr):
 J\in\mathbb N,\ \lambda_j\geq0,\ \sum_{j=1}^J\lambda_j=1,\
 G_j\in\Delta(\Theta_{\underline\theta})
 \right\}.
 \label{eq:main-segmentation-definition}
\end{equation}
\end{definition}

By construction, aggregate payoffs are finite convex combinations of single-market payoffs. The next theorem shows that this convexification leaves the maximum consumer surplus boundary unchanged.

\begin{theorem}[Segmentation does not enlarge the consumer surplus frontier]
\label{thm:main-segmentation}
Maintain Assumptions~\ref{ass:main-nu}--\ref{ass:main-C}. The segmented payoff set defined in \eqref{eq:main-segmentation-definition} satisfies
$\mathcal K_{\underline\theta}=\co(\mathcal V_{\underline\theta})$.
Under Assumption~\ref{ass:lf-strong-curvature}, for every feasible profit $p$,
\begin{equation*}
 \max\{c:(c,p)\in\mathcal K_{\underline\theta}\}
 =\max\{c:(c,p)\in\mathcal V_{\underline\theta}\}
 =\Gamma_{\underline\theta}(p).
\end{equation*}
Consequently, every continuous objective that is weakly increasing in consumer surplus at each fixed profit has a single-market optimizer. Segmentation does not improve the maximum consumer surplus frontier.
\end{theorem}
The bounds for single markets also constrain what segmentation can achieve. Fix a profit $p$ strictly between the minimum and
maximum attainable profits. The boundary characterization provides
a coefficient $\beta$ such that every market $G$ satisfies
$\CS(G)+\beta\Pi(G)
 \leq \Gamma_{\underline\theta}(p)+\beta p$.
Each segment obeys this inequality. Averaging across segments
whose aggregate profit is $p$ therefore bounds aggregate consumer
surplus by $\Gamma_{\underline\theta}(p)$, which a single market
already attains. The characterization with negative weight on profits supplies these bounds
at profits below that of the consumer-optimal market. At either
profit endpoint, every segment with positive weight must itself
earn that profit, so averaging cannot improve the
single-market consumer surplus maximum.

\section{Comparative Statics: Private Values and Beyond}
\label{sec:comparative-statics}

We first study comparative statics as we move along the Pareto frontier of the surplus set. Initially, we state the results that hold for the general model and then derive sharper comparative statics with private values. We next vary the lowest feasible type at a fixed profit weight and conclude the section by comparing environments in which willingness to pay and cost co-move.

\subsection{Changing the weight on profit}

For the degenerate top market, set $\ell_{1,\underline\theta}=\bar\theta$.

\begin{proposition}[Comparative statics in the profit weight]
\label{prop:main-beta-comparative}
Maintain Assumptions~\ref{ass:main-nu}--\ref{ass:main-curvature}. Let $0\leq\beta_2<\beta_1\leq1$. Then
$\ell_{\beta_2,\underline\theta}\leq\ell_{\beta_1,\underline\theta}$.
If $\beta_1<1$, equality holds if and only if the feasible minimum already binds at $\beta_1$. Consumer surplus is weakly higher under $\beta_2$, while profit and total surplus are weakly lower. All three inequalities are strict when the optimal markets differ. If one market is optimal for two distinct weights, its minimum feasible type constraint binds.
\end{proposition}

The lower-support comparison is generated by the free-boundary system. Put $\alpha_i=1-\beta_i$, so $\alpha_2>\alpha_1$, and let $x=X_\alpha(B,y)$ be the quality defined by the state--control relation. At a fixed state,
\begin{equation}
 \frac{\partial X_\alpha(B,y)}{\partial\alpha}
 =-\frac{C''(X_\alpha)X_\alpha}
 {(1+\alpha)C''(X_\alpha)+(\alpha X_\alpha-B)C'''(X_\alpha)}<0.
 \label{eq:main-beta-X-comparison}
\end{equation}
Using net value $y$ as the independent variable, the propagation state satisfies
\begin{equation*}
 \frac{dB_\alpha}{dy}
 =\mathcal R_\alpha(B_\alpha,y)
 :=\frac{[1-\rho(y)]B_\alpha+\rho(y)X_\alpha(B_\alpha,y)}
 {\rho(y)\{y-C'(X_\alpha(B_\alpha,y))\}}.
\end{equation*}
Equation~\eqref{eq:main-beta-X-comparison} implies that $\mathcal R_\alpha$ is strictly decreasing in $\alpha$ at a fixed state, while the terminal condition is $B_\alpha(\bar\nu)=\alpha\bar x$. A first crossing argument therefore gives $B_{\alpha_2}(y)>B_{\alpha_1}(y)$
on their common nonterminal domain. The more consumer-oriented path reaches $B=0$ at a lower net value and hence selects a weakly lower minimum support point.

Writing $(C_i,P_i)$ for the payoffs selected at $\beta_i$, by respective optimality we have
$C_1+\beta_1P_1\geq C_2+\beta_1P_2$ and $C_2+\beta_2P_2\geq C_1+\beta_2P_1$.
Adding them yields $(\beta_1-\beta_2)(P_1-P_2)\geq0$. The individual inequalities then give the consumer surplus comparison. Finally, since
$\CS+\beta\Pi=\TS-(1-\beta)\Pi$,
optimality at $\beta_1$ implies
$\TS_1-\TS_2
 \geq(1-\beta_1)(P_1-P_2)\geq0$.
Thus, a lower profit weight raises expected consumer surplus and lowers
profit and total surplus whenever the market changes.

\subsubsection{Private values}
\label{subsec:mr-comparative-statics}

We now specialize the environment to private values to obtain sharper comparative statics results. Then we identify separate restrictions under which these results extend to type-dependent costs.

For this subsection only, we denote the private type by $v$, and we normalize (without loss) the highest type to take the value one. The lowest feasible type is then $\underline v=\underline\theta/\bar\theta.$ For each welfare weight \(\beta\in[0,1)\) and lower admissible valuation \(\underline v\in[0,1)\), let
$G_{\beta,\underline v}$, $Q_{\beta,\underline v}$, $\phi_{\beta,\underline v}$ and $x_{\beta,\underline v}=q(\phi_{\beta,\underline v})$ denote the optimal valuation distribution, quantile, virtual-value profile, and rank-space quality schedule. Let $\ell_{\beta,\underline v}:=Q_{\beta,\underline v}(0)$
be the realized lower support, and let \(b_{\beta,\underline v}\) be the first rank in the atom at valuation one, so that the top atom has mass \(1-b_{\beta,\underline v}\). Write $\ell_\beta^0:=\ell_{\beta,0}$
for the endogenous lower support when the normalized payoff-neutral threshold is feasible. On a common continuous valuation support, let \(h_{\beta,\underline v}\) be the hazard rate, \(\psi_{\beta,\underline v}(v)\) the virtual value of valuation \(v\), and
$x_{\beta,\underline v}^V(v):=q(\psi_{\beta,\underline v}(v))$ the quality assigned to that valuation. Finally, to economize notation for $i=1,2$, write
\[
 G_i:=G_{\beta_i,\underline v},\quad Q_i:=Q_{\beta_i,\underline v},\quad
 \phi_i:=\phi_{\beta_i,\underline v},\quad x_i:=x_{\beta_i,\underline v},\quad
 \ell_i:=\ell_{\beta_i,\underline v},\quad b_i:=b_{\beta_i,\underline v},
\]
and put $\ell_i^0:=\ell_{\beta_i}^0$.

\begin{proposition}[Private-value comparative statics]
\label{prop:mr-comparative-statics}
Maintain Assumptions~\ref{ass:main-C} and~\ref{ass:main-curvature}, specialize to private values, fix a lower admissible valuation $\underline v\in[0,1)$, and let $0\leq\beta_2<\beta_1<1$. Then:
\begin{enumerate}[label=(\roman*)]
 \item The realized lower support is
 $\ell_i=\max\{\underline v,\ell_i^0\}$.
 Hence $\ell_2\leq\ell_1$, with equality if and only if $\underline v\geq\ell_1^0$.

 \item The separating region is longer and the highest-valuation atom is smaller under the more consumer-oriented objective:
 $b_2>b_1$.

 \item Valuations are weakly lower at every rank, $Q_2(u)\leq Q_1(u)$ for $u\in[0,1]$, with strict inequality for every $u\in(0,b_2)$.  At rank zero, strictness holds if and only if $\ell_2<\ell_1$.  Consequently, $G_2$ is first-order stochastically dominated by $G_1$.

 \item On the common continuous support $(\ell_1,1)$,
 $h_2(v)>h_1(v)$, $\psi_2(v)>\psi_1(v)$ and $x_2^V(v)>x_1^V(v)$.
 Both markets serve the atom at the highest valuation efficiently.

 \item Consumer surplus is strictly higher under $\beta_2$, while profit and total surplus are strictly lower.

 \item The rank space comparison of virtual values and quality depends on whether the lower admissible valuation binds:
 \begin{enumerate}[label=(\alph*)]
  \item If $\underline v\leq\ell_2^0$, the floor does not alter either payoff-neutral benchmark optimizer: $\phi_2(u)<\phi_1(u)$ and $x_2(u)<x_1(u)$ for every $u\in[0,b_2)$.
  In particular, this case applies whenever a first-best no-trade type is feasible and the normalization is used, so that $\underline v=0$.

  \item If $\underline v>\ell_1^0$, both lower-support constraints bind strictly.  At the bottom,
  \[
   \lim_{u\downarrow0}x_2(u)>\lim_{u\downarrow0}x_1(u),
   \qquad
   \lim_{u\downarrow0}\phi_2(u)>\lim_{u\downarrow0}\phi_1(u),
  \]
  whereas, for every $u\in(b_1,b_2)$,
  $x_2(u)<x_1(u)=\bar x$ and $\phi_2(u)<\phi_1(u)=1$.
  Hence the rank-space quality and virtual-value schedules cross at least once.

  \item If $\ell_2^0<\underline v\leq\ell_1^0$, only the more consumer-oriented problem has a strictly binding floor.
 \end{enumerate}
\end{enumerate}
\end{proposition}

The proposition separates results
that hold irrespective of the support constraint from the rank-quality
ordering, which depends on whether the constraint on the minimum feasible type binds. Irrespective of whether a payoff-neutral type is feasible, greater consumer orientation lowers the optimal distribution in the sense of first-order stochastic dominance, lengthens the separating region, shrinks the top atom, treats a fixed valuation more generously, raises consumer surplus, and lowers profit and total surplus. The feasibility of a payoff-neutral type matters only for the ordering of virtual values and quality by rank. Figure \ref{fig:mr-comparative-statics} illustrates.

The proof first compares the paths when a payoff-neutral type is feasible.  In this case, the private-value free-boundary system has no support multiplier: $B=A$, where $A$ is the rent-propagation state, while $\rho\equiv1$ and $y=Q$.  The separating branch satisfies
\begin{align}\label{eq:main-mr-state-control}
 Q=C'(x)+C''(x)(\alpha x-A), \quad
 \dot A=x, \quad
 \dot Q=Q-C'(x),
\end{align}
with terminal state $(A,Q,x)=(\alpha\bar x,1,\bar x)$.  Since $A$ is strictly increasing, it can be used as the independent variable.  If $X_\alpha(A,Q)$ denotes the quality solving \eqref{eq:main-mr-state-control}, then
\begin{equation}
 \frac{dQ}{dA}
 =\frac{Q-C'(X_\alpha(A,Q))}{X_\alpha(A,Q)}.
 \label{eq:main-mr-A-parametrization}
\end{equation}
At a fixed state, $X_\alpha$ decreases with $\alpha$, while the right-hand side of \eqref{eq:main-mr-A-parametrization} increases. The two paths have the same terminal valuation and quality but different terminal coordinates. On their common domain, the path corresponding to the more consumer-oriented market is already strictly below the other path at $A=\alpha_1\bar x$, and thus, a crossing can be ruled out. This yields the quantile-quality ordering in the presence of a payoff-neutral type.

Nonfeasibility of a payoff-neutral type changes only where each path stops.  For a fixed welfare weight,
$\ell_{\beta,\underline v}=\max\{\underline v,\ell_\beta^0\}$ and $G_{\beta,\underline v}=G_{\beta,0}(\,\cdot\mid v\geq\underline v\,)$ whenever $\underline v>\ell_\beta^0$.

Conditioning on the upper tail rescales the density and survival
probability by the same factor. Hence, at every retained valuation
in the continuous support, the hazard rate, virtual value, and
assigned quality remain unchanged. The distributional, fixed-valuation,
and top-atom comparisons therefore carry over. What can fail is the
ordering of quality by rank. When both markets have the same binding
lower valuation, the more consumer-oriented market assigns higher
quality at the bottom. At higher ranks, however, this market is still
separating buyers while the less consumer-oriented market has already
reached its efficiently served top atom. The two quality schedules
must therefore cross.

\begin{figure}[t]
 \centering
 \captionsetup{font=small,skip=4pt}
 \includegraphics[width=.80\textwidth]{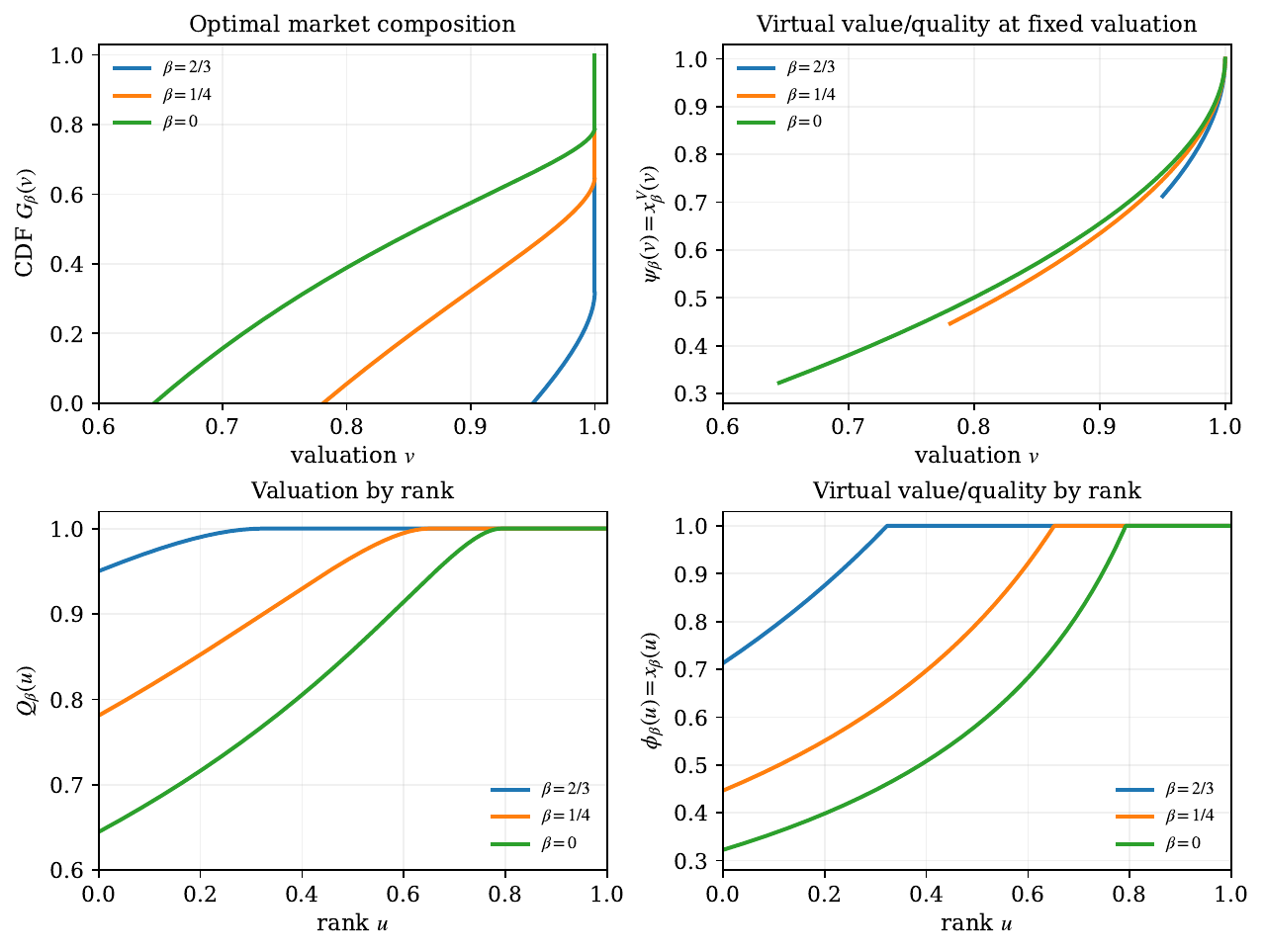}
 \caption{Comparative statics in the private value model with a feasible payoff-neutral type. The specification is \(C(x)=x^2/2\), valuations lie in \([0,1]\), and \(\beta\in\{2/3,1/4,0\}\). }
 \label{fig:mr-comparative-statics}
\end{figure}

Under private values, two identities distinguish comparisons at
a fixed rank from comparisons at a fixed valuation. Actual
valuations are upper-tail averages of virtual values:
\begin{equation*}
 Q_i(u)=\frac{1}{1-u}\int_u^1\phi_i(s)\,ds.
\end{equation*}
When the minimum-type restriction does not alter either
unconstrained optimizer, the more consumer-oriented market has
lower virtual values at each rank below $b_2$. It therefore has
lower valuations and qualities at those ranks as well.

At a fixed valuation, the relevant equation is the inverse
hazard:
\begin{equation*}
 v-\psi_i(v)=\frac{1-G_i(v)}{g_i(v)}=\frac{1}{h_i(v)},
\end{equation*}
where $g_i$ is the density and $h_i=g_i/(1-G_i)$ is the hazard
rate on the continuous support. Raising quality for buyers near
$v$ also raises the information rents of buyers with higher value.
The seller weighs the gain from improving quality for buyers
near $v$, whose density is $g_i(v)$, against the additional rents
left to the upper tail, whose mass is $1-G_i(v)$. Their ratio is
the marginal rent cost deducted from $v$ in the seller's virtual
value. A smaller inverse hazard therefore induces higher quality
at the same valuation, closer to the efficient level.

The proposition establishes that the more consumer-oriented
market has a smaller inverse hazard at every valuation in the
common continuous support, regardless of whether the minimum-type
restriction binds. This is stronger than the
stochastic ordering of the distributions. A given valuation
occupying a higher rank does not by itself imply a smaller inverse
hazard, because the density also changes. The fixed-rank and
fixed-valuation results concern different buyers. At a given
rank, valuation is lower in the more consumer-oriented market.
A buyer with a given valuation instead occupies a higher rank
and receives higher quality.

\subsubsection{Comparative statics beyond private values}
\label{sec:mc-welfare}

The private-value results extend beyond type-independent costs, but the
fixed-type and rank comparisons require different restrictions. Fix a
common feasible minimum \(\underline\theta\in[0,\bar\theta)\) and
\(0\leq\beta_2<\beta_1<1\). Write
\(G_i=G_{\beta_i,\underline\theta}\), with rank schedules
\((Q_i,\phi_i,x_i)\), selected minimum \(\ell_i\), and top atom
starting at rank \(b_i\). Let \(x_i^V(\theta)\) and \(U_i(\theta)\)
denote quality and rent at a fixed type. Superscript \(0\) denotes the zero-floor
optimum, so \(\ell_i^0=\theta_{\beta_i}^0\). Define
\[
 \kappa(x):=\frac{xC'''(x)}{C''(x)},\qquad
 \bar\kappa:=\sup_{0<x\leq\bar x}\max\{\kappa(x),0\},\qquad
 \bar p:=\max_{0\leq\theta\leq\bar\theta}m'(\theta).
\]
We distinguish the slope restriction
\begin{equation*}
 \bar\kappa<\infty,\qquad \bar p<\frac{2}{4+\bar\kappa},
 \tag{S}\label{eq:mc-slope}
\end{equation*}
from concavity of the type-dependent cost component,
\begin{equation*}
 m''(\theta)\leq0,\quad 0\leq\theta\leq\bar\theta.
 \tag{R}\label{eq:mc-rank}
\end{equation*}
Condition~\eqref{eq:mc-slope} limits how rapidly the marginal cost
of quality rises with type. Willingness to pay per unit of quality
rises at rate one, so $m'(\theta)$ measures the fraction of that
increase absorbed by higher marginal cost. For quadratic $C$,
$\bar\kappa=0$ and the condition reduces to $\bar p<1/2$:
less than half of the increase is absorbed by cost.
For general $C$, the bound also depends on $\bar\kappa$, the
largest positive elasticity of $C''$ with respect to quality.
A larger $\bar\kappa$ requires a tighter bound on $m'$.

Condition~\eqref{eq:mc-rank} restricts how the cost slope changes
across types. It requires $m'$ to be nonincreasing, or equivalently
$\nu'=1-m'$ to be nondecreasing. Thus, the marginal increase in
net value is weakly larger at higher types.

\begin{proposition}[Comparisons beyond private values]

Maintain Assumptions~1--3 and fix \(0\leq\beta_2<\beta_1<1\). Then:
\begin{enumerate}
\item[(i)] If \eqref{eq:mc-slope} holds, conclusions (i)--(v) of
Proposition~3 extend to type-dependent costs, with gross type \(\theta\)
in place of \(v\) and highest type \(\bar\theta\) in place of one.
In particular, quality, virtual net value, and the hazard are strictly
higher at every common interior type. Moreover, \(b_2>b_1\), and \(G_2\) is
first-order stochastically dominated by \(G_1\). Consumer surplus
is strictly higher, while profit and total surplus are strictly lower.
Moreover,
$U_2(\theta)\geq U_1(\theta)$, $\ell_1\leq\theta\leq\bar\theta$, with strict inequality for \(\theta>\ell_1\).

\item[(ii)] If \eqref{eq:mc-rank} holds and
\(\underline\theta\leq\ell_2^0\), then \(b_2>b_1\) and for $0\leq u<b_2$, $Q_2(u)<Q_1(u)$, $\phi_2(u)<\phi_1(u)$ and $x_2(u)<x_1(u)$.
Above and at \(b_2\), all functions equal their upper bounds.

\item[(iii)] If \eqref{eq:mc-slope} holds and
\(\underline\theta>\ell_1^0\), quality and virtual net value are strictly
higher in market 2 at rank zero but strictly lower on \((b_1,b_2)\).
In quantile space, therefore, they cross.
\end{enumerate}
\end{proposition}
These are sufficient conditions for different results.
Under~\eqref{eq:mc-slope}, greater consumer orientation raises
quality and rents at every common interior type, whether or not
the minimum-type restriction binds.
Under~\eqref{eq:mc-rank}, types, virtual net values, and qualities
are lower at each rank below $b_2$, provided the minimum-type
constraint leaves both unconstrained optimizers unchanged.
The two markets coincide at the remaining ranks. The result about qualities in rank space concerns buyers whose types differ across
markets. It therefore does not determine how a given type is
treated. Since concavity of $m$ allows its slope to be large,
this result can hold even when the quality
result for a fixed valuation reverses near the top.

\paragraph{Consumer friendliness and favorableness.}
\citet{laffonttirole1993} compare distributions through
first-order stochastic dominance and hazard-rate order and distinguish fixed-type incentives
from averages across types. In their procurement model, a less favorable
distribution for the regulator leaves weakly greater rent to a firm of a given
cost type. Our use of \emph{consumer friendliness} refers instead to higher
expected consumer surplus. The connection is the distinction between which types
populate the market and the incentives governing the treatment of a given type.

Here the relevant distortion on the continuous support is
\[
 d_i(\theta):=\frac{1-G_i(\theta)}{g_i(\theta)}
 =\nu(\theta)-C'(x_i^V(\theta)).
\]
Under \eqref{eq:mc-slope}, the market shifts toward lower types but
\(d_2<d_1\) at every type in the common interior support, so quality is higher there.
The envelope theorem and the weakly lower minimum give
\[
 U_2(\theta)-U_1(\theta)
 =\int_{\ell_2}^{\ell_1}x_2^V(z)\,dz
  +\int_{\ell_1}^{\theta}[x_2^V(z)-x_1^V(z)]\,dz\geq0.
\]
The new allocation would generate
more total surplus at the original distribution. Expected total surplus
falls because the downward distributional shift more than offsets this
improvement at each fixed type.
Type-dependent costs can instead separate consumer friendliness in the aggregate from individual gains in rent.

\begin{proposition}[Higher expected surplus with lower individual rents]
Maintain Assumptions~1--3, fix \(0\leq\beta_2<\beta_1<1\), and suppose
\(m'(\bar\theta)>1/2\). There exists
\(\theta^\dagger\in(0,\bar\theta)\) such that, for every common feasible
minimum \(\underline\theta\in(\theta^\dagger,\bar\theta)\), both optimal
markets have support \([\underline\theta,\bar\theta]\) and satisfy
$x_2^V(\theta)<x_1^V(\theta)$ for $\underline\theta<\theta<\bar\theta$ and $U_2(\theta)<U_1(\theta)$ for $\underline\theta<\theta\leq\bar\theta$. Information rents are zero at the lowest type in both markets, and the highest-type allocation is \(\bar x\). Nevertheless, consumer surplus is strictly higher in market 2,
while profit and total surplus are strictly lower. Market 2 first-order
stochastically dominates market 1 and has a larger atom at the top.
\end{proposition}
The threshold compares the effect of type on cost with its effect
on net value near the highest type. The Appendix establishes that
\begin{equation*}
 \lim_{\theta\uparrow\bar\theta}
 \frac{d_i(\theta)^2}{\bar\theta-\theta}
 =
 2C''(\bar x)\bar x\,
 \frac{\nu'(\bar\theta)+(1-\beta_i)m'(\bar\theta)}
      {2-\beta_i}.
\end{equation*}
The fraction is a weighted average of the cost and net-value
slopes. Lowering $\beta$ increases the relative weight on the
cost slope. Since $\nu'=1-m'$, the condition
$m'(\bar\theta)>1/2$ means that the cost slope exceeds the
net-value slope. The more consumer-oriented market therefore
has a larger gap between actual and virtual net value just below the top,
and hence lower quality at the same type.

A sufficiently high minimum type retains only the region where
this quality ordering holds. It also gives both rent schedules
the same starting point: zero rent at $\underline\theta$.
Since $U_i'=x_i^V$, lower quality throughout the interior implies
lower rents at every type above the minimum, including the
highest type despite its efficient allocation. The larger inverse
hazard also implies a smaller hazard throughout the common
interior support. Starting from the same lowest type, market 2
therefore first-order
stochastically dominates market 1 and has a larger atom at the top.

Together with the earlier payoff result, the ordering of rent schedules
gives
\begin{equation*}
 \int U_2\,dG_1
 <
 \int U_1\,dG_1
 <
 \int U_2\,dG_2.
\end{equation*}
Evaluating the new rent schedule at the original distribution
would lower consumer surplus. The new distribution places more
weight on higher types, who earn greater rents within each
market. This change in composition more than offsets the
reduction in each type's rent.

\subsection{Varying the minimum feasible type}
\label{subsec:minimum-feasible-type}

We now hold the weight on profit fixed and vary only the lowest type that may enter the optimal market.  Let $\theta_\beta^0$ denote the endogenous lower support point when a payoff-neutral type is feasible.  The central question is whether the new minimum support type lies below the freely chosen endpoint or truncates the optimal market.

If $\ell_1=\ell_2$, then $G_1=G_2$. Hence, in the following result, we assume that $\ell_2>\ell_1$, and put
$a:=G_1(\ell_2-)$, so that any mass at the new cutoff is retained. We refer to $G_2$ as the higher-floor market and to $G_1$ as the lower-floor market.
\begin{proposition}[Upper-tail truncation and the rent reset]
\label{prop:main-eligibility}
Maintain Assumptions~\ref{ass:main-nu}--\ref{ass:main-curvature}, fix $\beta\in[0,1)$, and let
$0\leq\underline\theta_1<\underline\theta_2<\bar\theta$ so that $\ell_2>\ell_1$. Then:
\begin{enumerate}[label=(\roman*)]
 \item \emph{Upper-tail truncation.} The higher-floor market is the lower-floor market conditional on types weakly above $\ell_2$:
 \begin{equation*}
  G_2(E)
  =\frac{G_1(E\cap[\ell_2,\bar\theta])}{1-a}
  \qquad\text{for every Borel set }E.
 \end{equation*}
 Equivalently, the retained ranks are renormalized according to
 \begin{equation*}
  Q_2(u)=Q_1(a+(1-a)u),\qquad
  \phi_2(u)=\phi_1(a+(1-a)u),\qquad
  x_2(u)=x_1(a+(1-a)u).
 \end{equation*}

 \item \emph{Composition and payoffs.} The higher-floor market first-order stochastically dominates the lower-floor market.  Its type, virtual-value, and quality schedules are weakly higher at every rank, its separating region is shorter, and its atom at the highest type is larger.  Moreover,
 $\Pi(G_2)>\Pi(G_1)$, $\TS(G_2)>\TS(G_1)$ and $\CS(G_2)<\CS(G_1)$.

 \item \emph{Allocation conditional on type.} Every type retained in both markets receives the same virtual net value and the same quality.  On the common continuous support, the hazard and inverse hazard are also unchanged.  If $x_j(\theta)$ denotes quality as a function of type, then
$x_2(\theta)=x_1(\theta)$ for $\theta\in[\ell_2,\bar\theta]$.

 \item \emph{Uniform rent and payment reset.} If $U_j(\theta)$ and $t_j(\theta)$ denote rent and payment, then every retained type satisfies
$ U_2(\theta)=U_1(\theta)-U_1(\ell_2)$ and $t_2(\theta)=t_1(\theta)+U_1(\ell_2)$. Thus, all retained rents fall by the same amount and all retained payments rise by that amount.
\end{enumerate}
\end{proposition}

Increasing the minimum feasible type changes rent extraction even when it leaves retained allocations
unchanged. Until the minimum feasible type exceeds the freely selected lower
endpoint, the optimal market is unaffected. Once it raises the selected
minimum from \(\ell_1\) to \(\ell_2\), every retained type receives the same
allocation but pays the additional amount \(U_1(\ell_2)>0\). In particular,
\[
 \operatorname{CS}(G_2)-\operatorname{CS}(G_1)
 =\left[\int U_1\,dG_2-\int U_1\,dG_1\right]-U_1(\ell_2)<0.
\]
The bracketed term is positive since selecting higher types would raise average
rent at the original menu. The principal's common payment increase more
than offsets this effect. Every retained type loses rent, while expected
total surplus rises because types generating lower gains from trade have
been removed. There is no change in any retained type's allocation or gains from trade. Conversely, reducing the minimum feasible type can protect existing participants' rents. When the support expands, retained types receive the same quality at uniformly lower payments.

\subsection{Value--cost co-movement}
\label{subsec:value-cost-dependence}

Under private values, a buyer's willingness to pay is also her net value to the seller and buyer together. When costs are type-dependent, a buyer can have high willingness to pay and still create only modest gains from trade because she is also expensive to serve. We use this distinction to investigate what happens when willingness to pay and costs rise more closely together. We consider two environments with the same type-independent allocation
cost \(C\), the same profit weight, and the same feasible net-value
interval \([\underline y,\bar\nu]\), where $\underline y\in[0,\bar\nu)$.
The two environments differ only in how each net value is divided between willingness to pay and seller's cost.

In environment $i\in\{H,L\}$, let
\[
 \nu_i(\theta)=\theta-m_i(\theta),
 \qquad
 \tau_i(y)=\nu_i^{-1}(y),
 \qquad
 k_i(y)=\tau_i(y)-y=m_i(\tau_i(y)).
\]
Thus, a buyer with net value $y$ has gross willingness to pay $\tau_i(y)$ and generates type-dependent marginal cost $k_i(y)$. The distribution choice is over
$\Theta_i(\underline y)
 :=[\tau_i(\underline y),\tau_i(\bar\nu)]
 =[\tau_i(\underline y),\bar\theta_i]$.

Define the \emph{value--cost loading} by
\begin{equation*}
 \lambda_i(y)
 :=k_i'(y)
 =\tau_i'(y)-1
 =\frac{1}{\nu_i'(\tau_i(y))}-1
 \geq0.
\end{equation*}
It captures the increase in cost associated with a one-unit increase in net value. Since
$\tau_i'(y)=1+\lambda_i(y)$,
it also determines how quickly gross willingness to pay rises with net value.

An equivalent measure of co-movement is the \emph{cost-absorption rate}.
Choose $\underline r>0$ so that Assumption~\ref{ass:main-nu} holds in both
environments with this common bound. Then
\[
 \sigma_i(y)
 :=1-\nu_i'(\tau_i(y))
 =m_i'(\tau_i(y))
 =\frac{\lambda_i(y)}{1+\lambda_i(y)}
 \in[0,1-\underline r]\subset[0,1).
\]
For a marginal increase in gross willingness to pay, the fraction
$\sigma_i(y)$ is absorbed by the type-dependent cost component, while
$1-\sigma_i(y)$ becomes net value. Private values correspond
to $\sigma_i\equiv0$.

\begin{definition}[Stronger value--cost co-movement]
\label{def:dependence-order}
Environment $H$ has stronger value--cost co-movement than environment $L$ if
$\lambda_H(y)\geq\lambda_L(y)$, or, equivalently, if $\sigma_H(y)\geq\sigma_L(y)$ for every $y\in[0,\bar\nu]$.
\end{definition}

In the stronger co-movement environment, the same increase in net value is accompanied by larger increases in both willingness to pay and costs. Since $\tau_i(0)=k_i(0)=0$, the derivative order also implies
$\tau_H(y)\geq\tau_L(y)$ and $k_H(y)\geq k_L(y)$ for $y\in[0,\bar\nu]$.

A simple affine example is
\begin{equation*}
 m_\kappa(\theta)=\kappa\theta,
 \qquad
 \nu_\kappa(\theta)=(1-\kappa)\theta,
 \qquad
 \tau_\kappa(y)=\frac{y}{1-\kappa},
 \qquad
 \lambda_\kappa=\frac{\kappa}{1-\kappa},
 \qquad
 \kappa\in[0,1).
\end{equation*}
A larger $\kappa$ leaves the feasible net-value interval unchanged but makes both willingness to pay and seller's cost rise more steeply with net value.

Before stating the result, recall that a virtual-value profile $\phi$ determines the seller's quality at each rank, $q(\phi(u))$, and the seller's profit, $\int_0^1\pi(\phi(u))\,du$. Holding $\phi$ fixed therefore means holding fixed the quality schedule and profit while asking which actual types support that schedule.

\begin{proposition}[Value--cost co-movement]
\label{thm:main-dependence}
Maintain the setup above and Assumptions~\ref{ass:main-nu}--\ref{ass:main-C} in both environments, and suppose that $H$ has a stronger value--cost co-movement than $L$.
\begin{enumerate}[label=(\roman*)]
 \item \emph{Holding the virtual-value profile fixed.} Let $\phi$ be a regular virtual-value profile with a terminal interval on which $\phi=\bar\nu$, and let
 $Y_i^\phi(u):=\nu_i(Q_i^\phi(u))$
 be the induced net value at rank $u$. If $\phi$ is feasible under the common floor $\underline y$ in environment $L$, then it is feasible in environment $H$ and
 $Y_H^\phi(u)\geq Y_L^\phi(u)$ and $Q_H^\phi(u)\geq Q_L^\phi(u)$ for $u\in[0,1]$.
 The two environments generate the same quality schedule and the same profit,
 $x_H^\phi(u)=x_L^\phi(u)=q(\phi(u))$ and $\Pi_H(\phi)=\Pi_L(\phi)$, while
$\CS_H(\phi)\geq\CS_L(\phi)$.

 \item \emph{Holding net value fixed after reoptimization.} Maintain Assumption~\ref{ass:main-curvature} and fix $\beta\in[0,1)$. At every net value $y$ in the common separating support of the two optimal markets, let $\phi_i(y)=C'(x_i(y))$.
 Then
 \[
  x_H(y)\leq x_L(y),
  \qquad
  \phi_H(y)\leq\phi_L(y),
  \qquad
  y-\phi_H(y)\geq y-\phi_L(y).
 \]
 Hence the stronger environment assigns weakly lower quality and virtual value to a buyer with a given net value, and generates a weakly larger rent wedge and deadweight loss. It also selects a weakly higher lowest net value and a weakly higher lowest gross type. If the lower endpoint is freely chosen in both environments, quality at the selected lowest type is weakly higher in $H$.

 \item \emph{Holding gross willingness to pay fixed.} Under the assumptions of part~(ii), let $\theta$ belong to the intersection of the realized gross-type supports, and let $U_i(\theta)$ denote information rent. Then,
 $\nu_H(\theta)\leq\nu_L(\theta)$, $x_H(\theta)\leq x_L(\theta)$ and $U_H(\theta)\leq U_L(\theta)$,
 and type-level total surplus is lower in $H$.

 \item \emph{Holding seller profit fixed.} Maintain Assumption~\ref{ass:lf-strong-curvature}, and let
 \[
  \Gamma_i(p)
  :=\max\{\CS_i(G):
      G\in\Delta(\Theta_i(\underline y)),\
      \Pi_i(G)=p\}.
 \]
 Then $\Gamma_H(p)\geq\Gamma_L(p)$
 for every profit attainable in environment $L$.
\end{enumerate}
\end{proposition}

\begin{figure}[t]
 \centering
 \captionsetup{font=small,skip=4pt}
 \includegraphics[width=0.98\textwidth]{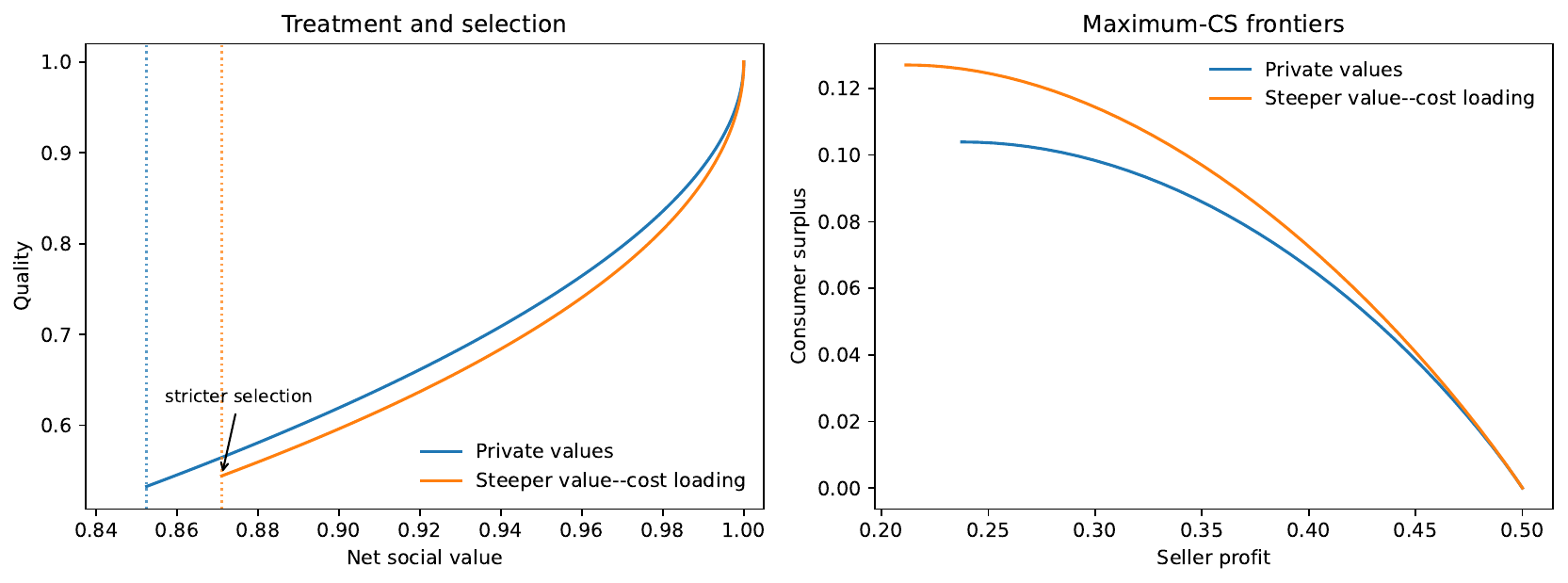}
 \caption{Value--cost co-movement in an affine example. Here $C(x)=x^2/2$ and $\nu(\theta)=r\theta$, normalized so that net value lies in $[0,1]$. The private-value environment has $r=1$, the stronger environment has $r=0.6$, and the left panel fixes $\beta=0.4$. At a common net value, the stronger environment assigns lower quality and selects a higher lower endpoint. At a common seller profit, however, it can generate weakly higher consumer surplus because the same virtual-value profile preserves quality and profit while being supported by higher actual net values.}
 \label{fig:value-cost-dependence}
\end{figure}

 Figure \ref{fig:value-cost-dependence} illustrates the results in an affine example. Holding the virtual-value profile fixed isolates what the co-movement itself does. To see the comparison, use $t=-\log(1-u)$ and write $\widehat\phi(t)=\phi(1-e^{-t})$. If $T$ is the beginning of the terminal highest-type interval, the induced net-value path satisfies
\begin{equation*}
 \dot Y_i(t)
 =\frac{Y_i(t)-\widehat\phi(t)}
        {1+\lambda_i(Y_i(t))},
 \qquad
 Y_i(T)=\bar\nu.
\end{equation*}
At a common rank and actual net value, a larger
$\lambda$ weakly lowers the slope in this equation. Starting
from the same terminal value $\bar\nu$, the backward reconstruction equation
therefore gives $Y_H^\phi\geq Y_L^\phi$. In particular, the minimum
net value cannot fall, so every profile of this form that is
feasible in $L$ remains feasible in $H$.

The common virtual net value profile $\phi$ fixes quality by rank
and seller profit. The higher actual net values in $H$ allow
those same qualities to generate weakly greater total surplus.
Since profit is unchanged, any gain accrues to consumers.
Applying this construction to the markets on $L$'s
maximum consumer surplus boundary gives
$\Gamma_H(p)\geq\Gamma_L(p)$ at every interior profit.
The endpoint results follow by taking limits.

To understand the comparison at a given net value, first hold
fixed a distribution $F$ of net values, with positive density
$f$ on the continuous part of its support. For a given quality
schedule $x(y)$, incentive compatibility makes information rents
grow with net value at rate $[1+\lambda_i(y)]x(y)$. The seller's
raw virtual net value is therefore
\begin{equation*}
 \widehat\varphi_i^F(y)
 =y-[1+\lambda_i(y)]\frac{1-F(y)}{f(y)}.
\end{equation*}
A larger loading raises the seller's cost (in terms of information rents) of
providing quality because improving the allocation near $y$ requires
leaving more rent to buyers with higher net values, while the
net gains from that quality remain unchanged. The proposition establishes the result after the
re-optimization of market composition at the same
profit weight. Although the distribution of net values can
adjust, quality is weakly lower in $H$ at every net value in
the common separating support. The lowest selected net value
is also weakly higher.
When both lower endpoints are freely chosen, quality at the
lowest selected type is weakly higher in $H$, so the more selective market can begin at a higher quality.

Finally, at the same gross willingness to pay, the stronger environment assigns a larger share of that willingness to cost and leaves a smaller net value. Quality, information rent, and total surplus are therefore lower. The frontier conclusion for a stronger co-movement environment comes from holding the net-value opportunity set fixed and asking how effectively different value--cost mappings generate information rents at a given seller profit.

Private values are the limiting case, with $\lambda(y)\equiv0$. They provide weakly higher quality and lower distortion to a buyer with a given net value and lead to a weakly lower selection threshold. By contrast, a stronger co-movement makes a given screening profile more effective at producing consumer rents. Therefore, a weaker co-movement improves treatment of a fixed type, while a stronger co-movement expands the consumer surplus possibilities available at a fixed profit.

\paragraph{Value-cost co-movement in insurance.}
Consider the insurance specification in Section~\ref{subsec:normalization-applications}, where a buyer with loss probability $p$ values coverage at rate $\chi p$, expected claims cost the insurer $p$ per unit of coverage, and $\chi>1$. Gross willingness to pay is $\theta=\chi p$, while net value is
$y=(\chi-1)p$.
To apply the result, hold $C$ and the feasible net-value
interval $[\underline y,\bar\nu]$ fixed. The corresponding
feasible loss probabilities are
\[
 p\in\left[
 \frac{\underline y}{\chi-1},
 \frac{\bar\nu}{\chi-1}
 \right].
\]
We take $\bar\nu\leq\chi-1$ in both environments, so these
intervals lie in $[0,1]$. The feasible risk interval therefore
changes with $\chi$, keeping the range of net insurance values
unchanged.

Expressed as functions of $y$,
\[
 \tau(y)=\frac{\chi}{\chi-1}y,
 \qquad
 k(y)=\frac{1}{\chi-1}y,
 \qquad
 \lambda(y)=\frac{1}{\chi-1}.
\]
A smaller $\chi$ gives stronger co-movement. At a given net
insurance value $y$, the policyholder has a higher loss probability
$p=y/(\chi-1)$ and higher gross willingness to pay. Holding the
insurer's virtual value profile fixed preserves coverage by rank
and insurer profit while weakly increasing expected policyholder
surplus.

After market composition is reoptimized at the same profit
weight, coverage is weakly lower at every net value in the common
separating support of the two optimal pools. The lowest selected
net value is also weakly higher. At a given gross willingness to
pay $\theta$, the policyholder's loss probability $p=\theta/\chi$
is higher in the stronger environment. Expected claims per unit
of coverage are therefore higher, while net value, coverage,
and information rent are weakly lower.

The results at a fixed net value or fixed willingness to pay
thus concern policyholders with different loss probabilities.
The frontier result establishes that, at each previously
attainable insurer profit, a freely chosen pool can generate
weakly greater expected policyholder surplus in the stronger
co-movement environment.

\section{Extensions}
The analysis so far has imposed neither restrictions on market composition nor a cost of inducing it. We first investigate the case where there is a cost of inducing the chosen market, and then discuss what happens if the induced market composition must have an exogenously fixed average type.

\subsection{Costly Attenuation}
\label{sec:costly-attenuation}

This extension considers a cost of inducing each market, motivated by costly obfuscation, as in \citet{CondorelliSzentes2020}.
This has a hold-up interpretation. Before the seller designs its menu, consumers may invest in privacy, obfuscation, or better outside options to limit later surplus extraction. A customer of cloud services, for example, can invest in portable data formats, interoperable applications, or migration tools to make it easy to switch from one provider to another. These investments are costly today but reduce future switching costs and willingness to pay at renewal, following the standard lock-in logic of \citet{FarrellKlemperer2007}. In an insurance environment, prevention can make the pool of individuals in the market for insurance contracts healthier by lowering residual risk and insurance need, but may require effort or resources. In each case, costly investment weakens the type distribution before screening occurs.

Let $\mathcal K(G)$ be the cost of inducing $G$ and consider
\begin{equation}
 \widehat U_{\underline\theta}(\beta)
 :=\max_{G\in\Delta([\underline\theta,\bar\theta])}
 \{\CS(G)+\beta\Pi(G)-\mathcal K(G)\},
 \qquad \beta\geq0.
 \label{eq:costly-attenuation-problem}
\end{equation}
Recall the raw adjusted revenue curve $\widehat R_G$ in
\eqref{eq:reduction-adjusted-raw}. We call $\mathcal K$ an
\emph{attenuation cost} if it is finite-valued, weakly lower
semicontinuous, and satisfies
\begin{equation}
 \widehat R_F(u)\geq\widehat R_G(u)
 \quad\text{for every }u\in[0,1)
 \quad\Longrightarrow\quad
 \mathcal K(F)\leq\mathcal K(G).
 \label{eq:costly-attenuation-order}
\end{equation}
The raw revenue curve captures the seller's adjusted revenue from upper-tail
quality increments, before ironing. The restriction requires that a
market be weakly cheaper to induce whenever this entire curve moves
upward. The top market $\delta_{\bar\theta}$ dominates every feasible
market in this order, so its cost can be normalized to zero.

\paragraph{Relation to stochastic orders.}
Our cost restriction is weaker than monotonicity in standard stochastic
orders. Under Assumption~\ref{ass:main-nu}, raw adjusted revenue
dominance implies
\begin{equation*}
 \widehat R_F\geq\widehat R_G
 \quad\Longrightarrow\quad
 F\succeq_{\mathrm{FOSD}}G
 \quad\Longrightarrow\quad
 F\succeq_{\mathrm{icx}}G.
\end{equation*}
First-order stochastic dominance means that types are weakly higher
at every quantile under $F$. Increasing-convex order, denoted
$F\succeq_{\mathrm{icx}}G$, is weaker. Defining the stop-loss function
$S_G(z):=\E_G[(\theta-z)_+]$, it is equivalent to
$S_F(z)\geq S_G(z)$ for every $z\in\mathbb R$.\footnote{See
\citet[Section~4.A.1]{shakedshanthikumar2007}.}
Thus, increasing-convex order compares expected upper-tail excess
at every threshold and also allows mean-preserving spreads, even
when some quantiles fall.

The implications for cost restrictions run in the opposite direction.
A cost decreasing in increasing-convex order is also decreasing in
first-order stochastic dominance, and either property implies
\eqref{eq:costly-attenuation-order}. We require only the last condition, i.e. that
changes that raise the raw adjusted revenue curve cannot increase
the cost of inducing the market. Under private values, raw-revenue dominance coincides with
first-order dominance, since $\widehat R_G(u)=(1-u)Q_G^+(u)$.
With type-dependent costs, it can be more restrictive,
so the corresponding assumption on attenuation costs is weaker.

This restriction accommodates a broad class of attenuation costs.
In particular, costs separable across buyers,
$\mathcal K(G)=\E_G[k(\theta)]$, satisfy it whenever $k$ is
finite, nonincreasing, and lower semicontinuous. This includes
convex attenuation costs, $k(\theta)=c(\bar\theta-\theta)^2$,
$c>0$, for which the marginal cost of reducing a buyer's type
increases with the extent of attenuation. Such costs need not
be decreasing in increasing-convex order. An example satisfying
that stronger restriction is
\begin{equation}
 \mathcal K_\omega(G)
 =\int_{\underline\theta}^{\bar\theta}\omega(z)
 \bigl[(\bar\theta-z)-S_G(z)\bigr]\,dz,
 \label{eq:costly-attenuation-example}
\end{equation}
where $\omega\geq0$ is integrable. This specification is also
separable across buyers and measures the weighted reduction in
expected excess valuations above each threshold, relative to
the top market.

\begin{proposition}[Costly attenuation]
\label{prop:costly-attenuation}
Maintain Assumptions~\ref{ass:main-nu}--\ref{ass:main-C}, and let
$\mathcal K$ be an attenuation cost, normalized by
$\mathcal K(\delta_{\bar\theta})=0$.
\begin{enumerate}[label=(\roman*)]
 \item Problem~\eqref{eq:costly-attenuation-problem} has a solution.
 Every optimizer is regular, and no positive mass of realized types
 is assigned zero quality.
 \item If $\beta\geq1$, the unique optimum is
 $\delta_{\bar\theta}$.
 \item There is a cutoff $\beta^{\mathcal K}\in[0,1]$ such that
 the top market is optimal if and only if
 $\beta\geq\beta^{\mathcal K}$, and uniquely optimal when
 $\beta>\beta^{\mathcal K}$.
 Every non-top optimum for $\beta\in[0,1)$ has connected support,
 no atom below $\bar\theta$, a strictly separating interior, and
 a positive atom at $\bar\theta$ served at the first-best quality.
 At a cutoff below one, a non-top optimum may coexist with the top market.
 \item If $\mathcal K(G)=\E_G[k(\theta)]$, with $k$ continuous at
 $\bar\theta$ and $k(\bar\theta)=0$, then
 $\beta^{\mathcal K}=1$.
\end{enumerate}
\end{proposition}

Thus, a general attenuation cost can make the degenerate top market
optimal even when $\beta<1$. Whenever a nondegenerate market is
optimal, however, rents are still created through a strictly separating
interior and the highest-type atom is served efficiently. These
claims apply to every optimizer. Part~(iv) distinguishes costs that are continuous at
the highest type from more general costs of inducing a distribution.
For every $\beta<1$, moving a sufficiently small group of buyers
slightly below the top can still create rents whose weighted gain
exceeds the expected cost. Thus, such costs can change the optimal
market but cannot make the top market optimal below one.

\subsection{Fixed-mean optimal market composition and buyer-optimal learning}
\label{sec:fixed-mean-learning}

We now consider the case of private values and require the chosen market to have an exogenously fixed average valuation. It is without loss to normalize the
highest valuation to one so the lowest feasible type in the support becomes $\underline v=\underline \theta/ \bar \theta$. The corresponding problem is

\[
 U^{\mathrm{mean}}_{\underline v,\mu}(\beta)
 :=
 \max_{\substack{G\in\Delta([\underline v,1])\\
                  \int_{\underline v}^1 v\,dG(v)=\mu}}
 \{\CS(G)+\beta\Pi(G)\},
 \qquad
 \underline v<\mu<1,\quad \beta\geq0.
\]
This problem can be interpreted as one where the specific technology of reshaping a given market is information provision to consumers prior to screening. Specifically, it is equivalent to designing learning about an
initially uncertain valuation $V\in\{\underline v,1\}$ with
prior mean $\mu$. An information designer commits to a
Blackwell experiment before the principal chooses its menu and the principal
observes the experiment but not the signal realization. Because the agent's payoff is linear in type, the expected payoff from any subsequent principal-optimal menu depends only on the posterior mean type $w=\E(V\mid s)$, where $s$ is the realization of the Blackwell experiment observed privately by the agent. For a binary prior,
Bayes plausibility requires that the posterior distribution of mean types is supported on a subset of $[\underline v,1]$ and its mean is equal to $\mu$. Varying $\beta$ gives the analogous fixed-mean
support problems. At $\beta=0$, this corresponds to buyer-optimal learning.
\citet{Kyriazis2026Learning} studies this problem when there are finitely many types and additionally considers the optimal joint design of learning and market segmentation.

\paragraph{Structure under the mean constraint.}
The fixed mean changes the regularization argument. Regularization
preserves seller profit and weakly raises consumer surplus, but it
generally raises the mean of the distribution. Bayes plausibility therefore requires an offsetting reduction in posterior means. \citet{Kyriazis2026Learning} shows how this
adjustment can be made in a way to minimize the loss of consumer surplus. When a payoff-neutral type is feasible, this adjustment may require the optimal market to have an atom at a payoff-neutral type. Without a payoff-neutral type,   the solution may feature a lower atom followed by a gap in the support of the optimal market. Buyers in this atom either purchase an isolated low-quality product or are excluded, even when trade with every buyer creates gains. Both configurations can be optimal for the same prior and profit weight, so the optimal market need not be unique.
 Above the lower component necessarily lies one connected region on which
posterior means, virtual values, and qualities rise continuously and
strictly, followed by a positive atom at the highest valuation served
efficiently. Thus, the separating region and efficient top atom from
Theorem~\ref{thm:main-structure} survive, while the additional
constraint can generate exclusion or a lower trading atom followed by a support gap.

\paragraph{Payoff-neutrality and frontier geometry.}
When a payoff-neutral type is feasible, excess mass can be placed there in the form of an atom, without changing the active part of the distribution. Thus, we can choose an active posterior shape and then scale how frequently it is used. In this case, \citet{Kyriazis2026Learning} shows that the optimized weighted value is concave in $\mu$, the fixed-mean single-market Pareto
frontier is concave, and market segmentation
adds no efficient payoffs.

When no payoff-neutral valuation is feasible,
\citet{Kyriazis2026Learning} shows that the single-market Pareto frontier can be nonconcave, even when segmentation is valuation-independent so no information about consumer valuation is disclosed to the seller, and studies the optimal joint design of information and segmentation.  Thus, in sharp contrast with our unrestricted case, with a fixed average type, segmentation can enlarge the frontier of efficient outcomes.

\section{Conclusion}
\label{sec:conclusion}

In this paper, we study how market composition shapes the welfare consequences of screening when the seller's costs may depend on buyers' types. By allowing the type distribution to vary, we characterize the markets that generate efficient trade-offs between consumer surplus and seller profit. Every nondegenerate Pareto-optimal market serves all selected buyers and combines a strictly separating interior, through which information rents propagate, with an atom at the highest type, served at first-best quality.
Beyond the Pareto frontier, we study the geometry of the attainable surplus set. Under the maintained curvature conditions, the set is convex when a payoff-neutral type is feasible or values are private. Without such a type, however, type-dependent costs can make it nonconvex. Segmentation then expands the set of implementable payoffs without increasing the maximum consumer surplus attainable at any level of seller profit. Thus, expanding the set of attainable outcomes need not improve the efficient trade-off between consumer surplus and profit.

Our analysis provides a way to understand how institutions that form or reshape markets affect the division of surplus through the contracts they induce a principal to offer. Specific market-formation technologies may restrict the distributions that can be induced, entail additional costs, or affect agents outside the market. Studying such constraints is a natural direction for future research, as pursued in \citet{Kyriazis2026Learning}. Our results here provide
a benchmark for distinguishing the surplus limits imposed by screening
from those imposed by a particular technology.

\clearpage
\printbibliography

\newpage

\section*{\LARGE Appendix}
\appendix

\theoremstyle{plain}
\newtheorem{appendixlemma}{Lemma}[section]
\newtheorem{appendixproposition}[appendixlemma]{Proposition}

\section{Optimal Markets and the Pareto Frontier}
\label{app:paper-main-theorem}

In this section we prove Theorem~\ref{thm:main-structure}.

\subsection{Screening a given market and reconstructing its distribution}
\label{app:screening-foundations}

Maintain Assumptions~\ref{ass:main-nu}--\ref{ass:main-C}. We first
solve the seller's problem for an arbitrary distribution and then, given a monotone virtual value profile,
we find the regular market that induces it.

For a feasible distribution $G$ let $\phi_G$ be the associated (ironed) virtual value. We
set $\phi_G(0)=\phi_G(0+)$ and
$\phi_G(1)=\phi_G(1-)$. The first limit may be $-\infty$. The second is
finite. We also set
$q(-\infty)=\pi(-\infty)=0$ when considering endpoint allocations.

\begin{appendixlemma}[Revenue curves and virtual net values]
\label{lem:gs-concavification-regularity}
The curve $R_G$ is finite, continuous, and absolutely continuous on
$[0,1]$, with $R_G(0)=\widehat R_G(0)$ and $R_G(1)=0$. Moreover,
\begin{equation}
 R_G(u)=\int_u^1\phi_G(v)\,dv,
 \qquad \phi_G\in L^1(0,1),\qquad \phi_G\text{ is nondecreasing},
 \label{eq:gs-R-tail-phi}
\end{equation}
\begin{equation}
 \phi_G(u)\leq\nu(Q(u))\leq\bar\nu\quad\text{almost everywhere}.
 \label{eq:virtual-net-value-bound}
\end{equation}
The curve is affine on every connected open interval where
$R_G>\widehat R_G$, and on the interior of every rank interval
corresponding to an atom.
\end{appendixlemma}

\begin{proof}
Since $m$ is nondecreasing,
$-m(\bar\theta)\leq\widehat R_G(u)/(1-u)
\leq\nu(Q^+(u))\leq\bar\nu$.
The affine functions $-m(\bar\theta)(1-u)$ and $\bar\nu(1-u)$ also
bound $R_G$ from below and above, respectively. Hence $R_G$ is finite
and continuous at one. To verify continuity at zero, observe that
$\widehat R_G$ is upper semicontinuous because $Q^+$ is
right-continuous and has only upward jumps. For each
$\varepsilon>0$, boundedness and upper semicontinuity at zero give
an affine majorant $\widehat R_G(0)+\varepsilon+K_\varepsilon u$
for some constant $K_\varepsilon$. Since $R_G$ is the least concave
majorant, it lies below this affine function. Letting $\varepsilon$ decrease to zero
gives $R_G(0)=\widehat R_G(0)$ and continuity there. A concave
function continuous on a compact interval is absolutely continuous;
integrating its monotone derivative therefore gives
\eqref{eq:gs-R-tail-phi}.

To prove \eqref{eq:virtual-net-value-bound}, note that for $0\leq a<b\leq1$,
$\widehat R_G(a)-\widehat R_G(b)
 \leq(b-a)\nu(Q^+(a))$.
For $b<1$, subtract the two expressions in
\eqref{eq:reduction-adjusted-raw} and use $Q^+(b)\geq Q^+(a)$ and
$m(Q(s))\geq m(Q^+(a))$ for almost every $s>a$. The case $b=1$
follows by taking a limit.

We call $a$ a contact point if $R_G(a)=\widehat R_G(a)$. At any
contact point $a<1$, the preceding bound and
$R_G(b)\geq\widehat R_G(b)$ imply
$-R'_{G,+}(a)\leq\nu(Q^+(a))$.
The least concave majorant is affine on each open interval where
it strictly exceeds the raw curve. On a maximal such interval
$(a,b)$, the two curves agree at both endpoints, so $R_G$ is the
chord joining them. The same bound then shows that the negative
slope of this chord is at most $\nu(Q^+(a))\leq\nu(Q(u))$ for
$a<u<b$. Since the quantile versions agree almost everywhere,
these two cases prove the desired inequality. Monotonicity,
integrability, and the upper bound also imply that the limit of
$\phi_G$ at one is finite.

Finally, consider the interior of a rank interval corresponding
to an atom. There $Q$ is constant, so $\widehat R_G$ is affine and
$R_G-\widehat R_G$ is nonnegative and concave. If this difference
vanishes at an interior point, concavity makes it vanish throughout
the interval, and $R_G$ is affine. Otherwise, the difference is
strictly positive throughout the interval, where the least concave
majorant is again affine.
\end{proof}

\begin{appendixproposition}[The seller's optimal allocation and payoffs]
\label{lem:gs-seller-reduction}
For every market $G$, the seller's optimal value equals
\begin{equation*}
 \max_{\substack{x:[0,1]\to[0,\bar x]\\x\text{ nondecreasing}}}
 \int_0^1\{\phi_G(u)x(u)-C(x(u))\}\,du,
 \qquad x^*(u)=q(\phi_G(u))\quad\text{a.e.}
\end{equation*}
The maximizing rank allocation is unique almost everywhere, and
\begin{align}
 \Pi(G)&=\int_0^1\pi(\phi_G(u))\,du,
 \notag\\
 \TS(G)&=\int_0^1\{\nu(Q(u))q(\phi_G(u))-C(q(\phi_G(u)))\}\,du,
 \label{eq:gs-TS-formula}\\
 \CS(G)&=\int_0^1\{\nu(Q(u))-\phi_G(u)\}q(\phi_G(u))\,du.
 \label{eq:gs-CS-formula}
\end{align}
\end{appendixproposition}

\begin{proof}
\clubpenalties 3 10000 10000 0\relax
Fix an incentive-compatible allocation and put
$\theta_0=\inf\supp G$. The allocation must be nondecreasing in
type. We apply the standard envelope argument
\citep{myerson1981,rochet1987,Hellwig2010}, allowing for atoms and
gaps in the distribution. Across each support gap, extend the
allocation by delaying any increase until the next realized type.
This minimizes rents without changing quality at realized types.
Setting the lowest type's rent to zero then gives the
profit-maximizing transfers for this allocation:
\begin{equation}
 U(\theta)=\int_{\theta_0}^{\theta}x(s)\,ds,
 \qquad t(\theta)=\theta x(\theta)-U(\theta).
 \label{eq:gs-envelope}
\end{equation}

Consider a marginal increase in quality, assigned to an upper set
of types, with rank cutoff $u$. Let
$\mathcal A_G=\bigcup_{\theta:G(\{\theta\})>0}
(G(\theta-),G(\theta))$ be the union of the interiors of atom
intervals. The increase cannot begin strictly inside one of these
intervals, since identical types must receive the same allocation.
For a cutoff $u\in[0,1)\setminus\mathcal A_G$, the extension across
support gaps chosen above gives the marginal price $Q^+(u)$.
Moreover, for any bounded Borel $h$,
\begin{equation*}
 \Pr\{G(\Theta)>u\}=1-u,
 \qquad
 \E[h(\Theta)\mathbf1_{\{G(\Theta)>u\}}]
 =\int_u^1h(Q(v))\,dv.
\end{equation*}
To verify these identities, represent the type draw as $\Theta=Q(V)$,
where $V$ is uniform on $[0,1]$. Outside the interiors of atom
intervals, the event in the display agrees almost surely with $V>u$.
Thus the revenue from the
quality increase, net of the linear type-dependent cost, is
$\widehat R_G(u)$.

An increasing allocation is obtained by adding these quality
increments. For a step allocation, they are ordinary jumps. For a
general monotone allocation, let $\mu$ be the measure describing
their sizes and rank cutoffs, so that $x(v)=\mu([0,v])$ almost
everywhere. Since no increment begins strictly inside an atom interval,
$\mu(\mathcal A_G)=0$; the quality bound gives
$\mu([0,1])\leq\bar x$. Conversely, increments satisfying these
restrictions implement a monotone allocation that treats identical
types alike. An increment at rank one does not affect payoffs and
may be omitted. Adding the net revenues of the increments and
subtracting the quality cost gives
\begin{equation}
 \begin{aligned}
 \Pi(x;G)
 &=\int\widehat R_G\,d\mu-\int_0^1C(x(u))\,du\\
 &\leq\int R_G\,d\mu-\int_0^1C(x(u))\,du
 =\int_0^1\{\phi_G(u)x(u)-C(x(u))\}\,du.
 \end{aligned}
 \label{eq:gs-layer-profit-exact}
\end{equation}
In the first equality, we may integrate over types and quality
increments in either order because $Q^+$ and $m\circ Q$ are bounded
and $\mu$ is finite. The inequality uses $R_G\geq\widehat R_G$.
For the last equality, substitute \eqref{eq:gs-R-tail-phi} and
apply Fubini's theorem. This is justified even when $\phi_G$ takes
negative values, since
\[
 \int_0^1|\phi_G(u)|\mu([0,u])\,du\leq\bar x\|\phi_G\|_1<\infty.
\]

The integrand on the right of \eqref{eq:gs-layer-profit-exact} is
maximized pointwise by $q(\phi_G)$. This allocation is nondecreasing
and constant on atom intervals and on intervals where
$R_G>\widehat R_G$. Its quality increments therefore occur only
where the two revenue curves agree, so the inequality in
\eqref{eq:gs-layer-profit-exact} holds with equality. An increase
at rank zero is also permissible, since the curves agree there.
Thus the allocation is implementable and attains the bound.
Strict convexity of $C$ gives uniqueness almost everywhere.

Averaging net trade surplus gives \eqref{eq:gs-TS-formula};
subtracting profit gives \eqref{eq:gs-CS-formula}.
Nonpositive virtual values induce $q=\pi=0$.
\end{proof}

\paragraph{Pointwise allocation convention.}
Expected payoffs do not determine the contract at a support endpoint.
Whenever a seller-optimal allocation on a nondegenerate connected
support admits a continuous representative, we use that representative
and its one-sided endpoint limits. Transfers are then given by
\eqref{eq:gs-envelope}.

\paragraph{Reconstructing a regular market.}
We keep $R_G$ fixed and construct a quantile whose raw adjusted revenue
curve is $R_G$. This changes the types in the market while preserving
virtual values and the seller's quality by rank.

\begin{appendixproposition}[Regularization and its welfare consequences]
\label{lem:gs-regularization}
For every market $G$, there is a unique bounded solution $\widetilde Q$ of
\begin{equation}
 (1-u)\widetilde Q(u)=R_G(u)+\int_u^1m(\widetilde Q(s))\,ds,
 \qquad u<1.
 \label{eq:gs-regularized-fixed-point}
\end{equation}
It is a continuous quantile on $[0,1)$, has raw adjusted revenue $R_G$, and
$Q\leq\widetilde Q\leq\bar\theta$.
Its market $\widetilde G$ has the same rank allocation and profit as $G$, and
\begin{equation}
 \CS(\widetilde G)-\CS(G)
 =\int_0^1\{\nu(\widetilde Q(u))-\nu(Q(u))\}q(\phi_G(u))\,du\geq0,
 \label{eq:gs-reg-CS}
\end{equation}
with equality when
\begin{equation}
 (\widetilde Q-Q)q(\phi_G)=0\quad\text{almost everywhere}.
 \label{eq:gs-reg-equality}
\end{equation}
Every lower-support restriction is preserved.
\end{appendixproposition}

\begin{proof}[Proof of the proposition and Lemma~\ref{lem:main-regularization}]
Extend $m$ constantly outside $[0,\bar\theta]$. For a continuous
$R$ with $R(1)=0$ and bounded $R(u)/(1-u)$, define
\begin{equation}
 (T_Rf)(u)=\frac{R(u)+\int_u^1m(f(s))\,ds}{1-u},\qquad u<1.
 \label{eq:app-tail-operator}
\end{equation}
A fixed point of $T_R$ is precisely a solution of the reconstruction
equation with revenue curve $R$. The assumptions on $R$ and boundedness of
the extended $m$ ensure that it maps bounded functions to bounded
functions. The derivative bound on $m$ gives a Lipschitz constant
$1-\underline r$, which the constant extension preserves. Since
the difference between $T_Rf$ and $T_Rh$ is an average of
$m(f)-m(h)$ over $[u,1]$,
\[
 \|T_Rf-T_Rh\|_\infty
 \leq(1-\underline r)\|f-h\|_\infty.
\]
The contraction mapping theorem \citep{deimling1985} therefore gives
a unique bounded fixed point, and iteration of $T_R$ converges to it.

Because $m$ is nondecreasing, $T_R$ preserves pointwise order in both
$R$ and $f$. If $R$ is concave and $f$ is nondecreasing, the numerator
in \eqref{eq:app-tail-operator} is concave and zero at one. Its ratio
to $1-u$ is therefore nondecreasing. For $R=R_G$,
\[
 T_{R_G}Q\geq T_{\widehat R_G}Q=Q^+\geq Q,
 \qquad T_{R_G}\bar\theta\leq\bar\nu+m(\bar\theta)=\bar\theta.
\]
Thus, iteration from $Q$ converges to a nondecreasing fixed point
between $Q$ and $\bar\theta$. Its defining numerator is continuous
and its denominator is positive below one, so the fixed point is
continuous there. Equation~\eqref{eq:gs-regularized-fixed-point}
identifies its raw adjusted revenue as the concave curve $R_G$;
hence the induced market is regular. The same argument applies on
any restricted tail $[a,1)$: the equation at a rank involves only
values at that rank and above.

The seller's allocation and profit depend only on $\phi_G$, so they
are unchanged. Subtracting the consumer surplus formulas gives
\eqref{eq:gs-reg-CS}, and strict increase of $\nu$ gives the stated
equality condition. Taking one-sided limits extends the quantile
order to the endpoints. In particular, every lower bound on the
original types remains valid for the regularized types.
\end{proof}

This does not yet show that every optimizing market is regular,
because the two quantiles may differ at ranks assigned zero quality.
After ruling out exclusion, we will use \eqref{eq:gs-reg-equality} to
identify each optimizing market with its regularization.

\paragraph{The profile domain and an order property.}
Define
\begin{equation*}
 \Phi=\left\{\phi:[0,1]\to[0,\bar\nu]:
 \begin{array}{l}
 \phi\text{ is nondecreasing and right-continuous on }[0,1),\\
 \phi(1)=\phi(1-)
 \end{array}\right\}.
\end{equation*}
For $\phi\in\Phi$, apply the reconstruction argument to
$R_\phi(u)=\int_u^1\phi(s)\,ds$. It gives the unique quantile
\begin{equation}
 (1-u)Q_\phi(u)=\int_u^1\{\phi(s)+m(Q_\phi(s))\}\,ds,
 \qquad 0\leq Q_\phi\leq\bar\theta.
 \label{eq:gs-profile-tail}
\end{equation}
The constant functions zero and $\bar\theta$ give lower and upper
bounds for the iteration. The solution is continuous, locally
absolutely continuous, and satisfies
\begin{equation}
 (1-u)Q_\phi'(u)=\nu(Q_\phi(u))-\phi(u)\geq0
 \quad\text{almost everywhere}.
 \label{eq:gs-Q-ode}
\end{equation}
Indeed, $Q_\phi(u)$ is the average of the nondecreasing function
$\phi+m(Q_\phi)$ over ranks above $u$, so it is at least that
function's value at $u$. Thus $\nu(Q_\phi(u))-\phi(u)\geq0$,
which gives both the inequality in \eqref{eq:gs-Q-ode} and the
monotonicity of the induced types. Under private values,
reconstruction reduces to
$Q_\phi(u)=(1-u)^{-1}\int_u^1\phi$.

\begin{appendixlemma}[Upper-tail changes that raise induced types]
\label{lem:attenuation-tail-order}
For $\phi,\psi\in\Phi$, if
\begin{equation*}
 \int_u^1\{\psi(s)-\phi(s)\}\,ds\geq0
 \quad\text{for every }u\in[0,1],
\end{equation*}
then $Q_\psi(u)\geq Q_\phi(u)$ at every rank. The induced market
first-order stochastically dominates the original, and every
minimum-type constraint remains satisfied.
\end{appendixlemma}

\begin{proof}
For $u<1$, write
\[
 R_\phi(u)=\int_u^1\phi(s)\,ds,
 \qquad
 R_\psi(u)=\int_u^1\psi(s)\,ds.
\]
The hypothesis says that $R_\psi\geq R_\phi$. Recall that
\[
 (T_R f)(u)
 =\frac{R(u)+\int_u^1m(f(s))\,ds}{1-u}.
\]
The reconstruction argument gives $Q_\phi=T_{R_\phi}Q_\phi$
and $Q_\psi=T_{R_\psi}Q_\psi$. To compare these fixed points,
iterate $T_{R_\psi}$ starting from $Q_\phi$: put
$Q^{(0)}:=Q_\phi$, $Q^{(n+1)}:=T_{R_\psi}Q^{(n)}$.
The first iterate is weakly above the initial quantile, since
\[
 Q^{(1)}(u)-Q^{(0)}(u)
 =\frac{R_\psi(u)-R_\phi(u)}{1-u}\geq0.
\]

Since $m$ is nondecreasing, $f\geq g$ implies
$T_{R_\psi}f\geq T_{R_\psi}g$. Induction therefore gives
\[
 Q_\phi=Q^{(0)}\leq Q^{(1)}\leq Q^{(2)}\leq\cdots.
\]
The sequence is increasing at each rank. All iterates remain in
$[0,\bar\theta]$: they are bounded below by $Q_\phi\geq0$, and
\[
 (T_{R_\psi}\bar\theta)(u)
 =\frac{R_\psi(u)}{1-u}+m(\bar\theta)
 \leq\bar\nu+m(\bar\theta)=\bar\theta,
\]
where $\bar\theta$ denotes the constant function on the left.
Since $T_{R_\psi}$ preserves order, this upper bound holds at every
iteration.

The derivative bound in Assumption~\ref{ass:main-nu} also gives
$\|T_{R_\psi}f-T_{R_\psi}g\|_\infty
 \leq(1-\underline r)\|f-g\|_\infty$.
Applying this bound successively to the iterates and the fixed point
$Q_\psi$ gives
\[
 \|Q^{(n)}-Q_\psi\|_\infty
 \leq(1-\underline r)^n\|Q_\phi-Q_\psi\|_\infty
 \longrightarrow0, \qquad n\geq1.
\]
The limit of the increasing sequence is therefore $Q_\psi$, so
$Q_\psi(u)\geq Q_\phi(u)$ for every $u<1$. Taking left limits
extends the inequality to rank one.

Finally, let $V$ be uniform on $[0,1]$ and denote the induced
cumulative distribution functions by $G_\phi$ and $G_\psi$.
Since $Q_\psi(V)\geq Q_\phi(V)$ almost surely, for every type threshold
$\theta$,
\[
 G_\psi(\theta)
 =\Pr\{Q_\psi(V)\leq\theta\}
 \leq\Pr\{Q_\phi(V)\leq\theta\}
 =G_\phi(\theta).
\]
In particular,
$Q_\psi(0)\geq Q_\phi(0)$, so any minimum-type restriction satisfied
by the original market remains satisfied.
\end{proof}

\begin{appendixlemma}[Removing nonpositive virtual net values]
\label{lem:gs-truncation}
If $\phi$ is the profile of a regular market, possibly unbounded below
at zero, then $\phi^+=\max\{\phi,0\}\in\Phi$,
$Q_{\phi^+}\geq Q_\phi$ at every rank, and both payoffs are unchanged.
\end{appendixlemma}

\begin{proof}
Replacing $\phi$ by $\phi^+$ weakly raises every upper-tail integral.
Iteration of $T_{R_{\phi^+}}$ from $Q_\phi$ therefore proves the
quantile order, as above. At any interior rank with $\phi(u)>0$,
monotonicity implies that the two profiles agree at that rank and
all higher ranks. Uniqueness of reconstruction on that tail gives
$Q_{\phi^+}(u)=Q_\phi(u)$. At interior ranks with nonpositive
virtual value, both profiles give zero quality and zero virtual
profit. The payoff formulas
therefore agree. Taking one-sided limits gives the quantile
comparisons at the endpoints.
\end{proof}

Let $M(\phi)=Q_\phi(0)$ be the minimum type induced by $\phi$, and put
$\Phi_{\underline\theta}=\{\phi\in\Phi:M(\phi)\geq\underline\theta\}$.
Regularization and truncation reduce the optimization problem to
\begin{equation*}
 J_\beta(\phi)=\int_0^1
 \left[\{\nu(Q_\phi(u))-\phi(u)\}q(\phi(u))
       +\beta\pi(\phi(u))\right]du,
 \qquad \phi\in\Phi_{\underline\theta}.
\end{equation*}
Regularization weakly raises consumer surplus and preserves profit,
so it weakly raises this objective for every real $\beta$.
Truncation preserves both payoffs. Conversely, every feasible profile
implements a feasible regular market. The next subsection establishes
the existence of a maximizing profile.

\subsection{Existence and optimality of the top market}
\label{app:rw-existence}

To prove existence, we show that the feasible profiles form a compact
set and that payoffs are continuous on that set. The main point is to
show that convergence of profiles in $L^1$ implies convergence of the
reconstructed types, including the lowest type. This requires a separate
estimate near rank one, where the reconstruction equation divides by
$1-u$.

\begin{appendixproposition}[Existence in the reduced problem]
\label{lem:gs-compact-continuous}
The sets $\Phi$ and $\Phi_{\underline\theta}$ are compact in $L^1$.
If $\phi_n\to\phi$ in $L^1$, then $Q_{\phi_n}\to Q_\phi$ in $L^1$
and at every $u<1$. In particular, $M(\phi_n)\to M(\phi)$, and both
payoff coordinates converge. For every $\beta\in\R$,
\begin{equation}
 \sup_{G\in\Delta(\Theta_{\underline\theta})}
       \{\CS(G)+\beta\Pi(G)\}
 =\max_{\phi\in\Phi_{\underline\theta}}J_\beta(\phi).
 \label{eq:gs-reduced-support-problem}
\end{equation}
\end{appendixproposition}

\begin{proof}
Helly's selection theorem gives a subsequence of any bounded sequence
of nondecreasing profiles that converges at the continuity points of
a nondecreasing limit \citep{billingsley1999}. Such a limit has at most
countably many discontinuities. Give the limit the one-sided values
required in the definition of $\Phi$. Bounded convergence then gives
convergence in $L^1$ to this element of $\Phi$, proving compactness.

To compare the reconstructed quantiles, write
$\Delta_n=Q_{\phi_n}-Q_\phi$ and $h_n=\phi_n-\phi$. The
mean-value formula gives
\[
 r_n(u):=\int_0^1\nu'\bigl(Q_\phi(u)+t\Delta_n(u)\bigr)\,dt
 \in[\underline r,1].
\]
Subtracting \eqref{eq:gs-Q-ode} for the two profiles yields
$(1-u)\Delta_n'=r_n\Delta_n-h_n$ almost everywhere. Multiplication
by $E_n(u)=\exp\{-\int_0^u r_n(t)/(1-t)\,dt\}$ gives
$(E_n\Delta_n)'=-E_nh_n/(1-u)$.
Since $E_n(u)\leq(1-u)^{\underline r}$ and $\Delta_n$ is bounded,
$E_n\Delta_n\to0$ at one. Integrating from $u$ to one therefore gives
\begin{equation}
 \Delta_n(u)=\int_u^1\frac{h_n(s)}{1-s}
       \exp\left\{-\int_u^s\frac{r_n(t)}{1-t}\,dt\right\}ds,
 \label{eq:gs-volterra-difference}
\end{equation}
\begin{equation*}
 |\Delta_n(u)|\leq(1-u)^{-\underline r}
       \int_u^1(1-s)^{\underline r-1}|h_n(s)|\,ds.
\end{equation*}
Fix $u<1$ and split the last integral at $1-\delta>u$. On
$[u,1-\delta]$ the weight is bounded, so this part tends to zero by
$L^1$ convergence. The contribution from ranks above $1-\delta$ is at most
$2\bar\nu(1-u)^{-\underline r}\delta^{\underline r}/\underline r$,
uniformly in $n$. First let $n\to\infty$ and then $\delta\downarrow0$.
This proves convergence of the quantiles at every $u<1$, including
zero; bounded convergence gives their $L^1$ convergence as well.

Convergence at zero makes the minimum-type constraint closed, so
$\Phi_{\underline\theta}$ is compact too. The functions $q$, $\pi$,
and $C\circ q$ are bounded and uniformly continuous on $[0,\bar\nu]$,
and $\nu$ has the same properties on $[0,\bar\theta]$. The payoff
formulas therefore converge under the $L^1$ limits just established.
Thus $J_\beta$ attains its maximum on $\Phi_{\underline\theta}$.
Regularization and truncation identify this maximum with the value
of the original problem, as asserted in
\eqref{eq:gs-reduced-support-problem}.
\end{proof}

The arguments about the surplus set below require continuity for
arbitrary markets. We establish that continuity next, together with
approximation by markets with finite support.

\begin{appendixproposition}[Continuity of market payoffs and finite approximation]
\label{lem:app-compact-density}
Weak convergence $G_n\Rightarrow G$ implies convergence of both payoffs;
hence $\mathcal V_{\underline\theta}$ is compact. Every feasible market
has finite-support approximations within the same type interval with
\begin{equation}
 \|Q_{G_n}-Q_G\|_\infty+\|Q_{G_n}^+-Q_G^+\|_\infty
       +\|(\phi_{G_n})^+-(\phi_G)^+\|_1\longrightarrow0.
 \label{eq:app-phi-L1}
\end{equation}
Every continuous function of the payoff pair attains its minimum and
maximum over all feasible markets; these values equal its infimum and
supremum over finite-support markets.
\end{appendixproposition}

\begin{proof}
We first show that uniformly small changes in types give uniformly
small changes in payoffs. We then prove continuity when only the
probabilities on a fixed finite grid change. Rounding types to such a
grid will give the result for arbitrary weakly convergent markets.

\emph{Step 1: uniformly close types have close payoffs.}
Represent two markets by type draws $X,Y$ on a common probability
space, with $|X-Y|\leq\delta$ almost surely, and let $G,H$ be their
distributions. Write $F_X$ and $F_Y$ for their cumulative distribution
functions.
The inequalities $F_X(t-\delta)\leq F_Y(t)\leq F_X(t+\delta)$ imply
that both versions of their quantiles differ by at most $\delta$.
Consequently, their raw revenue curves satisfy
$\|\widehat R_X-\widehat R_Y\|_\infty
 \leq(1+\operatorname{Lip}(m))\delta$.
The same bound applies to the least concave majorants because shifting either
majorant upward by the sup-norm difference between the raw curves
makes it a majorant of the other curve. To control the slopes that
determine quality, for concave $R$ with $R(1)=0$, let
\begin{equation*}
 R^0(u):=\max_{v\in[u,1]}R(v)
       =\int_u^1(-R'_+(s))^+\,ds.
\end{equation*}
The maximum replaces the rising part of $R$ by a constant and leaves
the falling part unchanged. Its negative derivative therefore retains the positive part of the virtual values. Taking this maximum
also cannot enlarge the sup-norm difference between two curves.

For nondecreasing $f:[0,1]\to[0,\bar\nu]$, the map
$f\mapsto\int_\cdot^1f$ is continuous from $L^1$ to the sup norm and
is one-to-one on $L^1$. Its inverse is uniformly continuous as well.
Indeed, otherwise there would be two sequences whose tail integrals
converge uniformly to each other while their $L^1$ distance stays
bounded away from zero. Write $f_n$ and $g_n$ for these sequences. By compactness,
we may pass to a subsequence along which $f_n\to f$ and
$g_n\to g$ in $L^1$. Passing to the limit in the difference
between their tail integrals gives
$\int_u^1 \bigl(f(s)-g(s)\bigr)\,ds=0$ for every $u\in[0,1]$. The limits must therefore be the
same, a contradiction. Applying this fact to $R^0$ shows that the
$L^1$ difference between the positive parts of the ironed virtual
values tends to zero uniformly as $\delta\to0$.

Negative virtual values generate neither quality nor virtual profit.
The functions $q$ and $\pi$ are bounded and uniformly continuous on
$[0,\bar\nu]$, as are $\nu$ on $[0,\bar\theta]$ and $C$ on
$[0,\bar x]$. Together with the quantile bound and the $L^1$ control
just obtained, this gives a function $\Omega(\delta)\downarrow0$,
independent of the two markets, such that
\begin{equation}
 |\Pi(G)-\Pi(H)|+|\TS(G)-\TS(H)|+|\CS(G)-\CS(H)|
 \leq\Omega(\delta).
 \label{eq:app-payoff-stability}
\end{equation}

\emph{Step 2: changing probabilities on a finite grid.}
\phantomsection\label{app:finite-market-preliminaries}
Fix grid points $\theta_1<\cdots<\theta_N$ and allow zero probabilities.
Write $S_i=\sum_{j\geq i}p_j$ and
$\Delta_i=\theta_{i+1}-\theta_i$. The finite analog of the envelope formula gives
\begin{equation*}
 \mathcal P(p,x)=\sum_{i=1}^Np_i\{\nu(\theta_i)x_i-C(x_i)\}
       -\sum_{i<N}S_{i+1}\Delta_i x_i,
 \qquad 0\leq x_1\leq\cdots\leq x_N\leq\bar x.
\end{equation*}
The second sum is the information rent: quality at type $i$
contributes $\Delta_i x_i$ to the rent of each higher type.

Hold the possible types $\theta_1,\ldots,\theta_N$ fixed, and let
the population share of type $\theta_i$ in the $n$th market be
$p_i^{(n)}$. Suppose that $p_i^{(n)}\to p_i$ for each $i$.
The finite-sum formula above implies that $\mathcal P(p^{(n)},x)$
converges to $\mathcal P(p,x)$ uniformly over
$0\leq x_1\leq\cdots\leq x_N\leq\bar x$; the bound on the profit
difference is independent of the allocation.

Uniform convergence of the profit functions implies convergence
of their maxima, so seller profit converges.
To establish convergence of total surplus, choose an optimal
allocation $x^{(n)}$ for each market and an optimal allocation
$x^*$ for the limiting market. Along any convergent subsequence
$x^{(n_k)}\to\hat x$, we can pass to the limit in
$\mathcal P(p^{(n_k)},x^{(n_k)})
 \geq \mathcal P(p^{(n_k)},y)$ for every feasible allocation $y$.
This gives $\mathcal P(p,\hat x)\geq\mathcal P(p,y)$,
so $\hat x$ is optimal for the limiting market.

Any two optimal allocations for the limiting market must
agree at every type with $p_i>0$. Otherwise, their average
would be feasible and, by strict convexity of $C$, would
yield strictly greater profit. Thus, every limit of a subsequence
agrees with $x^*$ at these types. Since the allocation
set is compact, it follows that $x_i^{(n)}\to x_i^*$
whenever $p_i>0$.
The contributions of these types to total surplus therefore
converge. Types with $p_i=0$ contribute vanishing amounts,
because $p_i^{(n)}\to0$ and qualities are bounded. Hence
total surplus converges. Subtracting seller profit then
gives convergence of consumer surplus.

\emph{Step 3: arbitrary weakly convergent markets.}
Let $G_n\Rightarrow G$. We approximate these markets on a
common finite grid so that Step 2 applies. Fix $\delta>0$
and partition the feasible type interval into subintervals
of length at most $\delta$, choosing the internal boundaries
at points to which $G$ assigns zero probability. Round each
type to the right endpoint of its subinterval, and denote
the resulting markets by $\widetilde G_n$ and $\widetilde G$. The boundary condition ensures that weak convergence implies
convergence of the probability in each subinterval. Thus
the rounded markets have a common finite set of types, and
the probability of each type converges to its probability
under $\widetilde G$. Their payoffs converge by Step 2.

Let $W$ denote any of $\Pi$, $\TS$, or $\CS$. Rounding
moves each type by at most $\delta$, so
\eqref{eq:app-payoff-stability} bounds each rounding error
by $\Omega(\delta)$. The triangle inequality therefore gives
$|W(G_n)-W(G)|
 \leq 2\Omega(\delta)
      +|W(\widetilde G_n)-W(\widetilde G)|$.
For fixed $\delta$, the last term tends to zero by Step 2.
Taking the limit superior as $n\to\infty$ and then letting
$\delta\downarrow0$ proves $W(G_n)\to W(G)$.

The set of distributions on the fixed feasible interval
is weakly compact \citep{billingsley1999}. Since the payoff
map is continuous, its image
$\mathcal V_{\underline\theta}$ is compact.

Finally, round a fixed market on grids whose mesh tends
to zero. The resulting markets have finite support within
the same feasible interval. Step 1 gives the uniform
quantile convergence and the $L^1$ convergence asserted
in \eqref{eq:app-phi-L1}, as well as convergence of payoffs.

Every continuous function of the payoff pair attains its
minimum and maximum on $\mathcal V_{\underline\theta}$.
Approximating a market attaining either extremum by
finite-support markets gives objective values converging
to that extremum. Hence the infimum and supremum over
finite-support markets equal the minimum and maximum
over all feasible markets.

\end{proof}

\begin{appendixproposition}[Existence and the top market]
\label{prop:rw-existence-top}
Under Assumptions~\ref{ass:main-nu}--\ref{ass:main-C}, the original and
reduced problems attain their maxima for every real $\beta$.
For $\beta\geq1$, the unique optimal market is $\delta_{\bar\theta}$.
For every finite $\beta<1$, that market is not optimal. In particular,
the optimized value is strictly positive when $0\leq\beta<1$.
\end{appendixproposition}

\begin{proof}
Existence follows from Proposition~\ref{lem:gs-compact-continuous}.
For $\beta\geq1$, individual rationality gives
$\Pi(G)\leq\TS(G)$, while first-best surplus per buyer is at most
$\pi(\bar\nu)$. Thus
\[
 \CS(G)+\beta\Pi(G)
 =\TS(G)+(\beta-1)\Pi(G)\leq\beta\pi(\bar\nu).
\]
The market concentrated at $\bar\theta$ attains this bound with
quality $\bar x$ and zero rent.
Equality requires total surplus equal to $\pi(\bar\nu)$. Since $\nu$ is strictly increasing and $C$ is strictly convex,
this requires type $\bar\theta$ and quality $\bar x$ almost surely. This proves uniqueness, including at $\beta=1$.

For $\beta<1$, consider a market $G_{\varepsilon,\lambda}$ that puts
mass $\lambda\in(0,1)$ on $\bar\theta-\varepsilon$ and the remaining
mass on $\bar\theta$, with
$0<\varepsilon<\bar\theta-\underline\theta$. The lower type earns zero
rent, and the top type earns $\varepsilon x_L$ when the lower quality
is $x_L$. The seller therefore maximizes
\[
 \lambda\{\nu(\bar\theta-\varepsilon)x_L-C(x_L)\}
 +(1-\lambda)\{\bar\nu x_H-C(x_H)-\varepsilon x_L\}.
\]
Its solution is $x_H=\bar x$, $x_L=q(z_\varepsilon)$, where
\[
 z_\varepsilon=\nu(\bar\theta-\varepsilon)
                  -\frac{1-\lambda}{\lambda}\varepsilon.
\]
For fixed $\lambda$ and sufficiently small $\varepsilon$,
$0<z_\varepsilon<\bar\nu$, so these qualities are monotone.
Using $\pi'=q$ and expanding at $\varepsilon=0$ gives
\[
 \begin{aligned}
 &\CS(G_{\varepsilon,\lambda})+\beta\Pi(G_{\varepsilon,\lambda})
       -\beta\pi(\bar\nu)\\
 &\qquad=\varepsilon\bar x
 \bigl[(1-\beta)(1-\lambda)-\beta\lambda\nu'(\bar\theta)\bigr]
       +o(\varepsilon).
 \end{aligned}
\]
The coefficient is positive for every $\lambda\in(0,1)$ when
$\beta<0$, and for sufficiently small $\lambda$ when $0\leq\beta<1$.
Choose such a $\lambda$ and then take $\varepsilon$ sufficiently
small. The two-type market strictly improves on the top market,
which is therefore not optimal. When $0\leq\beta<1$, the improvement
is over $\beta\pi(\bar\nu)\geq0$, so the optimized value is strictly
positive.
\end{proof}

\subsection{How a change in virtual net values changes welfare}
\label{app:common-perturbations}

Fix a profile $\phi$ and a bounded direction $h$. In the change
$\phi+\varepsilon h$, the function $h$ specifies where virtual values
rise or fall, while $\varepsilon$ determines the size of the change.
This changes both the seller's quality $q(\phi)$ and the types
$Q_\phi$ induced by the profile. We first compute the change in types
and then use the chain rule to obtain the change in welfare.

\paragraph{The change in types.}
Write $Z_{\phi,h}=DQ_\phi[h]$ for the derivative of the quantile
function in direction $h$.
Define
\[
 r_\phi(u)=\nu'(Q_\phi(u)),\qquad
 E_\phi(u)=\exp\left\{-\int_0^u\frac{r_\phi(s)}{1-s}\,ds\right\}.
\]
Here $r_\phi$ is the derivative of net value with respect to type,
evaluated at $Q_\phi$, and $E_\phi$ is the integrating factor for
the differential equation below.
The formulas below are for $u<1$; endpoint values of feasible
profiles are taken by one-sided limits.

\begin{appendixlemma}[The induced change in types]
\label{lem:gs-Q-M-derivative}
For a bounded direction $h$, the linear response is the unique bounded
solution of
\begin{equation}
 (1-u)Z_{\phi,h}'=r_\phi Z_{\phi,h}-h,
 \qquad
 Z_{\phi,h}(u)=\frac{1}{E_\phi(u)}
             \int_u^1\frac{E_\phi(s)}{1-s}h(s)\,ds.
 \label{eq:gs-linearized-Q}
\end{equation}
For $\phi,\psi\in\Phi$,
\begin{equation}
 \|Q_\psi-Q_\phi\|_\infty
 \leq\underline r^{-1}\|\psi-\phi\|_\infty,
 \label{eq:gs-Q-global-Lipschitz}
\end{equation}
\begin{equation}
 \|Q_\psi-Q_\phi-DQ_\phi[\psi-\phi]\|_\infty
 \leq \frac{\|m''\|_\infty}{2\underline r^{\,3}}
                    \|\psi-\phi\|_\infty^2.
 \label{eq:A-quantile-remainder}
\end{equation}
The response is linear in $h$ and varies continuously with the
profile. The derivative of $M(\phi)=Q_\phi(0)$ is
\begin{equation}
 DM(\phi)[h]=\int_0^1w_\phi(u)h(u)\,du,
 \quad w_\phi(u)=\frac{E_\phi(u)}{1-u},
 \quad 1\leq w_\phi(u)\leq(1-u)^{\underline r-1}.
 \label{eq:gs-DM}
\end{equation}
In particular $w_\phi$ is positive, integrable, and nondecreasing, with
\begin{equation}
 (1-u)w_\phi'=(1-r_\phi)w_\phi,\qquad w_\phi(0)=1.
 \label{eq:gs-w-ode}
\end{equation}
Both maps admit a continuously differentiable local extension to
bounded profiles, with a uniform quadratic remainder.
\end{appendixlemma}

\begin{proof}
Let $\mathcal A$ take the average of a function over each upper tail:
\begin{equation*}
 (\mathcal Af)(u)=\frac{1}{1-u}\int_u^1f(s)\,ds,
 \qquad Q_\phi=\mathcal A\{\phi+m(Q_\phi)\}.
\end{equation*}
Since averaging does not increase the sup norm,
$\|\mathcal Af\|_\infty\leq\|f\|_\infty$.
Subtracting the two reconstruction equations and using
$\operatorname{Lip}(m)\leq1-\underline r$ gives
\eqref{eq:gs-Q-global-Lipschitz}. For the derivative, consider the
equation $Z=\mathcal A\{h+m'(Q_\phi)Z\}$. Its right-hand side is a
contraction in $Z$ for the sup norm, with Lipschitz constant at most
$1-\underline r$.
It therefore has a unique bounded solution, which satisfies
$\|Z\|_\infty\leq\|h\|_\infty/\underline r$.

To verify that it is the derivative, put
$\Delta=Q_\psi-Q_\phi$ and $h=\psi-\phi$. Taylor's theorem gives
\[
 m(Q_\psi)-m(Q_\phi)=m'(Q_\phi)\Delta+e,
 \qquad \|e\|_\infty\leq\tfrac12\|m''\|_\infty\|\Delta\|_\infty^2.
\]
Subtract the equations for $\Delta$ and $Z$. Then
$\Delta-Z=\mathcal A\{m'(Q_\phi)(\Delta-Z)+e\}$, implying
$\|\Delta-Z\|_\infty\leq\|e\|_\infty/\underline r$.
Together with \eqref{eq:gs-Q-global-Lipschitz}, this proves
\eqref{eq:A-quantile-remainder}. To compare the derivatives at two
profiles, subtract their response equations. The same norm estimate gives
\begin{equation*}
 \|DQ_\psi[h]-DQ_\phi[h]\|_\infty
 \leq\frac{\|m''\|_\infty}{\underline r^{\,3}}
          \|\psi-\phi\|_\infty\|h\|_\infty.
\end{equation*}

Differentiating the linear integral equation on compact intervals
below one gives the differential equation in \eqref{eq:gs-linearized-Q}.
Apply the integrating-factor argument from
\eqref{eq:gs-volterra-difference}, with $r_\phi$ in place of $r_n$.
Since $Z$ is bounded, $E_\phi(u)Z(u)\to0$ as $u\uparrow1$;
this determines the constant of integration and gives the stated
integral solution. Evaluating at zero gives \eqref{eq:gs-DM}.
The bounds $1-u\leq E_\phi(u)\leq(1-u)^{\underline r}$ imply the
stated bounds on $w_\phi$, and differentiating the logarithm of
$w_\phi$ gives \eqref{eq:gs-w-ode}.

To define the derivatives on an open neighborhood, extend $m$ outside
$[0,\bar\theta]$ to a $C^2$ function with bounded second derivative
and Lipschitz constant $L<1$. Such an extension exists because the
absolute value of the original derivative is uniformly less than one.
The same contraction and Taylor arguments apply to bounded, possibly
nonmonotone profiles, with $1-L$ in place of $\underline r$.
The linear response therefore approximates every small increment
with a uniform quadratic remainder, and depends continuously on the
profile. This proves continuous Fr\'echet differentiability.
Evaluating the reconstructed quantile at zero gives the same conclusion
for $M$.
Taking the supremum over $\|h\|_\infty\leq1$ in the difference of
\eqref{eq:gs-DM} also gives local $L^1$ continuity of $w_\phi$.
\end{proof}

\paragraph{The chain rule for welfare.}
At a fixed rank write $Q=Q_\phi$, $x=q(\phi)$ and
$g_\phi=\nu(Q_\phi)-\phi$.
The objective density is
$F(Q,\phi)=[\nu(Q)-\phi]q(\phi)+\beta\pi(\phi)$.
For changes supported where $\phi$ is bounded away from zero, products
with $h$ below are set to zero outside its support. The two partial
derivatives on that positive range are
\[
 F_Q=\nu'(Q)q(\phi),\qquad
 F_\phi=[\nu(Q)-\phi]q'(\phi)+(\beta-1)q(\phi).
\]
The second equality uses $\pi'=q$.
Thus the first-order change has the form
\begin{equation*}
 DJ_\beta(\phi)[h]
 =\int_0^1 r_\phi x_\phi DQ_\phi[h]\,du
  +\int_0^1\{g_\phi q'(\phi)+(\beta-1)x_\phi\}h\,du.
\end{equation*}
The first integral is the welfare effect of the change in actual net
values, holding qualities at their original levels.

Substitute \eqref{eq:gs-linearized-Q} in the first integral and reverse
the order over $0\leq u\leq s\leq1$. Its coefficient on $h(s)$ is
\begin{equation}
 A_\phi(s):=w_\phi(s)\int_0^s
                 \frac{r_\phi(u)x_\phi(u)}{E_\phi(u)}\,du,
 \qquad \int_0^1A_\phi(u)\,du\leq\frac{\bar x}{\underline r}.
 \label{eq:gs-A-explicit}
\end{equation}
Thus $A_\phi(s)$ adds the welfare effects of changes in types at
lower ranks when virtual value changes near $s$. Differentiating this expression gives
\begin{equation}
 (1-u)A_\phi'=(1-r_\phi)A_\phi+r_\phi x_\phi,
 \qquad A_\phi(0)=0.
 \label{eq:gs-A}
\end{equation}
Under private values these objects reduce to
\[
 DQ_\phi[h](u)=\frac{1}{1-u}\int_u^1h(s)\,ds,
 \qquad w_\phi\equiv1,
 \qquad A_\phi(u)=\int_0^u\frac{x_\phi(s)}{1-s}\,ds.
\]
A change at a higher rank affects a larger set of lower ranks. With
type-dependent costs, the resulting changes in types also change
costs; the kernel in the reconstruction formula accounts for this
further effect.

The next lemma states the uniform error bound needed for feasible
perturbations. We restrict changes to positive virtual values because
$q'$ need not be bounded near zero.

\begin{appendixlemma}[First-order welfare effects]
\label{lem:gs-first-variation}
Fix a finite $\beta$. On positive virtual values define
\begin{equation*}
 H_\beta[\phi]=A_\phi+g_\phi q'(\phi)+(\beta-1)q(\phi).
\end{equation*}
Let $S$ be a measurable set on which $\phi\geq2\delta>0$.
There are a sup-norm neighborhood $\mathcal N$ of $\phi$ and a modulus
$\omega_\delta(t)\to0$ such that, for $\psi\in\mathcal N\cap\Phi$
and bounded $h$ supported on $S$, with $\|h\|_\infty\leq\delta/2$
and $0\leq\psi+h\leq\bar\nu$,
\begin{equation}
 \begin{aligned}
 J_\beta(\psi+h)-J_\beta(\psi)
 &=\int_0^1H_\beta[\psi](u)h(u)\,du+R_J(\psi,h),\\
 |R_J(\psi,h)|&\leq\omega_\delta(\|h\|_\infty)\|h\|_\infty.
 \end{aligned}
 \label{eq:gs-uniform-first-variation}
\end{equation}
\end{appendixlemma}

\begin{proof}
Choose $\mathcal N$ so that the profile values on $S$ remain in a
compact subset of $(0,\bar\nu]$. On this range, $q$ and $\pi$ have
uniformly continuous one-sided derivatives. Put $Z=DQ_\psi[h]$.
The quantile estimate and Taylor's theorem give
$\nu(Q_{\psi+h})-\nu(Q_\psi)=r_\psi Z+O(\|h\|_\infty^2)$,
uniformly for $\psi\in\mathcal N\cap\Phi$. Expanding the payoff density gives
\[
 J_\beta(\psi+h)-J_\beta(\psi)
 =\int_0^1r_\psi x_\psi Z\,du
  +\int_S\{g_\psi q'(\psi)+(\beta-1)x_\psi\}h\,du+R_J,
\]
with the stated remainder. Outside $S$, quality does not change, but
types may change. Their effect is included in the first integral,
which requires no derivative of $q$ at zero.

To justify the change in the order of integration, first take
$h\equiv1$ in the equation for the quantile derivative. The kernel
in \eqref{eq:gs-linearized-Q} is nonnegative, so Tonelli's theorem gives
\[
 \int_0^1A_\phi(u)\,du
 =\int_0^1r_\phi(u)x_\phi(u)Z_{\phi,\mathbf1}(u)\,du
 \leq\frac{\bar x}{\underline r}.
\]
For a general bounded direction $h$, which may take both signs,
the same calculation bounds the double integral of the absolute
value by $\|h\|_\infty\int_0^1A_\phi<\infty$.
Fubini's theorem therefore gives
$\int r_\psi x_\psi DQ_\psi[h]=\int A_\psi h$ and proves
\eqref{eq:gs-uniform-first-variation}. Differentiating
\eqref{eq:gs-A-explicit} on compact intervals proves \eqref{eq:gs-A}.
\end{proof}

\subsection{The structure of an optimal market}
\label{app:qualitative-structure}

Fix $\beta\in[0,1)$ and an optimizer
$\phi\in\Phi_{\underline\theta}$. Its value is positive by
Proposition~\ref{prop:rw-existence-top}. We show that its virtual
values are positive at every positive rank and reach their upper
bound on an interval ending at one. At lower ranks, we rule out
positive flat intervals by constructing a profitable separation of
their types. Throughout this subsection we maintain
Assumptions \ref{ass:main-nu} and \ref{ass:main-C}.

\paragraph{No positive mass of buyers is excluded.}
Suppose $\phi=0$ on an initial interval of length $a>0$. Since the
objective is positive, $a<1$. Remove these ranks and rescale the
remaining ones:
\begin{equation*}
 \widehat\phi(v)=\phi(a+(1-a)v),
 \qquad Q_{\widehat\phi}(v)=Q_\phi(a+(1-a)v)\geq Q_\phi(v).
\end{equation*}
The quantile identity follows by changing variables in the
reconstruction equation and using uniqueness on the retained tail.
The lowest type weakly rises, so the new market is feasible. The
removed ranks do not contribute to either payoff, since $q(0)=\pi(0)=0$.
Consequently
\begin{equation}
 (\CS,\Pi,J_\beta)(\widehat\phi)
 =\frac{(\CS,\Pi,J_\beta)(\phi)}{1-a}.
 \label{eq:paper-app-active-payoff-scaling}
\end{equation}
The original objective is positive, so this change strictly improves
it. Hence $\phi(u)>0$ for every $u>0$. This does not yet establish
$\phi(0)>0$: the positive virtual values could approach zero as
rank decreases to zero.

\paragraph{An atom at the highest type.}
For a nonzero profile, $A_\phi$ becomes arbitrarily large near rank
one, while the other terms in $H_\beta[\phi]$ remain bounded there. An
increase in virtual values at high ranks is therefore profitable
unless the profile is already at its upper bound on an interval
ending at one. The following argument applies to every finite $\beta$.

\begin{appendixlemma}[A profitable increase near the highest rank]
\label{lem:rw-top-atom}
Fix a finite $\beta$ and a nonzero $\phi\in\Phi_{\underline\theta}$.
If $\phi$ has no terminal interval on which it equals $\bar\nu$, there
are arbitrarily small monotone pointwise upward changes that strictly
increase $J_\beta$. Their induced quantiles weakly rise everywhere, so
every minimum-type restriction is preserved.
\end{appendixlemma}

\begin{proof}
A nonzero nondecreasing profile has $u_0>0$ with
$x_0:=\inf_{u\geq u_0}q(\phi(u))>0$.
Since $A_\phi\geq0$ and $\underline r\leq r_\phi\leq1$,
\eqref{eq:gs-A} gives
\[
 A_\phi'(u)\geq\frac{\underline r x_0}{1-u}\quad(u\geq u_0),
 \qquad A_\phi(u)\longrightarrow\infty\quad(u\uparrow1).
\]
The other terms in $H_\beta[\phi]$ are bounded on this tail, since
virtual values there are bounded away from zero. Choose $t>u_0$ and
$\kappa>0$ so that $H_\beta[\phi]\geq2\kappa$ on $[t,1)$.
Under the hypothesis of the lemma, fix $s\in(t,1)$
and choose $0<\varepsilon<\bar\nu-\phi(s)$. Set
\[
 \phi_\varepsilon(u)=
 \begin{cases}
 \phi(u),&u<t,\\
 \min\{\bar\nu,\phi(u)+\varepsilon\},&u\geq t.
 \end{cases}
\]
The new profile is nondecreasing and remains within its bounds.
Since it weakly exceeds the original profile at every rank,
Lemma~\ref{lem:attenuation-tail-order} implies that the minimum type
weakly rises. The increment equals
$\varepsilon$ throughout $[t,s]$ and is nonnegative elsewhere.
The uniform welfare expansion therefore gives
$J_\beta(\phi_\varepsilon)-J_\beta(\phi)
\geq2\kappa\varepsilon(s-t)+o(\varepsilon)>0$ for small $\varepsilon$.
\end{proof}

The interval on which the profile equals its upper bound cannot
be the whole market, by Proposition~\ref{prop:rw-existence-top}.
Let $b\in(0,1)$ be its first rank.

\begin{appendixlemma}[Connected support and the top atom]
\label{lem:gs-strict-wedge}
Let $\phi\in\Phi$ satisfy $\phi<\bar\nu$ on $[0,b)$ and
$\phi=\bar\nu$ on $[b,1]$, with $0<b<1$. Then
$\nu(Q_\phi(u))-\phi(u)>0$ for every $u<b$.
The quantile is continuous and strictly increasing on $[0,b)$ and
is $\bar\theta$ on $[b,1]$. Its market has connected support
$[Q_\phi(0),\bar\theta]$, no atom below $\bar\theta$, and an atom
of mass $1-b$ at $\bar\theta$, served at quality $\bar x$.
\end{appendixlemma}

\begin{proof}
On $[b,1]$, the constant quantile $\bar\theta$ uniquely solves the
reconstruction equation. Put
$f=\phi+m(Q_\phi)$. This is nondecreasing, strictly below $\bar\theta$
at every rank below $b$, and equal to $\bar\theta$ on $[b,1]$.
Because the interval $[b,1]$ has positive length, the upper-tail
average of $f$ strictly exceeds $f(u)$ at every $u<b$. That average
is $Q_\phi(u)$, so
\[
 \nu(Q_\phi(u))-\phi(u)=Q_\phi(u)-f(u)>0.
\]
Local absolute continuity and \eqref{eq:gs-Q-ode} imply that $Q_\phi$
is strictly increasing below $b$. Continuity excludes gaps in the
support, and strict increase excludes atoms below the highest type.
The constant quantile on $[b,1]$ represents the atom at the highest
type, whose virtual value $\bar\nu$ induces quality $\bar x$.
\end{proof}

Different ranks below $b$ thus represent different types. A positive
flat interval of virtual values there would assign the same quality
to distinct types. On the interval representing the top atom, the
buyers are identical and receive the same quality.

\paragraph{Separating a bunched interval.}
On a flat interval, $H_\beta[\phi]$ increases with rank. This suggests
lowering virtual value on the lower part and raising it on the
upper part. To preserve monotonicity, the changes must also respect
the profile values on either side of the interval. The minimum and
maximum in the construction below ensure this. We choose the split
so that the total profile change is zero; the resulting upper-tail
inequalities then preserve the minimum actual type.

\begin{appendixlemma}[A feasible separation of a bunched interval]
\label{lem:attenuation-no-bunching-tilt}
Fix a finite $\beta$ and a feasible profile with the terminal interval
in Lemma~\ref{lem:gs-strict-wedge}. Suppose $(\ell,r)$ is a nondegenerate
maximal flat interval at a level $\gamma\in(0,\bar\nu)$, with
$0\leq\ell<r\leq b$. There are monotone feasible profiles
$\phi_\varepsilon$ satisfying $\|\phi_\varepsilon-\phi\|_\infty\leq
\varepsilon$ and
\begin{equation}
 \int_0^1(\phi_\varepsilon-\phi)\,du=0,
 \qquad
 \int_u^1(\phi_\varepsilon-\phi)\,ds\geq0
 \quad\text{for every }u,
 \label{eq:attenuation-balanced-tail}
\end{equation}
such that $J_\beta(\phi_\varepsilon)>J_\beta(\phi)$ for all sufficiently
small $\varepsilon>0$. Every induced type weakly rises.
\end{appendixlemma}

\begin{proof}
For a proposed split $c\in[\ell,r]$, define
\begin{equation*}
 \phi_{\varepsilon,c}(u)=
 \begin{cases}
 \min\{\phi(u),\gamma-\varepsilon\},&u<c,\\
 \max\{\phi(u),\gamma+\varepsilon\},&u\geq c,
 \end{cases}
 \qquad 0<\varepsilon<\min\{\gamma/2,\bar\nu-\gamma\}.
\end{equation*}
On the original flat interval, this gives the two levels
$\gamma-\varepsilon$ and $\gamma+\varepsilon$. On each side the
minimum or maximum of the original nondecreasing profile with a
constant remains nondecreasing. At $c$ there is an upward jump,
and the profile bounds are preserved.

Outside the flat interval, changes occur only at values in
$(\gamma-\varepsilon,\gamma)$ on the left or
$(\gamma,\gamma+\varepsilon)$ on the right. Thus only ranks whose
original virtual values are within $\varepsilon$ of the flat level
are affected. By maximality of the flat interval, the total length
of the affected intervals outside it tends to zero as
$\varepsilon\downarrow0$.
The change at each affected rank is at most $\varepsilon$ in absolute
value. The integral of the absolute change outside the flat interval is therefore
bounded by $\varepsilon$ times a quantity tending to zero, and hence is
$o(\varepsilon)$.
Let $R_\varepsilon$ be the signed integral of these changes outside
the flat interval. It does not depend on $c$ because at ranks below $\ell$
the construction always uses the minimum, and at ranks above $r$
it always uses the maximum. The changes occur on
a fixed compact interval below rank one where $\phi\geq\gamma/2$.
If $\ell=0$, there is no change to the left. If $r=b$, the top
interval remains at $\bar\nu>\gamma+\varepsilon$, so there is no
change to the right.

The total profile change is
$\varepsilon(\ell+r-2c)+R_\varepsilon$. To make it zero, choose
\[
 c_\varepsilon=\frac{\ell+r}{2}+\frac{R_\varepsilon}{2\varepsilon}.
\]
Since $R_\varepsilon=o(\varepsilon)$, this split lies inside the
flat interval for small $\varepsilon$ and converges to its midpoint
$c_0=(\ell+r)/2$. Put $\phi_\varepsilon=\phi_{\varepsilon,c_\varepsilon}$
and $\Delta_\varepsilon=\phi_\varepsilon-\phi$.
The function $\Delta_\varepsilon$ has integral zero, is nonpositive
below $c_\varepsilon$, and is nonnegative above $c_\varepsilon$.
Hence, for $u<c_\varepsilon$,
\[
 \int_u^1\Delta_\varepsilon(s)\,ds
 =-\int_0^u\Delta_\varepsilon(s)\,ds\geq0;
\]
for $u\geq c_\varepsilon$ the same inequality follows directly.
Lemma~\ref{lem:attenuation-tail-order} now gives
$Q_{\phi_\varepsilon}\geq Q_\phi$, proving feasibility even if
the minimum-type constraint binds.

To compare welfare, write $g=\nu(Q_\phi)-\gamma>0$ on the flat interval.
Equations~\eqref{eq:gs-Q-ode} and \eqref{eq:gs-A} give
$(1-u)g'=r_\phi g$ and
\begin{equation*}
 (1-u)H_\beta'[\phi]
 =(1-r_\phi)A_\phi+r_\phi q(\gamma)
                +r_\phi gq'(\gamma)>0
 \quad\text{a.e. on }(\ell,r).
\end{equation*}
The rank dependence of $H_\beta[\phi]$ on this interval comes from
$A_\phi$ and $gq'(\gamma)$. The remaining term
$(\beta-1)q(\gamma)$ is constant and cancels when integrated over
the two equal-length halves. Since $H_\beta[\phi]$ is strictly
increasing, we have
\[
 D_J:=\int_{c_0}^{r}H_\beta[\phi](u)\,du
          -\int_{\ell}^{c_0}H_\beta[\phi](u)\,du>0.
\]
The changes outside the flat interval contribute $o(\varepsilon)$
to the integral of $H_\beta[\phi]\Delta_\varepsilon$, since
$H_\beta[\phi]$ is bounded on the fixed compact interval where these
changes occur and the profile is bounded away from zero. Moving
the split from $c_0$ to $c_\varepsilon$ changes that integral by at
most a constant times
$\varepsilon|c_\varepsilon-c_0|=o(\varepsilon)$.
The uniform first variation therefore gives
\begin{equation*}
 J_\beta(\phi_\varepsilon)-J_\beta(\phi)
 =\varepsilon D_J+o(\varepsilon)>0.
\end{equation*}
Thus a positive flat interval of this kind cannot occur in an optimizer.
\end{proof}

\begin{appendixproposition}[Structure under Assumptions 1--2]
\label{prop:app-qualitative-structure}
For $\beta\in[0,1)$, every reduced optimizer has positive value and,
for some $b\in(0,1)$,
\begin{equation}
 0<\phi(u)<\bar\nu\quad(0<u<b),\qquad
 \phi(u)=\bar\nu\quad(b\leq u\leq1),
 \label{eq:paper-app-terminal-profile}
\end{equation}
with $\phi$ strictly increasing on $(0,b)$.
Every original-market optimizer equals its regularization. Its support
is connected, it has no atom below the highest type, it has a positive
efficiently served highest-type atom, no positive mass is excluded,
and distinct interior types receive strictly different qualities.
\end{appendixproposition}

\begin{proof}
The preceding arguments rule out an initial interval of zero virtual
values, establish an atom at the highest type, and rule out every
positive flat interval below it. A nondecreasing function that is not
strictly increasing is constant on
an interval of positive length. Thus
\eqref{eq:paper-app-terminal-profile} and strict increase follow.

Now let $G$ solve the original market problem. Regularization preserves
profit and weakly raises consumer surplus, so optimality forces equality
in \eqref{eq:gs-reg-CS}. Truncation of the regularized profile gives a
reduced optimizer. Its positive virtual values at every positive rank
imply $q(\phi_G)>0$ almost everywhere. The equality condition
\eqref{eq:gs-reg-equality} therefore gives
$Q_G=Q_{\widetilde G}$ almost everywhere, identifying the same
distribution. Thus, $G=\widetilde G$. The original and truncated profiles
agree except possibly at rank zero.
The support and atom claims then follow by Lemma~\ref{lem:gs-strict-wedge}. Since both $Q_\phi$ and
$q(\phi)$ are strictly increasing below $b$, distinct realized
types below the highest type receive distinct qualities.
\end{proof}

This proves the qualitative part of Theorem~\ref{thm:main-structure}.
It remains to identify the optimal separating profile.

\subsection{First-order conditions and boundary behavior}
\label{app:free-boundary-characterization}

We derive the first-order condition using small changes
in the separating profile that preserve monotonicity.
If the minimum-type constraint is slack, changes of
either sign are feasible, so the derivative of $J_\beta$
must vanish. If the constraint binds, we adjust the size
of a second perturbation to keep the minimum type fixed.
The same argument then applies to the combined change.

\paragraph{Changes that preserve monotonicity.}
Choose $0<a<c<b$ and a continuously differentiable function $f$ that
vanishes outside a compact subinterval of $(\phi(a),\phi(c))$.
Consider
\begin{equation}
 T_\varepsilon(z)=z+\varepsilon f(z),
 \qquad \phi_\varepsilon(u)=T_\varepsilon(\phi(u))
                          =\phi(u)+\varepsilon f(\phi(u)).
 \label{eq:A-monotone-transformation}
\end{equation}
If $|\varepsilon|\|f'\|_\infty<1$, then $T_\varepsilon$ is strictly
increasing. It also fixes zero and $\bar\nu$. Thus composition with
$\phi$ preserves both monotonicity and the profile bounds. The change
is supported inside $[a,c]$, where virtual values are bounded away
from zero.

These transformations allow us to test integral identities on every
interior rank interval.
For fixed $a<c$, take smooth functions $0\leq f_n\leq1$, with compact
support in $(\phi(a),\phi(c))$, converging to one at every value inside
that interval. Strict increase of $\phi$ implies
$f_n(\phi(u))\to\mathbf1_{(a,c)}(u)$ except possibly at the two
endpoints. By bounded convergence, an integral identity that holds for
all such tests therefore holds on every compact interior rank interval.

\begin{appendixlemma}[The minimum-type constraint and the first-order condition]
\label{lem:minimum-type-first-order-condition}
Fix a finite $\beta\in\R$. Let $\phi$ maximize $J_\beta$ over
$\Phi_{\underline\theta}$ and be strictly increasing with
$0<\phi<\bar\nu$ on $(0,b)$, where $0<b<1$. There is $\xi\geq0$ such that
\begin{equation}
 H_\beta[\phi](u)+\xi w_\phi(u)=0
 \quad\text{a.e. on }(0,b),
 \qquad \xi\{M(\phi)-\underline\theta\}=0.
 \label{eq:gs-Euler}
\end{equation}
\end{appendixlemma}

\begin{proof}
Write $H=H_\beta[\phi]$ and $w=w_\phi$.
Suppose first that $M(\phi)>\underline\theta$. By continuity of $M$,
the perturbations in \eqref{eq:A-monotone-transformation} are feasible
for all sufficiently small positive and negative $\varepsilon$.
Optimality and the first-variation formula then give
$\int_0^1H(u)f(\phi(u))\,du=0$ for every test $f$ described above.
Approximating interval indicators as above yields $\int_a^cH=0$
for every $0<a<c<b$. The function $H$ is continuous at every
continuity point of $\phi$. If it were nonzero at such a point, it
would have one strict sign on a neighboring interval, contradicting
the zero integral on that interval. Since a monotone $\phi$ has at
most countably many discontinuities, $H=0$ almost everywhere.
We may therefore take $\xi=0$.

Now suppose $M(\phi)=\underline\theta$.
Choose a nonnegative smooth test function $f_0$ whose composition
$h_0=f_0\circ\phi$ is positive on a rank interval.
To find such a function, take a continuity point of the strictly
increasing profile in $(0,b)$ and a smooth function positive near
its virtual value. Positivity of $w$ then gives
\[
 W_0:=\int_0^1w(u)h_0(u)\,du>0.
\]
For any other test function $f$, put $h=f\circ\phi$ and consider
\begin{equation*}
 \phi_{\varepsilon,s}(u)
 =\phi(u)+\varepsilon h(u)+s h_0(u).
\end{equation*}
For all $(\varepsilon,s)$ in a sufficiently small open rectangle
about $(0,0)$, the derivative of the transformation applied to
$\phi$ is $1+\varepsilon f'+s f_0'>0$. The transformation is therefore
increasing, and all these profiles respect their bounds.

By Lemma~\ref{lem:gs-Q-M-derivative},
$(\varepsilon,s)\mapsto M(\phi_{\varepsilon,s})$ is a continuously
differentiable function of two real variables. Its derivative with
respect to $s$ at the origin is $W_0>0$, so the implicit function
theorem gives a continuously differentiable $s(\varepsilon)$ such that
\[
 s(0)=0,\qquad M(\phi_{\varepsilon,s(\varepsilon)})=M(\phi),
 \qquad s'(0)=-\frac{\int_0^1w h}{W_0}.
\]
Along this path the minimum type is unchanged, so both signs of
$\varepsilon$ are feasible. Optimality and the uniform first-variation
formula imply
\begin{equation*}
 \int_0^1Hh
 =\frac{\int_0^1Hh_0}{W_0}\int_0^1wh.
\end{equation*}
Let $\zeta=\int Hh_0/W_0$. Applying the interval approximation to
$H-\zeta w$ gives zero integral on every interior interval. This
function is continuous wherever $\phi$ is, so the argument in the
first paragraph gives $H=\zeta w$ almost everywhere. To determine
the sign of the coefficient, consider the positive direction $h_0$
without compensation. It raises the profile pointwise and preserves
the minimum-type restriction. Optimality therefore gives
$\int Hh_0\leq0$, and hence $\zeta\leq0$. Taking $\xi=-\zeta$
proves \eqref{eq:gs-Euler}; complementary slackness holds because
the constraint binds.
\end{proof}

\paragraph{From marginal welfare to a state--control equation.}
Set $\alpha=1-\beta$, $y=\nu(Q_\phi)$, $x=q(\phi)$, and
$\rho(y)=\nu'(\nu^{-1}(y))$. Thus, $r_\phi(u)=\rho(y(u))$, and
the two symbols describe the same net-value slope, evaluated at a
rank and at a net value, respectively. Introduce $B$ to combine
the effect $A_\phi$ of changing induced types with the constraint
term $\xi w_\phi$:
\begin{equation}
 B=A_\phi+\xi w_\phi,
 \qquad (1-u)B'=(1-\rho(y))B+\rho(y)x,
 \qquad B(0)=\xi.
 \label{eq:gs-B}
\end{equation}
The differential identity follows by adding \eqref{eq:gs-A} and
$\xi$ times \eqref{eq:gs-w-ode}. Since $\phi=C'(x)$ and
$q'(\phi)=1/C''(x)$, multiplying \eqref{eq:gs-Euler} by $C''(x)$ gives
\begin{equation}
 y=C'(x)+C''(x)(\alpha x-B).
 \label{eq:gs-state-control}
\end{equation}
Below the top interval, Lemma~\ref{lem:gs-strict-wedge} makes
$y-C'(x)>0$, so $0\leq B<\alpha x$. In particular, a nondegenerate
optimizer of this structure requires $\alpha>0$.
The next lemma shows that, for a fixed state $(B,y)$, this equation
can have only one admissible quality.

\begin{appendixlemma}[Uniqueness of quality near zero]
\label{lem:gs-state-control-monotone}
Under Assumption~\ref{ass:main-curvature}, for $0<\alpha\leq1$ the
right side of \eqref{eq:gs-state-control} is strictly increasing in
$x>0$ wherever $0\leq B\leq\alpha x$. Moreover,
\begin{equation*}
 \lim_{x\downarrow0}xC''(x)=0,
 \qquad \lim_{z\downarrow0}\frac{q(z)}{q'(z)}=0.
\end{equation*}
\end{appendixlemma}

\begin{proof}
Fix $B$. The derivative of the right side with respect to quality is
\begin{equation}
 D_\alpha(x,B)=(1+\alpha)C''(x)+(\alpha x-B)C'''(x).
 \label{eq:gs-control-derivative}
\end{equation}
It is positive if $C'''\geq0$. Otherwise, since
$0\leq\alpha x-B\leq\alpha x$, the curvature restriction gives
\[
 \frac{D_\alpha(x,B)}{C''(x)}
 \geq1+\alpha-\alpha\left(-\frac{xC'''(x)}{C''(x)}\right)
 >1-\alpha\geq0.
\]
Thus the right side is strictly increasing over admissible qualities,
so two such qualities cannot solve the equation at the same state.

For the limit at zero, fix $\kappa>2$. The curvature bound implies
$(\log C'')'>-\kappa/x$. For $s\in[x,2x]$ with $2x\leq\bar x$,
integration gives $C''(s)\geq2^{-\kappa}C''(x)$, hence
\[
 0\leq2^{-\kappa}xC''(x)\leq C'(2x)-C'(x)\longrightarrow0.
\]
Finally, $z=C'(x)$ and $q'(z)=1/C''(x)$ give
$q(z)/q'(z)=xC''(x)$.
\end{proof}

For $0<\alpha\leq1$, the preceding lemma supplies the cost conditions
needed below. To use the characterization also for
negative profit weights, we give the remaining arguments for any
fixed $\alpha=1-\beta>0$ satisfying
\begin{equation}
 \begin{gathered}
 C\in C^3((0,\bar x]),\qquad \lim_{x\downarrow0}xC''(x)=0,\\
 D_\alpha(x,B)>0
 \quad(0<x\leq\bar x,\ 0\leq B\leq\alpha x).
 \end{gathered}
 \label{eq:rw-characterization-conditions}
\end{equation}
In the following result, we assume that an optimizer has
positive, strictly increasing $\phi$ on $(0,b)$ and $\phi=\bar\nu$ on
$[b,1]$, with $0<b<1$. Subsection~\ref{app:negative-frontier} verifies
that structure separately before applying the result.

\begin{appendixproposition}[Endpoint behavior and regularity]
\label{prop:rw-endpoints}
Under \eqref{eq:rw-characterization-conditions}, let $\phi$ be an
optimizer with the structure just specified. Then $\phi(0)>0$.
The functions $Q,y,B,x,\phi$ are absolutely continuous on $[0,b]$,
and the separating allocation meets the top allocation continuously:
\[
 Q(b-)=\bar\theta,\qquad \phi(b-)=\bar\nu,
 \qquad x(b-)=\bar x,\qquad B(b)=\alpha\bar x.
\]
\end{appendixproposition}

\begin{proof}
\emph{Quality at the lowest type.}
The reconstruction equation and \eqref{eq:gs-B} imply that $Q,y,B$
are absolutely continuous on $[0,b]$.
The top atom gives
\[
 Q(0)=\int_0^1\{\phi(u)+m(Q(u))\}\,du
 \geq(1-b)\bar\theta>0.
\]
If $\phi(0)=0$, monotonicity and right continuity imply $x(u)\to0$
as $u\downarrow0$. Yet the first-order condition gives
\begin{equation*}
 0<y-\phi=C''(x)(\alpha x-B)\leq\alpha xC''(x)
 \quad\text{a.e. on }(0,b).
\end{equation*}
The left side tends to $\nu(Q(0))>0$, whereas the right side tends to
zero by \eqref{eq:rw-characterization-conditions}. This contradiction
proves
$\phi(0)>0$ and $x(0)>0$.
Here the values at zero are the one-sided limits specified by the
pointwise allocation convention above.

\emph{Continuity at the top interval.}
Let $p=\phi(b-)>0$. Since $Q(b-)=\bar\theta$, taking one-sided
interior limits in the first-order condition gives
\begin{equation}
 B(b)=\alpha q(p)-(\bar\nu-p)q'(p).
 \label{eq:gs-pre-pasting-B}
\end{equation}
If $p<\bar\nu$, then $B(b)<\alpha\bar x$.
Let $K=H_\beta[\phi]+\xi w_\phi$. The first-order condition makes
it zero almost everywhere below $b$, while on the top interval it
equals $B-\alpha\bar x$. By continuity of $B$, we therefore have
$K<0$ on a short interval
$I_T=[b,b+\delta)$, with $b+\delta<1$.
We will lower virtual values on this interval and compensate with
an interior change that keeps the minimum type fixed. The resulting
welfare change will be determined by the negative values of $K$.

Choose the nonnegative interior direction $h_0=f_0\circ\phi$ used
in the Euler argument, and put
$W_0=\int_0^1w_\phi h_0>0$,
$W_T=\int_{I_T}w_\phi>0$. For $\varepsilon>0$, consider
\[
 \psi_{\varepsilon,s}
 =\phi-\varepsilon\mathbf1_{I_T}+s h_0.
\]
For sufficiently small $\varepsilon$ and $|s|$, the profile is
nondecreasing and respects its bounds. Below $b$, the interior
change is an increasing transformation. At $b$, the gap
$\bar\nu-p$ permits the decrease by $\varepsilon$, and at
$b+\delta$ the profile jumps upward to $\bar\nu$. The direction
$h_0$ vanishes near both boundaries.

The quadratic remainder for the quantile at zero gives, uniformly
for $s=O(\varepsilon)$,
\[
 M(\psi_{\varepsilon,s})-M(\phi)
 =-\varepsilon W_T+sW_0+O(\varepsilon^2).
\]
At $s=\varepsilon W_T/(2W_0)$ this is negative, and at
$s=2\varepsilon W_T/W_0$ it is positive, for small $\varepsilon$.
Continuity in $s$ gives a value $s_\varepsilon$ between them with
$M(\psi_{\varepsilon,s_\varepsilon})=M(\phi)$ and $s_\varepsilon=\varepsilon W_T/W_0+O(\varepsilon^2)$.
This profile preserves the minimum type, so it is feasible
whether the original minimum-type constraint binds or is slack.
Since $\int H_\beta[\phi]h_0=-\xi W_0$, the welfare change is
\[
J_\beta(\psi_{\varepsilon,s_\varepsilon})-J_\beta(\phi)
 =-\varepsilon\int_{I_T}H_\beta[\phi]
        +s_\varepsilon\int_0^1H_\beta[\phi]h_0+o(\varepsilon)
 =-\varepsilon\int_{I_T}K+o(\varepsilon)>0.
\]
The strict inequality follows from $K<0$ on $I_T$, and contradicts
optimality. Hence $p=\bar\nu$, and
\eqref{eq:gs-pre-pasting-B} gives $B(b)=\alpha\bar x$.

\emph{Absolute continuity of the allocation.}
We already know that $B$ and $y$ are absolutely continuous on $[0,b]$,
that $x$ is nondecreasing, and that $x(0)>0$. To prove absolute
continuity of $x$, we first exclude jumps and then bound changes
in quality by changes in $B$ and $y$.

Suppose $x$ had an interior jump at $u_0$, with one-sided limits
$x_-<x_+$. Taking limits along ranks where
\eqref{eq:gs-state-control} holds gives
$y(u_0)=C'(x_\pm)+C''(x_\pm)\{\alpha x_\pm-B(u_0)\}$.
Both limiting qualities satisfy $B(u_0)\leq\alpha x_\pm$.
For every $z\in[x_-,x_+]$, therefore,
\[
 \frac{\partial}{\partial z}
 \bigl[C'(z)+C''(z)\{\alpha z-B(u_0)\}\bigr]
 =D_\alpha(z,B(u_0))>0.
\]
The same strictly increasing function cannot take the value $y(u_0)$
at two distinct qualities. Thus $x$ has no interior jump. Together
with right continuity at zero and the continuity at $b$ just proved,
this makes $x$ continuous on $[0,b]$. In particular,
$x_*:=\min_{0\leq u\leq b}x(u)=x(0)>0$.
All terms in \eqref{eq:gs-state-control} are now continuous in rank,
so the identity, initially obtained almost everywhere, holds everywhere
on $[0,b]$.

Next, define
$F(z,B,y):=C'(z)+C''(z)(\alpha z-B)-y$.
At every point of the allocation path, $F(x(u),B(u),y(u))=0$ and
$F_z=D_\alpha>0$. For $x(u)<\bar x$, the implicit function theorem
therefore expresses the root of the equation locally as a continuously differentiable
function of $(B,y)$.

To apply the same theorem at $x=\bar x$, we need a function defined
on an open neighborhood. Retain $C$ below $\bar x$ and, for
$h\geq0$ sufficiently small, define
\[
 \widetilde C(\bar x+h)
 =C(\bar x)+C'(\bar x)h
   +\frac12 C''(\bar x)h^2+\frac16 C'''(\bar x)h^3.
\]
The values and the first three derivatives match at $\bar x$, so this
is a $C^3$ extension near that endpoint. Let $\widetilde D_\alpha$
denote the expression defining $D_\alpha$ with $C$ replaced by
$\widetilde C$. At the state reached at rank $b$,
$\widetilde D_\alpha(\bar x,\alpha\bar x)
 =(1+\alpha)C''(\bar x)>0$.
Continuity preserves this inequality on a sufficiently small open
neighborhood, so the implicit function theorem applies there as well.

Fix a rank and write $x=X_j(B,y)$ for the local solution given by
the implicit function theorem. Continuity of $x$, $B$, and $y$
ensures that this same function determines quality at all sufficiently
nearby ranks.
Choose a small ball around the corresponding state $(B,y)$ whose
closure lies inside the domain of $X_j$. Restrict the neighborhood,
if necessary, so that the states at nearby ranks remain in this ball. The partial
derivatives of $X_j$ are continuous and therefore bounded on the
ball's compact closure. Since the line segment between any two states
in the ball stays inside it, the mean value theorem gives a finite
constant $L_j$ such that
$|X_j(B_2,y_2)-X_j(B_1,y_1)|
 \leq L_j\bigl(|B_2-B_1|+|y_2-y_1|\bigr)$.
Thus, for ranks in this neighborhood, the change in quality is
bounded by a fixed multiple of the changes in the two states.

These neighborhoods cover the rank interval. By compactness of
$[0,b]$, finitely many suffice. There is also
a positive number $d$ such that every rank interval of length less
than $d$ lies entirely in one of these neighborhoods. We can therefore
choose a finite partition
$0=u_0<u_1<\cdots<u_N=b$
with each piece shorter than $d$. On each piece, one of the functions
$X_j$ represents quality and the preceding bound applies. Taking the maximum
of the finitely many constants, and increasing it to one if necessary,
gives a common bound $L\geq1$.

To prove absolute continuity of $x$, fix $\eta>0$. Since $B$ and $y$
are absolutely continuous, we can choose $\delta>0$ so that any finite
collection of disjoint rank intervals $[s_k,t_k]\subseteq[0,b]$ with
total length less than $\delta$ satisfies
\[
 \sum_k
 \bigl(
 |B(t_k)-B(s_k)|+|y(t_k)-y(s_k)|
 \bigr)
 <\frac{\eta}{L}.
\]
For any such collection, split its intervals at the partition points
$u_j$. The resulting collection has the same total length, and every
piece lies in a neighborhood with the Lipschitz bound above.
Applying that bound on each piece and then the triangle inequality gives
\[
 \sum_k|x(t_k)-x(s_k)|
 \leq L\sum_{\text{pieces }[s,t]}
       \bigl(|B(t)-B(s)|+|y(t)-y(s)|\bigr)<\eta.
\]
This proves that $x$ is absolutely continuous. Finally, $C''$ is
bounded on $[x_*,\bar x]$, so $C'$ is Lipschitz on the range of
qualities attained by $x$. Applying the same argument to $C'\circ x$
shows that $\phi=C'(x)$ is absolutely continuous on $[0,b]$.
\end{proof}

\subsection{The unique market and its lower endpoint.}
\label{app:rw-uniqueness}

The first-order condition and the equation reconstructing types from
virtual values give a system of two ordinary differential equations.
Every optimal market must solve this system with the same values at
the highest type. Uniqueness for the system will therefore identify
a common path when it is followed backward from that type. The
minimum-type condition determines where the path stops, and hence
which market it describes.

\paragraph{The differential system.}
Set $t=-\log(1-u)$ and $T=-\log(1-b)$, and use dots for derivatives
with respect to $t$. Since $du/dt=1-u$, this change of variable
removes the factors $1-u$ from the rank equations.
Equations~\eqref{eq:gs-B}, \eqref{eq:gs-Q-ode}, and
\eqref{eq:gs-state-control} give
\begin{equation}
 \begin{aligned}
 \dot B&=[1-\rho(y)]B+\rho(y)X_\alpha(B,y),\\
 \dot y&=\rho(y)\{y-C'(X_\alpha(B,y))\}.
 \end{aligned}
 \label{eq:rw-two-state-system}
\end{equation}
Here $x=X_\alpha(B,y)$ is the unique admissible quality determined
by \eqref{eq:gs-state-control}. Given the two states, the actual
type is $Q=\nu^{-1}(y)$ and the virtual value is $\phi=C'(x)$.
Thus the states determine both types and qualities throughout the
screening profile.

By the regularity established above, we may differentiate the
relation between states and quality to obtain
$\dot y=D_\alpha\dot x-C''(x)\dot B$.
Substituting \eqref{eq:rw-two-state-system} yields
\begin{equation}
 \dot x=
 \frac{C''(x)\{\rho(y)(1+\alpha)x+[1-2\rho(y)]B\}}
      {D_\alpha(x,B)}.
 \label{eq:gs-x-log}
\end{equation}
On the separating interval, $\dot B>0$ and $\dot y>0$ because
$x>0$, $\rho>0$, and $y-C'(x)>0$. Also $\dot x>0$:
if $\rho\leq1/2$, every term in its numerator is nonnegative and
the first is positive; if $\rho>1/2$,
$B<\alpha x$ bounds the expression in braces below by
$[\rho+\alpha(1-\rho)]x>0$. Its denominator is positive by
\eqref{eq:rw-characterization-conditions}.

The boundary conditions are
\begin{equation}
 (B(T),y(T),x(T))=(\alpha\bar x,\bar\nu,\bar x),
 \label{eq:gs-top-boundary}
\end{equation}
\begin{equation}
 B(0)\geq0,\qquad y(0)\geq\nu(\underline\theta),
 \qquad B(0)\{y(0)-\nu(\underline\theta)\}=0.
 \label{eq:gs-bottom-boundary}
\end{equation}
The first condition describes the continuous transition into the
atom at the highest type, where quality is efficient. At the lower
endpoint, $B(0)=\xi$ and $y(0)=\nu(M(\phi))$, so
\eqref{eq:gs-bottom-boundary} is exactly the minimum-type constraint
and complementary slackness, using strict increase of $\nu$.

\begin{appendixproposition}[Characterization and uniqueness]
\label{prop:app-stopped-characterization}
Under the hypotheses of Proposition~\ref{prop:rw-endpoints}, every
optimizer in the specified structural class satisfies
\eqref{eq:gs-state-control}, \eqref{eq:rw-two-state-system}, and
\eqref{eq:gs-top-boundary}--\eqref{eq:gs-bottom-boundary}.
Its states and quality strictly increase on the separating interval.
Traced backward from the top state, its solution stops at the first
occurrence of
$B=0$ or $y=\nu(\underline\theta)$.
There is at most one reduced optimizer in this class. Its separating
length determines the highest-type mass by
\begin{equation}
 b=1-e^{-T},\qquad 1-b=e^{-T}.
 \label{eq:paper-app-top-mass}
\end{equation}
\end{appendixproposition}

\begin{proof}
The preceding calculations give the equations and boundary
conditions. The function determining quality is locally $C^1$, and
$\rho$ is locally Lipschitz. Hence the right-hand side of
\eqref{eq:rw-two-state-system} is locally Lipschitz near each point
of an optimizer's path. At the top endpoint, the same conclusion
holds using a local extension. Uniqueness for the initial-value
problem therefore implies that any two optimal paths coincide
when followed backward from their common top state, for as long
as both are defined.

Their stopping points must also agree. Complementary slackness
rules out stopping while both $B>0$ and $y>\nu(\underline\theta)$.
Nor can a path continue backward beyond either boundary. At
$B=0$, the backward derivative of $B$ is $-\rho(y)x<0$, so continuation
would make $B$ negative. At $y=\nu(\underline\theta)$, the backward
derivative of $y$ is $-\rho(y)\{y-C'(x)\}<0$, so continuation would
violate the minimum-type restriction. The difference $y-C'(x)$ is
strictly positive even at the lower endpoint, by
Lemma~\ref{lem:gs-strict-wedge}. The same argument applies if the
two boundaries are reached at once. Every optimizer therefore has
the same stopping point and duration $T$. Its qualities are
determined by the relation between states and quality, its types
by reconstruction, and its highest-type mass by
\eqref{eq:paper-app-top-mass}. The market is consequently unique.
\end{proof}

For $\beta\in[0,1)$,
Propositions~\ref{prop:rw-existence-top} and
\ref{prop:app-qualitative-structure} give existence and show that
every optimizer belongs to the required class. Under
Assumption~\ref{ass:main-curvature},
Lemma~\ref{lem:gs-state-control-monotone} supplies
\eqref{eq:rw-characterization-conditions}. The characterization thus
identifies the unique reduced optimizer. The regularization
equation then identifies the unique original market.

The state equations also imply that $Q$ has a continuous, strictly
positive derivative at every interior rank below $b$. Its inverse
therefore gives a positive continuous density $g_G$ on the open
type support. At $\theta=Q(u)$,
\[
 \frac{1-G(\theta)}{g_G(\theta)}
 =(1-u)Q'(u)=\nu(\theta)-\phi(u).
\]

\paragraph{The lowest selected type.}
When the minimum feasible type is zero, define
$\theta_\beta^0=Q_{\beta,0}(0)$. The atom at the highest type implies
$\theta_\beta^0>0$. The minimum-type constraint is therefore slack,
its multiplier is zero, and the backward path ends at $B=0$.
Changing the minimum feasible type leaves the differential system
and its top boundary condition unchanged. It can only make the
path reach the minimum-type boundary earlier. Consequently
\begin{equation}
 \supp G_{\beta,\underline\theta}
 =[\ell_{\beta,\underline\theta},\bar\theta],
 \qquad
 \ell_{\beta,\underline\theta}
 =\max\{\underline\theta,\theta_\beta^0\}.
 \label{eq:gs-realized-lower-support}
\end{equation}
If the imposed minimum type is below $\theta_\beta^0$, the
constraint does not affect the market. If it is above
$\theta_\beta^0$, the same path is truncated. At equality, the
restriction binds and the market is unchanged.

\subsection{The Pareto frontier and unique implementation}
\label{app:pareto-conditioning}

Maintain Assumptions~\ref{ass:main-nu}--\ref{ass:main-curvature}.
For a fixed minimum type, write $(C_\beta,P_\beta)$ for the consumer
surplus and profit of the unique market selected by weight
$\beta\in[0,1]$. These payoffs are Pareto efficient. To prove the
converse we use
continuity to find a payoff with the same profit as any
given efficient payoff.

\paragraph{Continuity of the selected markets and payoffs.}
Let $\beta_n\to\beta\in[0,1]$. By weak compactness, any sequence of
selected markets has a convergent subsequence. Its payoffs converge
by Proposition~\ref{lem:app-compact-density}. Passing to the limit
in the optimality inequalities shows that the limiting market is
optimal at $\beta$, and uniqueness identifies it as $G_\beta$.
Every limit of a subsequence is therefore the same. This proves
$G_{\beta_n}\Rightarrow G_\beta$ and convergence of the payoffs.
The same argument on the compact reduced domain gives
$\phi_{\beta_n}\to\phi_\beta$ in $L^1$. Both conclusions include
$\beta=1$, where the selected market is degenerate.

\paragraph{Ordering the payoffs.}
For $\beta_1<\beta_2$, compare each selected market under both objectives.
Rearranging the two optimality inequalities gives
$\beta_1(P_{\beta_2}-P_{\beta_1})
 \leq C_{\beta_1}-C_{\beta_2}
 \leq\beta_2(P_{\beta_2}-P_{\beta_1})$.
Thus, profit is nondecreasing and consumer surplus nonincreasing in
$\beta$. If profits are equal, the inequalities force equal
consumer surpluses, and uniqueness then implies that the markets
are the same. Distinct markets therefore have strictly ordered
profits. Their consumer surpluses are also strictly ordered. This
follows from the left inequality when $\beta_1>0$, while for
$\beta_1=0$ equality would give a second market maximizing consumer
surplus, contrary to uniqueness.

\paragraph{Every efficient payoff is selected.}
For $\beta>0$, both objective weights are positive, so a selected
payoff is Pareto efficient. At $\beta=0$, uniqueness rules out
another market with the same maximal consumer surplus and higher
profit. Conversely, take any Pareto-efficient pair $(c,p)$. It must
satisfy $P_0\leq p\leq
P_1$ because a pair with lower profit would be dominated by the market
maximizing consumer surplus, and $P_1$ is the largest feasible
profit. By continuity of $P_\beta$, some $\beta\in[0,1]$ has
$P_\beta=p$. If $c<C_\beta$, the selected pair dominates $(c,p)$.
If $c>C_\beta$, $(c,p)$ dominates the selected pair. Both are
impossible, so $(c,p)=(C_\beta,P_\beta)$. Any second market
implementing this pair would maximize the same weighted objective,
again contradicting uniqueness.

\paragraph{One market may be selected at several weights.}
Suppose two weights $\beta_1<\beta_2$ in $[0,1]$ select the same market.
Neither weight can be one, since the top market is not optimal
below one. If the minimum-type constraint were slack, its
multiplier would be zero for both weights. Subtracting the two
Euler equations would then give $(\beta_2-\beta_1)q(\phi)=0$ on
the separating interval, where quality is positive, a contradiction.
Thus, selection of the same market by more than one weight requires
a binding minimum-type constraint.

\section{Payoff Boundaries, Geometry, and Segmentation}
\label{app:paper-right-frontier}

We first complete the characterization of the maximum consumer surplus frontier by finding the
greatest consumer surplus a single market can generate at each profit.
We then prove the convexity results, give the
nonconvexity example, and establish the segmentation results.

\subsection{The maximum consumer surplus frontier}
\label{app:negative-frontier}

Maintain Assumptions~\ref{ass:main-nu}--\ref{ass:main-C}. For consistency with the main text,
write $L_{\underline\theta}(\beta)=U_{\underline\theta}(\beta)$ when
$\beta\leq0$, where $U$ is defined in
\eqref{eq:model-support-problem}. Thus, $L$ denotes the same maximization
problem, restricted to nonpositive weights. Define
\begin{equation}
 P^{\min}_{\underline\theta}:=\min_G\Pi(G),\qquad
 P^{\max}:=\pi(\bar\nu),\qquad
 C^{\min}_{\underline\theta}:=
       \max_{\Pi(G)=P^{\min}_{\underline\theta}}\CS(G).
 \label{eq:lf-min-profit}
\end{equation}
Here $C^{\min}_{\underline\theta}$ is the greatest consumer surplus
among markets that minimize profit. The superscript refers to
\emph{minimum profit}. Throughout this subsection, every market
satisfies the indicated minimum-type constraint. By compactness and
payoff continuity in Proposition~\ref{lem:app-compact-density}, the
extrema in \eqref{eq:lf-min-profit} are attained, as is the
fixed-profit maximum $\Gamma_{\underline\theta}(p)$. Write
$z^{\min}_{\underline\theta}=(C^{\min}_{\underline\theta},
P^{\min}_{\underline\theta})$.

The degenerate market at the highest type attains the upper profit
bound. To determine when the lower bound is zero, observe that a
zero-profit market has a regularized and truncated profile satisfying
$\pi(\phi)=0$ almost everywhere. Hence $\phi=0$, and the
reconstructed quantile is $Q_\phi=0$. Regularization weakly
raises types, so the original market must also be $\delta_0$. Thus,
$P^{\min}_{\underline\theta}>0$ when $\underline\theta>0$, whereas
$z^{\min}_0=(0,0)$.

\paragraph{Adding buyers who do not trade.}
Adding buyers of type zero reduces both payoffs in the same proportion.
The seller continues to offer the original buyers an optimal menu and
offers the added buyers the outside option.

\begin{appendixlemma}[Dilution by the payoff-neutral type]
\label{lem:lf-zero-dilution}
For $G\in\Delta(\Theta_0)$ and $\lambda\in[0,1]$, the market
$G^\lambda=(1-\lambda)\delta_0+\lambda G$ satisfies
$(\CS(G^\lambda),\Pi(G^\lambda))
       =\lambda(\CS(G),\Pi(G))$.
\end{appendixlemma}
\begin{proof}
The cases $\lambda=0$ and $\lambda=1$ are immediate. Suppose
$0<\lambda<1$. Extend a seller-optimal menu for $G$ by offering the
outside option to type zero. She weakly
prefers the outside option to every other option. The original buyers
retain their optimal choices. This menu earns $\lambda\Pi(G)$.

Conversely, take any feasible menu for $G^\lambda$. Its restriction
to $G$ earns at most $\Pi(G)$. Individual rationality at zero gives
$t(0)\leq0$, and $m(0)=0$, so the added buyers contribute
$t(0)-C(x(0))\leq0$ each. Total profit is therefore at most
$\lambda\Pi(G)$. Since the preceding menu attains this bound, both
inequalities must be equalities at an optimum. Strict convexity and
$C(0)=C'(0)=0$ then imply $x(0)=t(0)=0$, and the restriction of the
menu to $G$ is seller optimal. By
Proposition~\ref{lem:gs-seller-reduction}, all seller-optimal menus
for $G$ give the same consumer surplus. The added buyers receive
zero surplus, so consumer surplus also scales by $\lambda$.
\end{proof}

Call a market fully active if no positive mass of its realized buyers
receives zero quality.

\begin{appendixproposition}[Completing the maximum-consumer surplus frontier]
\label{prop:app-negative-frontier}
Impose also Assumption~\ref{ass:lf-strong-curvature}.
\begin{enumerate}[label=(\roman*)]
 \item If \(\underline\theta>0\), every finite \(\beta\leq0\) has
 a unique original-market optimizer, irrespective of the sign of the
 optimized value. It is regular and fully active, has connected support,
 strictly separates below an efficiently served highest-type atom, and
 has no other atoms. Its payoff
 \(z^L_{\beta,\underline\theta}=
 (C^L_{\beta,\underline\theta},P^L_{\beta,\underline\theta})\)
 is continuous in \(\beta\); both coordinates are nondecreasing and
 strictly increase whenever the market changes. Moreover,
 \begin{align*} z^L_{\beta,\underline\theta}&\longrightarrow
 z^{\min}_{\underline\theta},\quad\beta\downarrow-\infty \quad\text{and}\quad
 \mathcal F^L(\mathcal V_{\underline\theta})
 =\{z^L_{\beta,\underline\theta}:\beta\in(-\infty,0]\}
       \cup\{z^{\min}_{\underline\theta}\}.
 \end{align*}
 Every weak accumulation point of the selected markets implements
 \(z^{\min}_{\underline\theta}\).

 \item If \(\underline\theta=0\), define
 \begin{equation}
 \underline\beta:=\inf\{\beta\leq0:L_0(\beta)>0\}\in[-\infty,0).
 \label{eq:lf-critical-beta}
 \end{equation}
 For \(\underline\beta<\beta\leq0\), the original-market optimizer
 is unique and has the structure in (i). Its payoff path is continuous
 and strictly increasing in both coordinates. At every finite
 \(\beta\leq\underline\beta\), the unique optimizer is \(\delta_0\).
 As \(\beta\downarrow\underline\beta\), interpreted as
 \(\beta\to-\infty\) if \(\underline\beta=-\infty\), the selected
 markets converge weakly to \(\delta_0\), and their payoffs converge
 to \((0,0)\). Consequently,
 $\mathcal F^L(\mathcal V_0)
 =\{z^L_{\beta,0}:\underline\beta<\beta\leq0\}\cup\{(0,0)\}$.
 In particular, no nontrivial segment of maximizing payoffs occurs at a
 finite critical coefficient.
\end{enumerate}
For either floor regime,
\begin{equation}
 \mathcal F^P(\mathcal V_{\underline\theta})
 \cup\mathcal F^L(\mathcal V_{\underline\theta})
 =\{(\Gamma_{\underline\theta}(p),p):
          P^{\min}_{\underline\theta}\leq p\leq P^{\max}\}.
 \label{eq:lf-right-frontier}
\end{equation}
Every such payoff with \(p>P^{\min}_{\underline\theta}\) is uniquely
implemented. The minimum-profit endpoint is uniquely implemented when
type zero is feasible, and also under private values.
\end{appendixproposition}

\begin{proof}[Proof, including Proposition~\ref{prop:main-right-frontier}]
\emph{1. A positive minimum type precludes exclusion.}
Existence follows from Proposition~\ref{lem:gs-compact-continuous}.
Fix a finite \(\beta\leq0\) and \(\underline\theta>0\). We show that
every optimizer serves almost every buyer, even if the maximized
objective is nonpositive. An identically zero reduced profile is
infeasible. Suppose a reduced optimizer excludes positive mass, and let
\(a:=|\{u:\phi(u)=0\}|\in(0,1)\). We raise its smallest virtual
values to a common positive level. Put
\(b_\varepsilon=\inf\{u:\phi(u)\geq\varepsilon\}\) and replace
\(\phi\) by \(\phi_\varepsilon=\max\{\phi,\varepsilon\}\).
Then \(b_\varepsilon\downarrow a\), the new profile remains
nondecreasing, and its reconstructed quantile weakly rises. Thus the
minimum-type restriction is preserved.

On the initial interval where \(\phi=0\), the objective gain is at least
$a q(\varepsilon)\{
 \nu(\underline\theta)-(1+|\beta|)\varepsilon\}$,
since \(\pi(\varepsilon)\leq\varepsilon q(\varepsilon)\).
On \([a,b_\varepsilon)\), quality is at most \(q(\varepsilon)\)
both before and after the change, and the profit is at most
\(\pi(\varepsilon)\leq\varepsilon q(\varepsilon)\). Since virtual
values and net values are bounded, the absolute change in objective is at most
\(K_\beta q(\varepsilon)\) for a constant \(K_\beta\) independent
of \(\varepsilon\). The length of this interval tends to zero, so
its contribution is \(o(q(\varepsilon))\). Above
\(b_\varepsilon\), the profile is unchanged, and so is the quantile,
because reconstruction depends only on the upper tail. The total gain
is positive for sufficiently small \(\varepsilon\), a contradiction.

Regularizing and truncating any original optimizer produces a reduced
optimizer, which must therefore be fully active. Equality in
\eqref{eq:gs-reg-CS} identifies the original market with its
regularization.

\par\medskip\noindent\emph{2. Structure and uniqueness of a fully active optimizer.}
For every finite \(\beta\leq0\), strong curvature supplies
\eqref{eq:rw-characterization-conditions} with \(\alpha=1-\beta\geq1\).
For a fully active optimizer, Lemmas~\ref{lem:rw-top-atom},
\ref{lem:gs-strict-wedge}, and \ref{lem:attenuation-no-bunching-tilt}
give an interval on which \(\phi=\bar\nu\) and strict separation
below it. This interval corresponds to an efficiently served atom at
the highest type. Proposition~\ref{prop:rw-existence-top} rules out
a market consisting only of that atom.
Propositions~\ref{prop:rw-endpoints} and
\ref{prop:app-stopped-characterization} give regularity and uniqueness.
Equality in regularization then gives uniqueness among original
markets. Together with Step 1, these results establish the structure
and uniqueness claimed in (i) at every finite weight. When type zero
is feasible, the same argument applies once full activity has been
proved.

\par\medskip\noindent\emph{3. Optimal markets when type zero is feasible.}
The market \(\delta_0\) gives zero payoffs, so \(L_0\geq0\).
For each market, the objective is affine in \(\beta\), with slope
equal to its profit. These slopes lie in \([0,P^{\max}]\), so taking
the maximum over markets gives a value \(L_0\) that is nondecreasing,
convex, and \(P^{\max}\)-Lipschitz. Moreover, \(L_0(0)>0\) by
Proposition~\ref{prop:rw-existence-top}. It follows that the cutoff
in \eqref{eq:lf-critical-beta} has the stated range. The value is
positive on \((\underline\beta,0]\), and, if
\(\underline\beta\) is finite, the value is zero at every smaller
coefficient by definition and at the cutoff by continuity.

We next show that \(\delta_0\) is the only optimizer whenever the
value is zero. Fix a finite \(\beta\leq0\) with
\(L_0(\beta)=0\), and suppose an optimizer has nonzero payoff.
Regularization and truncation yield a nonzero \(\phi\in\Phi\)
with \(J_\beta(\phi)=0\). For any \(a\in(0,1)\), insert an
initial interval of ranks with zero virtual value:
\[
 \psi(u)=
 \begin{cases}
  0,&0\leq u<a,\\
  \phi\!\left(\dfrac{u-a}{1-a}\right),&a\leq u\leq1.
 \end{cases}
\]
Reconstruction on the upper tail gives
\(Q_\psi(a+(1-a)r)=Q_\phi(r)\) for \(r<1\); consequently
\(J_\beta(\psi)=(1-a)J_\beta(\phi)=0\). Thus \(\psi\) is also
optimal. By \eqref{eq:gs-profile-tail} and \(m\geq0\), however,
\[
 M(\psi)=\int_0^1[\psi(u)+m(Q_\psi(u))]\,du
 \geq(1-a)\int_0^1\phi(u)\,du>0.
\]
Thus the added ranks have positive reconstructed types: even the lowest
type, \(M(\psi)\), is strictly positive. Choose an auxiliary
minimum type \(0<\widehat\theta<M(\psi)\). The profile \(\psi\)
is feasible and remains optimal under this restriction. It attains
zero, and no market in the smaller feasible set can exceed the maximum
with a zero minimum type. Yet \(\psi\) excludes at least mass
\(a\), contradicting Step 1. Hence every optimizer of a zero-valued
problem gives the zero payoff, which is implemented only by
\(\delta_0\).

Now suppose \(L_0(\beta)>0\). A reduced optimizer cannot assign
zero quality on an initial interval of positive length: conditioning
on the remaining ranks would strictly raise its positive objective by
\eqref{eq:paper-app-active-payoff-scaling}. The optimizer is therefore
fully active, and Step 2 gives uniqueness among original markets. Its
nonzero reduced profile satisfies \(M(\phi)\geq\int\phi>0\), so
the minimum-type constraint is slack and its multiplier vanishes.

\par\medskip\noindent\emph{4. Continuity and payoff order.}
At every finite coefficient where the optimizer is unique, compactness
and payoff continuity imply that any weak limit of optimal markets is
optimal at the limiting coefficient. Uniqueness then gives weak
continuity of the selected market and continuity of its payoffs, as
in Subsection~\ref{app:pareto-conditioning}.

Take \(\beta_1<\beta_2\leq0\), and write \((C_i,P_i)\) for the
optimal payoffs at \(\beta_i\). The two optimality inequalities give
$P_1\leq P_2$ and $-\beta_2(P_2-P_1)\leq C_2-C_1\leq-\beta_1(P_2-P_1)$.
Equal profits force equal surpluses, so the same market is optimal at
both coefficients. If the markets differ, profit must therefore
strictly increase. Consumer surplus also strictly increases because the lower
bound above is positive when \(\beta_2<0\), and uniqueness of the
consumer surplus maximizer implies a strict increase when \(\beta_2=0\).

When the minimum type is zero and the value is positive, distinct
coefficients must select different markets. Otherwise, both
minimum-type multipliers vanish by Step 3, and subtracting
\eqref{eq:gs-Euler} at the two coefficients gives
\((\beta_2-\beta_1)q(\phi)=0\) on a positive-quality interval,
a contradiction. Thus both payoff coordinates strictly increase on
the positive-value branch in (ii).

\par\medskip\noindent\emph{5. Minimum-profit limits.}
For \(\beta<0\), comparison with a market implementing
\(z^{\min}_{\underline\theta}\) gives
\begin{equation*}
 0\leq P_\beta-P^{\min}_{\underline\theta}
 \leq\frac{C_\beta-C^{\min}_{\underline\theta}}{|\beta|}
 \leq\frac{\overline C-C^{\min}_{\underline\theta}}{|\beta|},
 \qquad \overline C:=\max_G\CS(G).
\end{equation*}
Thus profit tends to its minimum as \(\beta\to-\infty\), while
\(C_\beta\geq C^{\min}_{\underline\theta}\). By compactness and
payoff continuity, every weak accumulation point is a minimum-profit
market. Its consumer surplus is at most
\(C^{\min}_{\underline\theta}\) by definition, and at least that
amount by the preceding inequality. Thus, every such market implements
\(z^{\min}_{\underline\theta}\), and the entire payoff path
converges to this point.

Now let the minimum type be zero. If \(\underline\beta\) is finite,
continuity of the objective and its maximized value imply that every
weak accumulation point as \(\beta\downarrow\underline\beta\)
is optimal at \(\underline\beta\). By Step 3, it must be
\(\delta_0\), so the entire market path converges to \(\delta_0\).
If \(\underline\beta=-\infty\), the preceding estimate gives
profit converging to zero. Since only \(\delta_0\) has zero profit,
compactness again gives convergence of the markets to \(\delta_0\).

\par\medskip\noindent
\emph{6. The boundary and its implementing markets.}
Let \(P_0=\Pi(G_{0,\underline\theta})\). For \(\beta<0\),
a feasible payoff with weakly higher consumer surplus and
weakly lower profit than an optimizer, with one strict
inequality, would give a higher value of \(\CS+\beta\Pi\).
Thus every optimizer is lower efficient. The optimizer
at \(\beta=0\) is lower efficient by uniqueness, and
\(z^{\min}_{\underline\theta}\) is lower efficient by
definition. Continuity and the endpoint limits established
in Steps 4--5 therefore give a lower-efficient payoff at
every profit in
\([P^{\min}_{\underline\theta},P_0]\). At each such profit
\(p\), lower efficiency rules out any higher consumer
surplus, so this payoff is
\((\Gamma_{\underline\theta}(p),p)\).

These are all the lower-efficient payoffs. A payoff with
profit above \(P_0\) cannot be lower efficient, since
\(G_{0,\underline\theta}\) gives at least as much consumer
surplus and strictly less profit. At any feasible profit
\(p\leq P_0\), a lower-efficient payoff must coincide with
\((\Gamma_{\underline\theta}(p),p)\). This proves the
characterization of
\(\mathcal F^L(\mathcal V_{\underline\theta})\).
Subsection~\ref{app:pareto-conditioning} identifies the
remaining boundary over \([P_0,P^{\max}]\), proving
\eqref{eq:lf-right-frontier}.

For each
\(p\in(P^{\min}_{\underline\theta},P^{\max}]\),
the payoff \((\Gamma_{\underline\theta}(p),p)\) maximizes
\(\CS+\beta\Pi\) for some finite \(\beta\). Any market
implementing this payoff must therefore maximize the same
objective. Since that objective has a unique optimal
market, the implementing market is unique.

The same argument applies at the minimum-profit endpoint
if it is attained at a finite weight. Otherwise, the
convergence established in Step 5 does not by itself prove
uniqueness. When type zero is feasible,
\(\delta_0\) uniquely implements this endpoint.

Under private values, the minimum-profit market is also
unique when \(\underline v>0\). To see this, take any
feasible market \(G\). In normalized units, regularization
and truncation preserve profit and give a profile
\(\phi:[0,1]\to[0,1]\) satisfying
\(Q_\phi(0)=\int_0^1\phi(u)\,du\geq\underline v\).
Jensen's inequality and monotonicity of \(\pi\) give
$\Pi(G)\geq\pi(\underline v)$.
The market \(\delta_{\underline v}\) attains this bound.
Strict increase of \(\pi\) requires
\(\int_0^1\phi(u)\,du=\underline v\) whenever the bound
is attained, and strict convexity then requires
\(\phi=\underline v\) almost everywhere. The reconstructed
quantile is therefore constant at \(\underline v\).
Since regularization and truncation weakly raise the
original quantile,
$\underline v\leq Q_G(u)\leq Q_\phi(u)=\underline v$.
Thus, \(Q_G\) is also constant at \(\underline v\), and
\(G=\delta_{\underline v}\).
\end{proof}

\subsection{The geometry of the surplus set}
\label{app:paper-geometry}

The preceding results identify the greatest consumer surplus at each
profit. To describe the whole payoff set, we must also implement the
points below this boundary. A neutral type lets us reduce both payoffs
in the same proportion. Under private values, we can instead reduce
consumer surplus continuously while keeping profit fixed, even when no
neutral type is feasible.

Mixing type distributions need not mix their payoffs, since the seller
reoptimizes its menu. We therefore begin with the convex hull
$K=\co(\mathcal V)$ and show that each of its points is implemented by
a single market.

\begin{proof}[Proof of Theorem~\ref{thm:main-geometry}(i)--(ii)]
The dilution claim follows from Lemma~\ref{lem:lf-zero-dilution} under
Assumptions~\ref{ass:main-nu}--\ref{ass:main-C}. For the convexity
claims, impose also Assumption~\ref{ass:lf-strong-curvature}.
Proposition~\ref{lem:app-compact-density} makes the payoff set compact,
so its convex hull $K=\co(\mathcal V)$ is compact as well.

We first record a geometric observation that will apply in both cases.
Suppose all points of $K$ at minimum or maximum profit have zero
consumer surplus. The segment joining these two payoffs lies in $K$,
so $K$ contains a zero-surplus point at every feasible profit. Now take
a boundary point $z$ with positive consumer surplus and a supporting
normal at $z$, oriented so that its linear objective is maximized
there. The coefficient on consumer surplus cannot be negative since the
zero-surplus point at the same profit would then give a strictly larger
value. It can't also be zero because an objective depending
only on profit is maximized at a profit endpoint. We can therefore
normalize the normal to $(1,\beta)$. The points of $K$ on the supporting
line are precisely the convex combinations of single-market payoffs
maximizing $c+\beta p$. Indeed, in any finite convex representation of
such a point, each positive-weight payoff must attain the same bound. Otherwise, the average would fall strictly below it.

\par\medskip\noindent\emph{1. A feasible neutral type.}
When type zero is feasible, zero profit implies zero surplus, and only the top
market attains maximum profit.
A market concentrated entirely on the payoff neutral type implements the zero payoff pair. For any nonzero $x\in K$,
extend the ray from the origin through $x$ to its last point $z$ in
$K$, which exists by compactness. We will implement $z$ and then use
dilution to reach $x$.

If $z$ has zero surplus, it is the top payoff because the whole segment from
the origin to that payoff lies in $K$, and no point has greater profit.
If $z$ has positive surplus, the preceding observation gives a
supporting objective $c+\beta p$ at $z$.
Theorem~\ref{thm:main-structure} and
Proposition~\ref{prop:app-negative-frontier} give a unique maximizing
market at every finite coefficient. Hence the supporting line meets
$K$ at a single payoff, which belongs to $\mathcal V_0$.
In this case, $z$ is Pareto efficient when $\beta\geq0$ and lower
efficient when $\beta<0$; the top payoff is also Pareto efficient.
Thus $z\in\mathcal V_0$ in either case. Writing $x=\lambda z$ with
$\lambda\in(0,1]$, Lemma~\ref{lem:lf-zero-dilution} implements
$x$ with one market. This proves \eqref{eq:main-radial-set} and
$K=\mathcal V_0$.

\par\medskip\noindent\emph{2. Private values without a neutral type.}
Normalize valuations to $[\underline v,1]$, with $\underline v>0$.
Regularization and truncation preserve profit and weakly increase
consumer surplus. The resulting profile takes values in $[0,1]$,
and its feasibility condition is $\int_0^1\phi\geq\underline v$.
Jensen's inequality, followed by monotonicity of $\pi$, gives the lower
bound below; the upper bound follows from $\phi\leq1$. Thus every
original market satisfies
$\pi(\underline v)\leq\Pi=\int_0^1\pi(\phi)\,du\leq\pi(1)$.
At the lower endpoint, strict convexity and the feasibility inequality
force $\phi=\underline v$ almost everywhere. At the upper endpoint,
strict increase of $\pi$ forces $\phi=1$ almost everywhere. Both
constant profiles yield zero consumer surplus. Consumer surplus is
nonnegative, and regularization and truncation weakly increase it,
so the original market also has zero surplus. Degenerate markets implement every
$(0,p)$ in this interval.

We next show how to reduce consumer surplus at a fixed profit. Let a
nonnegative regular schedule $\phi$ have payoff $(c,p)$ with $c>0$,
and choose $r$ so that $\pi(r)=p$. Jensen's inequality and $\pi$ being strictly increasing on positive arguments imply $\underline v\leq\int\phi\leq r<1$, where the last inequality follows
because a market at maximum profit has zero surplus. Define
\begin{equation*}
 \phi_t=(1-t)\phi+tr,\qquad
 \phi_{t,s}=(1-s)\phi_t+s\cdot1,\qquad t,s\in[0,1].
\end{equation*}
\par\begingroup\interlinepenalty=10000\relax\noindent
The first change moves $\phi$ toward the constant schedule
$r$, which implements a degenerate private-value market and gives no
consumer surplus. Profit may fall by doing this. The
second moves toward the highest schedule, $1$, and
will let us restore profit to $p$.\par\endgroup

Both changes preserve monotonicity and feasibility. Along the
first, the integral lies between $\int\phi$ and $r$, and both are at least
$\underline v$. Along the second, it can only increase. Convexity of
$\pi$ gives $\Pi(\phi_t)\leq p$. For each $t$, profit is continuous
and strictly increasing in $s$, reaching $\pi(1)>p$ at $s=1$.
Strictness follows because $\phi_t<1$ on a set of positive measure,
including at $t=1$, when $\phi_t=r$. There is therefore a unique
$s(t)$ for which profit equals $p$.

This choice is continuous in $t$. For any sequence $t_n\to t$,
compactness gives converging subsequences of $s(t_n)$. Continuity of
profit makes every such limit a solution at $t$, and uniqueness makes
it equal to $s(t)$. Thus, $s(t_n)\to s(t)$. Moreover, $s(0)=s(1)=0$ since both
endpoint schedules already earn $p$. The resulting market thus varies
continuously from a market with consumer surplus $c$ to one with
consumer surplus zero, earning profit $p$ throughout. By the
intermediate value theorem, it attains every surplus between zero
and $c$.

Finally, fix a profit and take the point of $K$ with the greatest
consumer surplus at that profit. If this surplus is positive, the
supporting-line observation applies. The unique maximizing market
from Subsections~\ref{app:pareto-conditioning}--\ref{app:negative-frontier}
then implements the point and is regular. The preceding argument
implies that every smaller nonnegative surplus can be implemented at the same profit.
If the greatest surplus is zero, the degenerate market already
implements the point. Hence, with
$\Gamma^{PV}_{\underline v}(p)=\max_{\Pi(G)=p}\CS(G)$,
\begin{equation*}
 \mathcal V^{PV}_{\underline v}
 =\{(c,p):\pi(\underline v)\leq p\leq\pi(1),\quad
          0\leq c\leq\Gamma^{PV}_{\underline v}(p)\},
\end{equation*}
Every point of the convex hull is therefore implemented by a single
market, proving $K=\mathcal V^{PV}_{\underline v}$ and part~(ii).
\end{proof}

\subsection{A lower bound for all markets and the nonconvexity example}
\label{app:nonconvexity}

To prove Theorem~\ref{thm:main-geometry}(iii), we construct two markets
that are the only minimizers of the same linear function of surplus and
profit. No market can then attain an interior point of the line segment
joining their payoff pairs since such a market would be a third minimizer.

The construction uses a function $V(\theta)$ that bounds
$\CS+\lambda\Pi$ from below whenever the market's lowest type is
$\theta$. For this bound to hold at a degenerate market, we need
$V(\theta)\leq\lambda\pi(\nu(\theta))$. The next lemma introduces another
inequality for actual and virtual net values that differ, and shows
that the two conditions suffice. Its proof also identifies what
equality requires of the market.

\begin{appendixlemma}[A lower bound and its equality conditions]
\label{lem:surplus-lower-bound}
Maintain Assumptions~\ref{ass:main-nu}--\ref{ass:main-C}. Fix a feasible
interval $[a,b]$, $\lambda>0$, and a nondecreasing convex
$V\in C^1([a,b])$ such that
$V(\theta)\leq\lambda\pi(\nu(\theta))$. Suppose that, for every
$\theta\in[a,b]$ and $z<\nu(\theta)$,
\begin{equation*}
 \mathcal R(\theta,z):=V'(\theta)+q(z)
       -\frac{V(\theta)-\lambda\pi(z)}{\nu(\theta)-z}\geq0.
\end{equation*}
Then every $G\in\Delta([a,b])$ satisfies
\begin{equation}
 \CS(G)+\lambda\Pi(G)\geq V(Q_G(0))\geq V(a).
 \label{eq:surplus-lower-bound}
\end{equation}
If $V$ is strictly increasing and strictly convex, equality throughout
forces $Q_G(0)=a$, no support gaps, a concave adjusted raw revenue curve,
and local absolute continuity of $Q_G$ on $[0,1)$. Moreover,
$\mathcal R(Q_G,\phi_G)=0$ almost everywhere where $Q_G'>0$, and every
positive-mass atom is at a type where $V=\lambda\pi\circ\nu$.
\end{appendixlemma}

\begin{proof}
\emph{The lower bound.}
Write \(Q=Q_G\), \(Y=Q^+\), \(R=\cav(\widehat R)\),
\(\phi=-R'_+\), and \(\mathcal I=R-\widehat R\geq0\).
Thus \(\mathcal I\) measures the difference between the concavified
and raw revenue curves.
Lemma~\ref{lem:gs-concavification-regularity} gives
\(\omega:=\nu(Q)-\phi\geq0\) almost everywhere and
\(\mathcal I(0)=\mathcal I(1)=0\). Set \(X=V'(Y)\) and
\(H(u)=(1-u)V(Y(u))\).
Convexity of \(V\) makes \(X\) nondecreasing. The function \(H\)
applies the proposed bound to the market above rank \(u\), weighted
by its remaining mass \(1-u\); integrating its changes will leave
the bound \(V(Q(0))\) for the whole market.

\phantomsection\label{app:bc-verification-calculus}
We use right-continuous bounded-variation functions, with
\(df((s,t])=f(t)-f(s)\) and \(f_-\) denoting the left limit. The
Stieltjes product and convex chain rules are
\begin{align*}
 [fg]_0^1&=\int_{(0,1]}f_-\,dg+\int_{(0,1]}g\,df,
 \\
 d(V\circ Y)&=V'(Y_-)\,dY+d\Xi,
 \\
 \Delta\Xi(u)&=V(Y(u))-V(Y_-(u))-V'(Y_-(u))\Delta Y(u)\geq0.
\end{align*}
Here \(\Xi(0)=0\), and \(\Xi\) consists of the nonnegative
remainders at jumps of \(Y\). On the continuous part, including
singular continuous variation, the chain rule follows from uniform
continuity of \(V'\). The total mass of \(\Xi\) is finite, bounded by
\(2\|V'\|_\infty\operatorname{Var}(Y)\)
\citep{ambrosiofuscopallara2000,folland1999}.

The quantile conventions give \(Y(0)=Q(0)\), \(Y_-=Q\) on
\((0,1]\), and no jump at one. Differentiating the raw revenue curve
as a measure gives
\begin{align}
 d\widehat R&=-\nu(Q)\,du+(1-u)\,dY,
 \notag\\
 d\mathcal I&=\omega\,du-(1-u)\,dY.
 \label{eq:ironing-gap-differential}
\end{align}
Define
\begin{equation*}
 F=\omega\{V'(Q)+q(\phi)\}+\lambda\pi(\phi)-V(Q).
\end{equation*}
For \(\omega>0\), this is \(\omega\mathcal R(Q,\phi)\); for
\(\omega=0\), it is \(\lambda\pi(\nu(Q))-V(Q)\). Hence
\(F\geq0\) in either case. The contribution to
\(\CS+\lambda\Pi\) at rank \(u\) is
\(L=\omega q(\phi)+\lambda\pi(\phi)\). Applying the product and
chain rules to \(H\), and substituting for \(d\mathcal I\), gives
$L\,du+dH=F\,du-X_-\,d\mathcal I+(1-u)\,d\Xi$.
All terms can be integrated. Indeed, \(\phi\in L^1\), and
\(Q,V'\) are bounded. Since \(\phi\leq\nu(Q)\), the functions
\(q(\phi)\) and \(\pi(\phi)\) are bounded as well: \(q,\pi\)
are bounded on the feasible range of positive net values and vanish
for nonpositive arguments. The Stieltjes integrals have bounded
integrands and integrators of bounded variation. The product rule and
the zero endpoint values of \(\mathcal I\) therefore give
\(-\int X_-\,d\mathcal I=\int\mathcal I\,dX\). Since
\(H(1)=0\) and \(H(0)=V(Q(0))\), integration yields
\begin{equation}
 \CS(G)+\lambda\Pi(G)-V(Q(0))
 =\int_0^1F\,du+\int_{(0,1]}\mathcal I\,dX
                 +\int_{(0,1]}(1-u)\,d\Xi\geq0.
 \label{eq:surplus-bound-decomposition}
\end{equation}
The last two terms are nonnegative because \(\mathcal I\geq0\),
\(X\) is nondecreasing, and \(\Xi\) has nonnegative jumps.
Since \(Q(0)\geq a\), monotonicity of \(V\) proves
\eqref{eq:surplus-lower-bound}.

\emph{Equality.}
Suppose \(V\) is strictly increasing and strictly convex and equality
holds throughout the bound. The fact that $V$ is strictly increasing first implies \(Q(0)=a\).
All three nonnegative terms in the decomposition must vanish.
At every positive jump of \(Y\) below one, strict convexity makes
\(\Delta\Xi>0\), so the last integral would be positive. Thus
\(Y\) has no such jumps, and the quantile convention excludes a
jump at one. It follows that \(Q=Y\), the market has no support
gaps, and \(\widehat R,\mathcal I\) are continuous.

If \(\mathcal I>0\) on a connected component \((\ell,r)\),
\(\int\mathcal I\,dX=0\) makes \(X\) constant there. Strict
convexity of \(V\) then makes \(Q\) constant as well.
Equation~\eqref{eq:ironing-gap-differential} gives
\(d\mathcal I=\omega\,du\geq0\) on this interval. But a positive,
nondecreasing \(\mathcal I\) cannot satisfy \(\mathcal I(r)=0\),
as continuity and the endpoint conditions require.
Hence \(\mathcal I\equiv0\), so the adjusted raw revenue curve
is concave.
It follows that
$[\nu(Q)-\phi]\,du=(1-u)\,dQ$.
On each compact interval below one, \(1-u\) is bounded away from
zero. This identity therefore makes \(Q\) absolutely continuous
there, excluding singular continuous variation as well as jumps.
Where \(Q'>0\), it gives \(\omega>0\), and \(F=0\) implies
\(\mathcal R(Q,\phi)=0\). On a rank interval corresponding to a
positive-mass atom, \(Q'=0\), so \(\phi=\nu(Q)\) and \(F=0\)
gives \(V(Q)=\lambda\pi(\nu(Q))\). These are the claimed
equality conditions.
\end{proof}

\begin{appendixproposition}[Verification of the nonconvexity example]
\label{prop:nonconvexity-counterexample}
The primitives of Example~\ref{ex:main-nonconvexity} satisfy
Assumptions~\ref{ass:main-nu},~\ref{ass:main-C}, and
\ref{ass:lf-strong-curvature}. There is a coefficient
\(\lambda_*\in(1/2,1)\) and a market \(G^B\) such that
\begin{equation*}
 \mathop{\rm argmin}_{G\in\Delta([1,2])}
       \{\CS(G)+\lambda_*\Pi(G)\}=\{\delta_1,G^B\}.
\end{equation*}
The market \(G^B\) has a positive continuous density on \((1,2)\),
no bottom atom, and a positive efficiently served atom at \(2\);
its continuous types are strictly separated. It has positive consumer
surplus and lower profit than \(\delta_1\). The two payoff pairs are
the only single-market points on their supporting line.
\end{appendixproposition}

\begin{proof}[Proof, including Theorem~\ref{thm:main-geometry}(iii)]
\emph{Primitives and construction of \(V\).}
Here
\[
 \nu(\theta)=\frac{\theta(32-\theta)}{80},\qquad
 \nu'(\theta)=\frac25-\frac{\theta}{40},\qquad
 \nu''=-\frac1{40}.
\]
On \([0,2]\), \(m(0)=0\) and \(3/5\leq m'\leq13/20<1\).
Also \(C'(9/16)=\nu(2)=3/4\) and
\(-xC'''/C''=1/2<1\). Thus all primitive assumptions hold.
For \(0\leq z\leq3/4\), \(q(z)=z^2\) and \(\pi(z)=z^3/3\);
both vanish for \(z<0\).

We will choose \(\lambda\) and a strictly increasing, strictly
convex function \(V\) satisfying Lemma~\ref{lem:surplus-lower-bound}.
We require \(V(\theta)=\lambda\pi(\nu(\theta))\) exactly at
types \(1\) and \(2\), and require the residual to vanish along
one separating allocation. Thus it is the coefficient \(\lambda\)
that must allow equality at both endpoints; the cost primitives
remain fixed.

For \(\lambda\in[1/2,1]\), define
\[
\begin{aligned}
 A_\lambda(r)&=2r+(\lambda-4)r^2+\frac23(3-\lambda)r^3,\\
 B_\lambda(r)&=2r-(3-\lambda)r^2,\\
 r_\lambda&=\frac{\lambda}{2(3-\lambda)},\qquad
 \sigma_\lambda=\frac{4-\lambda}{4(3-\lambda)}.
\end{aligned}
\]
The auxiliary function \(r(\theta)\) will be the ratio of virtual
to actual net value among the continuous types of \(G^B\); it is
not a primitive cost parameter. We seek
\(V=A_\lambda(r)\nu^3\) with \(V'=B_\lambda(r)\nu^2\).
Differentiating the first expression and equating it with the second
gives the equation for \(r\):
\begin{equation}
 r'(\theta)=f_\lambda(\theta,r(\theta))
 :=\frac{B_\lambda(r)-3A_\lambda(r)\nu'(\theta)}
          {\nu(\theta)A_\lambda'(r)},\qquad r(1)=r_\lambda.
 \label{eq:nonconvexity-shooting-equation}
\end{equation}

\emph{Choice of the coefficient.}
On \(0\leq r\leq r_\lambda\leq1/4\),
\[
 A_\lambda'(r)=2(1-r)[1-(3-\lambda)r]\geq\frac34,
 \quad
 A_\lambda(r_\lambda)=\frac{\lambda}{3},
 \quad
 B_\lambda(r_\lambda)=\lambda\sigma_\lambda.
\]

For now, extend the right-hand side of \eqref{eq:nonconvexity-shooting-equation}
by replacing its $r$-argument with
$\min\{\max\{r,0\},r_\lambda\}$. This lets us compare solutions
for all candidate coefficients without first proving that they stay
in \([0,r_\lambda]\). The extension is bounded, uniformly
Lipschitz in \(r\), and continuous in \(\lambda\), so its solution
\(r(\theta;\lambda)\) exists on \([1,2]\) and depends continuously
on \(\lambda\). Since zero is a solution, uniqueness keeps every
solution starting at \(r_\lambda>0\) positive.

At the upper boundary,
\[
 f_\lambda(\theta,r_\lambda)
 =\frac{\lambda[\sigma_\lambda-\nu'(\theta)]}
        {\nu(\theta)A_\lambda'(r_\lambda)}.
\]
Since \(\sigma_{1/2}=\nu'(2)=7/20\) and
\(\sigma_1=\nu'(1)=3/8\), the right-hand side
$f_\lambda(\theta,r_\lambda)$ is strictly negative on $[1,2)$
when $\lambda=1/2$ and strictly positive on $(1,2]$
when $\lambda=1$. Comparison with the constant function
\(r_\lambda\) therefore gives
$r(2;1/2)<r_{1/2}$ and $r(2;1)>r_1$.
The strict signs of $f_\lambda(\theta,r_\lambda)$ on
$(1,2)$ rule out equality at $\theta=2$.
Continuous dependence on \(\lambda\) now gives some
\(\lambda_*\in(1/2,1)\) satisfying
\(r(2;\lambda_*)=r_{\lambda_*}\).

For this coefficient, \(\nu'(2)<\sigma_{\lambda_*}<\nu'(1)\).
Since $\nu'$ is strictly decreasing,
$f_{\lambda_*}(\theta,r_{\lambda_*})$ changes sign exactly once,
from negative to positive. The same argument keeps the solution
strictly below \(r_{\lambda_*}\) for types above \(1\), up to
and including the point where the sign changes.
Any first return to $r_{\lambda_*}$ before type $2$
would therefore occur after the sign change. From that point on,
the extended right-hand side is strictly positive at and above
$r_{\lambda_*}$, so the solution would stay above this value.
This contradicts $r(2)=r_{\lambda_*}$. Hence, writing
\(r(\theta)=r(\theta;\lambda_*)\),
\begin{equation}
 r(1)=r(2)=r_{\lambda_*},\qquad
 0<r(\theta)<r_{\lambda_*},\quad1<\theta<2.
 \label{eq:main-nonconvex-r}
\end{equation}
The extension consequently agrees with the original right-hand side
along the solution, which satisfies \eqref{eq:nonconvexity-shooting-equation} on the
entire interval.

\emph{The lower bound.}
Suppress the subscript \(\lambda_*\) and write \(\lambda=\lambda_*\).
Let \(V(\theta)=A(r(\theta))\nu(\theta)^3\). Equation
\eqref{eq:nonconvexity-shooting-equation} gives
$V'=B(r)\nu^2>0$.
Direct differentiation also gives
\[
\begin{aligned}
 (r\nu)'&=\frac{r}{A'(r)}
 \{2-(3-\lambda)r-\nu'[4-(4-\lambda)r]\},\\
 V''&=\frac{\nu r}{1-r}
 \{2-(3-\lambda)r-\nu'[2-(2-\lambda)r]\}.
\end{aligned}
\]
Both braces are positive. Using \(\nu'\leq2/5\), \(r\leq1/4\),
and \(\lambda\geq0\), they are bounded below respectively by
\((2-7r+3\lambda r)/5\geq1/20\) and
\((6-11r+3\lambda r)/5\geq13/20\).
Thus \(r\nu\) strictly increases and \(V\) is strictly convex.

The difference between the weighted payoff at a degenerate market and \(V\)
factors as
\begin{equation}
 \lambda\pi(\nu)-V
 =\frac{\nu^3}{3}(1-r)^2[\lambda-2(3-\lambda)r]\geq0,
 \label{eq:nonconvexity-obstacle-gap}
\end{equation}
with equality at \(\theta=1,2\) by
\eqref{eq:main-nonconvex-r}. For \(0\leq z<\nu(\theta)\), put
\(s=z/\nu(\theta)\). The residual of Lemma~\ref{lem:surplus-lower-bound}
is
\begin{equation*}
 \mathcal R(\theta,z)
 =\nu(\theta)^2
   \frac{(s-r(\theta))^2[3-(3-\lambda)(2r(\theta)+s)]}
        {3(1-s)}\geq0.
\end{equation*}
Indeed, \(2(3-\lambda)r\leq\lambda\), so the bracket is at least
\((3-\lambda)(1-s)>0\). Thus, the residual vanishes exactly when
\(z=r\nu\).
For \(z<0\),
\begin{equation*}
 \mathcal R(\theta,z)
 =V'-\frac{V}{\nu-z}
 >\nu^2[B(r)-A(r)]>0,
\end{equation*}
because \(B(r)-A(r)=r^2[1-\tfrac23(3-\lambda)r]>0\).
The verification lemma now gives the global bound
\begin{equation}
 \CS(G)+\lambda\Pi(G)\geq V(1)
   =\lambda\pi(\nu(1))
 \quad\text{for all }G\in\Delta([1,2]).
 \label{eq:nonconvexity-global-bound}
\end{equation}
The lemma applies to every market, including those with atoms,
support gaps, or a raw revenue curve that is not concave.

\emph{Markets attaining the bound.}
The market \(G^A=\delta_1\) attains the bound. Define the other candidate by
\begin{equation}
 1-G^B(\theta)
 =\exp\left\{-\int_1^\theta
               \frac{ds}{[1-r(s)]\nu(s)}\right\}
 \quad(1\leq\theta<2),\qquad G^B(2)=1.
 \label{eq:main-nonconvex-market-B}
\end{equation}
Its inverse hazard \(d(\theta)=[1-r(\theta)]\nu(\theta)\) is
continuous and bounded below by \(93/320>0\). Hence the integral
in the exponent remains finite up to type \(2\), leaving positive
mass at that type. Differentiating the formula gives a positive
continuous density on \((1,2)\), and \(G^B(1)=0\) rules out a
bottom atom. On its continuous ranks,
$(1-u)(Q^B)'=d(Q^B)$ and $\phi^B=r(Q^B)\nu(Q^B)$.
We have shown that \(r\nu\) is strictly increasing. Thus virtual
net value strictly increases along the continuous part and jumps
upward to \(\nu(2)\) on the top atom. The adjusted raw revenue
curve is therefore concave and the allocation is seller optimal.
Since \(q(z)=z^2\) for positive feasible net values, the continuous
types are strictly separated, and the atom at \(2\) is efficiently
served. The residual vanishes on the continuous part, and the top
atom satisfies \eqref{eq:nonconvexity-obstacle-gap} with equality. There are
also no support gaps. Consequently, \(F=0\), \(\mathcal I=0\),
and \(\Xi=0\) in \eqref{eq:surplus-bound-decomposition}, proving
that \(G^B\) attains the bound.

Conversely, let any market attain \eqref{eq:nonconvexity-global-bound}.
The verification lemma gives lowest type \(1\), connected support,
a concave adjusted raw revenue curve, and a locally absolutely
continuous quantile. Atoms are possible only at \(1\) and \(2\),
where \(V=\lambda\pi\circ\nu\). At almost every rank with
\(1<Q<2\), we must have \(\nu(Q)-\phi>0\): otherwise \(F\)
would equal the strictly positive difference
\(\lambda\pi(\nu(Q))-V(Q)\). The equality conditions then give
$\phi=r(Q)\nu(Q)$ and $(1-u)Q'=d(Q)$
almost everywhere on $\{1<Q<2\}$.
An atom at \(1\) followed by higher types is impossible. Virtual
net value on the atom would be \(\nu(1)\), while just above it
\(r(Q)\nu(Q)\) approaches \(r(1)\nu(1)<\nu(1)\). This would
violate monotonicity of virtual net value. Thus, an atom at \(1\)
forces the market to be \(G^A\).

Without an atom at \(1\), the quantile starts at \(Q(0)=1\).
Until it reaches \(2\), the change of variables
\(t=-\log(1-u)\) gives \(\dot Q=d(Q)\).
The strictly positive, Lipschitz function \(d\)
gives a unique solution reaching \(2\) in finite time. The market
cannot end at an interior type: an atom there would require
\(V=\lambda\pi\circ\nu\), contrary to the strict inequality
already proved. The remaining mass is therefore at \(2\), and the
market is exactly
\eqref{eq:main-nonconvex-market-B}. This exhausts equality among all
feasible distributions.

\emph{Nonconvexity.}
Quality and \(\nu(Q^B)-\phi^B\) are strictly positive on the continuous
component, so \(\CS(G^B)>0=\CS(G^A)\). Equal objective values imply
\(\Pi(G^B)<\Pi(G^A)\). Any single market implementing an interior
point of their payoff chord would also attain the bound. We have
proved that only \(G^A\) and \(G^B\) do so. This proves nonconvexity and
Theorem~\ref{thm:main-geometry}(iii).
\end{proof}

\subsection{Segmentation}
\label{app:segmentation-convex-hull}

\begin{proof}[Proof of Theorem~\ref{thm:main-segmentation}]
Segment payoffs are averaged, and both the markets and their
weights can be chosen freely. Hence the segmented payoff set
is \(K=\co(\mathcal V_{\underline\theta})\).
By Carath\'eodory's theorem, every point of \(K\) can be
implemented with at most three markets
\citep{rockafellar1970}. Since
\(\mathcal V_{\underline\theta}\) and the simplex of mixing
weights are compact, the set of these mixtures is compact.
Thus \(K\) is compact.

Let \(z\) be a boundary point of \(K\), and let \(L\) be
a supporting line through \(z\). Every component with
positive weight in a mixture implementing \(z\) must lie
on \(L\). Otherwise, since all components lie on the same
side of \(L\), their average would lie strictly off the
line. The convex hull of
\(\mathcal V_{\underline\theta}\cap L\) is a line segment
whose endpoints belong to
\(\mathcal V_{\underline\theta}\). Mixing these endpoints
implements \(z\). Thus two markets suffice for every
boundary payoff, including every point on the minimum
consumer surplus boundary.

Now suppose Assumption~\ref{ass:lf-strong-curvature} holds,
and fix
\(p\in(P^{\min}_{\underline\theta},P^{\max})\).
Subsections~\ref{app:pareto-conditioning} and
\ref{app:negative-frontier} show that
\((\Gamma_{\underline\theta}(p),p)\) maximizes
\(\CS+\beta\Pi\) for some finite \(\beta\).
For a segmentation with weights \(\lambda_j\), segment
payoffs \((c_j,p_j)\), and aggregate profit
\(\sum_j\lambda_jp_j=p\), it follows that
\[
 \sum_j\lambda_jc_j+\beta p
 =\sum_j\lambda_j(c_j+\beta p_j)
 \leq \Gamma_{\underline\theta}(p)+\beta p.
\]
Subtracting \(\beta p\) shows that aggregate consumer
surplus cannot exceed \(\Gamma_{\underline\theta}(p)\).

At either endpoint of the profit interval, every segment
with positive weight must have that same profit.
Its consumer surplus is therefore at most
\(\Gamma_{\underline\theta}(p)\), and the same bound holds
for the average. Since a single market attains this bound
at every feasible profit, maximum consumer surplus is
unchanged by segmentation.

Finally, consider any continuous objective that is weakly
increasing in consumer surplus at fixed profit. It attains
a maximum on the compact set \(K\). If \((c,p)\) is a
maximizer, replacing it with the single-market payoff
\((\Gamma_{\underline\theta}(p),p)\) preserves profit and
weakly raises consumer surplus. The replacement is
therefore also a maximizer.
\end{proof}

\section{Comparative Statics}
\label{app:paper-comparative-statics}

We begin by showing how the optimal market changes when the minimum
feasible type rises. Once this minimum exceeds the lowest type selected
without the constraint, the new market is the old one conditioned on
the remaining types. This allows us to compare solutions before they
are truncated by the constraint and then recover the
markets for any feasible minimum. We use this observation below to
vary the profit weight, the minimum feasible type, and the value--cost
mapping.

\begin{appendixproposition}[Conditioning after an increase in the minimum type]
\label{prop:oa-floor-conditioning}
Maintain Assumptions~\ref{ass:main-nu}--\ref{ass:main-curvature}, fix
$\beta\in[0,1)$, and take
$0\leq\underline\theta_1<\underline\theta_2<\bar\theta$.
Let $G_j$ be the optimal markets and
$\ell_j=\max\{\underline\theta_j,\theta_\beta^0\}$.
If $\ell_1=\ell_2$, the markets coincide. Otherwise, with
$a=G_1(\ell_2-)$, one has $0<a<b_1$ and
\begin{equation}
 G_2(E)=\frac{G_1(E\cap[\ell_2,\bar\theta])}{1-a}
 \quad\text{for every Borel set }E.
 \label{eq:oa-floor-conditional-law}
\end{equation}
The rank schedules and top masses satisfy
\begin{align*}
 (Q_2,\phi_2,x_2)(u)&=(Q_1,\phi_1,x_1)(a+(1-a)u),\qquad
 b_2=\frac{b_1-a}{1-a}<b_1,
 \qquad 1-b_2=\frac{1-b_1}{1-a}>1-b_1.
\end{align*}
Thus $G_2$ first-order stochastically dominates $G_1$; its types,
virtual values, and qualities weakly rise by rank, and its separating
region shrinks. Every retained type has unchanged virtual value and
quality under the continuous implementation convention. The hazard
and inverse hazard are unchanged on the common open type support.
\end{appendixproposition}

\begin{proof}
Both markets follow the same solution of the autonomous system,
traced backward from the common top state. If $\ell_1=\ell_2$, they
stop at the same point and hence coincide. Suppose instead that
$\ell_1<\ell_2$. The first market follows the solution farther
backward, so its duration in the variable $t=-\log(1-u)$ is longer:
$T_1>T_2$. Write $\tau=T_1-T_2$. A point reached at time $t_2$ in
the second market is reached at time $t_1=t_2+\tau$ in the first.
Returning to ranks gives
\[
 u_1=1-e^{-(t_2+\tau)}=a+(1-a)u_2,
 \qquad a=1-e^{-\tau}.
\]
In particular, $Q_1(a)=\ell_2$. Since the first market has connected
support and no atom below the highest type, $a=G_1(\ell_2-)$ and
$0<a<b_1$. The second market is therefore the first market conditioned
on $\theta\geq\ell_2$, proving
\eqref{eq:oa-floor-conditional-law}. The same rank
transformation applies to $\phi$ and $x$. The top atom
starts when $a+(1-a)u=b_1$, giving
$b_2=(b_1-a)/(1-a)$ and mass $(1-b_1)/(1-a)$.

Since $a+(1-a)u\geq u$ and the schedules are nondecreasing,
types, virtual values, and qualities are weakly higher
at each rank in the second market. In particular,
$Q_2\geq Q_1$, so $G_2$ first-order stochastically
dominates $G_1$.

At a fixed retained type, virtual value and quality are
unchanged because conditioning changes its rank, but both markets
follow the same solution at that type. The equality extends
to $\ell_2$ by taking limits from above. At the highest type,
both markets assign virtual value $\bar\nu$ and quality
$\bar x$. Finally, on $(\ell_2,\bar\theta)$, conditioning
divides both the density and the survival probability by
$1-a$. This factor cancels in the hazard rate, so the
hazard and inverse hazard are unchanged.
\end{proof}

\subsection{Changing the weight on seller profit}
\label{app:profit-weight-comparison}

Hold the primitives and minimum feasible type fixed, and let
$\beta_2<\beta_1$. Subsection~\ref{app:pareto-conditioning} gives
the payoff comparisons. To compare the lowest selected types, we
first consider the solutions before the minimum-type restriction
truncates them. We call these the \emph{unconstrained solutions},
or free solutions, and denote their lowest types by
$\theta_{\beta_i}^0$. The lowest type in the realized market is
$\max\{\underline\theta,\theta_{\beta_i}^0\}$.

Use actual net value $y=\nu(\theta)$ as the independent variable,
and recall $\rho(y)=\nu'(\nu^{-1}(y))$. As in Part~A, the variable
$B=A_\phi+\xi w_\phi$ combines the marginal effect of changing
induced types with the term from the minimum-type constraint.
Put $\alpha_i=1-\beta_i$, so a lower weight on seller profit
corresponds to a larger $\alpha$.

\begin{appendixlemma}[Profit weights and the unconstrained lower endpoint]
\label{prop:oa-beta-free-path}
Under Assumptions~\ref{ass:main-nu}--\ref{ass:main-curvature}, let
$0\leq\beta_2<\beta_1<1$. Write $B_i(y)$ for the solution traced backward from the common
highest net value, with $B_i(\bar\nu)=\alpha_i\bar x$, until
$B_i=0$. Let $y_i^0=\nu(\theta_{\beta_i}^0)$ be the net value
at this unconstrained lower endpoint. Then
$B_2(y)>B_1(y)$ on the common nonterminal domain,
$y_2^0<y_1^0$, and $\theta_{\beta_2}^0<\theta_{\beta_1}^0$.
\end{appendixlemma}
\begin{proof}
The first-order condition \eqref{eq:gs-state-control} determines
quality $x=X_\alpha(B,y)$. The derivative of its right-hand side
with respect to quality is $D_\alpha>0$. Thus,
$\partial_yX_\alpha=1/D_\alpha>0$, while
\eqref{eq:main-beta-X-comparison} gives
$\partial_\alpha X_\alpha<0$. These comparisons apply throughout
$\alpha\in[\alpha_1,\alpha_2]$ at a state admissible for both
paths. Indeed, $B<\alpha_1q(y)$ implies $B<\alpha q(y)$ throughout
that interval, and the unique root satisfies $B/\alpha<x<q(y)$.

The paths start from different values of $B$ at the highest net
value. We compare them at each common net value by dividing the
equations in Part~A by $\dot y>0$:
\begin{equation*}
 \frac{dB}{dy}=\mathcal R_\alpha(B,y)
 :=\frac{[1-\rho(y)]B+\rho(y)X_\alpha(B,y)}
 {\rho(y)\{y-C'(X_\alpha(B,y))\}}.
\end{equation*}
At a fixed $(B,y)$ admissible for both paths, increasing
$\alpha$ lowers $X_\alpha(B,y)$. Since $0<\rho\leq1$ and
$C'$ is strictly increasing, this lowers the numerator
and raises the denominator of $\mathcal R_\alpha$.
Both are positive, so
$\mathcal R_{\alpha_2}(B,y)<\mathcal R_{\alpha_1}(B,y)$.

Since $B_2(\bar\nu)>B_1(\bar\nu)$, continuity gives
$B_2>B_1$ just below $\bar\nu$. Suppose the paths first
meet at $y_*$ as we move toward lower net values.
Then $B_2-B_1$ is zero at $y_*$ and positive immediately
to its right, so $(B_2-B_1)'(y_*)\geq0$. But the
differential equations give
\[
 (B_2-B_1)'(y_*)
 =\mathcal R_{\alpha_2}(B_1(y_*),y_*)
  -\mathcal R_{\alpha_1}(B_1(y_*),y_*)<0,
\]
a contradiction. Thus, $B_2>B_1$ throughout their common
interior domain. This implies $y_2^0\leq y_1^0$. Indeed,
if $y_2^0>y_1^0$, taking $y\downarrow y_2^0$ would give
$0\geq B_1(y_2^0)>0$, a contradiction.

To rule out equality, suppose both endpoints equal $y_0$.
The bottom qualities are positive and strictly below
efficient quality, so both paths have finite right
derivatives there. The differential equations give
\[
 B_2'(y_0+)
 =\mathcal R_{\alpha_2}(0,y_0)
 <\mathcal R_{\alpha_1}(0,y_0)
 =B_1'(y_0+).
\]
Since both paths are zero at $y_0$, this would imply
$B_2<B_1$ just above $y_0$, contradicting the order already
proved. Hence $y_2^0<y_1^0$. Since $\nu$ is strictly
increasing, $\theta_{\beta_2}^0<\theta_{\beta_1}^0$ as well.
\end{proof}

\begin{proof}[Proof of Proposition~\ref{prop:main-beta-comparative}]
Suppose first that \(\beta_1<1\). By
Lemma~\ref{prop:oa-beta-free-path}, the unconstrained lowest type
is strictly lower at \(\beta_2\). Taking the maximum of each
unconstrained lowest type and the common minimum feasible type,
as in \eqref{eq:gs-realized-lower-support}, gives the claimed
support order. Equality holds exactly when
\(\underline\theta\geq\theta_{\beta_1}^0\). This includes the
case in which the minimum feasible type equals the unconstrained
lowest type, so the constraint binds with zero multiplier.

Subsection~\ref{app:pareto-conditioning} gives the profit and
consumer surplus comparisons, with strict inequalities when the
markets differ. It also shows that a market selected at two
distinct weights must have a binding minimum-type constraint.
Writing \(P_i=\Pi(G_i)\), optimality at \(\beta_1\) gives
\begin{equation*}
 \TS(G_1)-\TS(G_2)\geq(1-\beta_1)(P_1-P_2)\geq0,
\end{equation*}
where the last inequality is strict if the markets differ, since
then \(P_1>P_2\) and \(\beta_1<1\).

If \(\beta_1=1\), market 1 is concentrated at the highest type,
whereas market 2 is nondegenerate. This gives the support order,
and the payoff comparisons from
Subsection~\ref{app:pareto-conditioning} still apply. The market
concentrated at the highest type uniquely maximizes total surplus,
so its total surplus is strictly larger as well.
\end{proof}

\subsection{Private values and extensions to type-dependent costs}
\label{app:private-value-comparative-statics}

We use the same comparison argument for
Proposition~\ref{prop:mr-comparative-statics} and the two results in
Section~\ref{sec:mc-welfare}. Maintain
Assumptions~\ref{ass:main-nu}--\ref{ass:main-curvature}, fix
\(0\leq\beta_2<\beta_1<1\), and put
\[
 \alpha_i=1-\beta_i,\qquad \tau=\nu^{-1},\qquad
 \rho=\nu'\circ\tau,\qquad \lambda=(1-\rho)/\rho.
\]
For a common feasible minimum \(\underline\theta\), let
\(G_i,Q_i,\phi_i,x_i,\ell_i,b_i\) denote the optimal markets,
their rank schedules, their lowest types, and the ranks at which their
atoms at the highest type begin. Write \(x_i^V(\theta)\) and \(U_i(\theta)\)
for the quality and rent of type \(\theta\). Superscript \(0\) denotes
the optimum when the minimum feasible type is zero, so
\(\ell_i^0=\theta_{\beta_i}^0\). At support endpoints, allocations
are given by their continuous one-sided limits. Under private values,
\(\theta=v\), \(\bar\theta=\bar\nu=1\), and \(\lambda=0\).

\begin{proof}[Proof of Proposition~\ref{prop:mr-comparative-statics}
and the extension beyond private values]
\emph{1. Comparing the free solutions.}
We first compare the solutions before either is truncated by a minimum
type. Write \(x=X_\alpha(B,y)\) for the quality determined by
\eqref{eq:gs-state-control}. Its derivatives satisfy
\[
 \partial_yX_\alpha=\frac1{D_\alpha}>0,\qquad
 \partial_\alpha X_\alpha=-\frac{xC''(x)}{D_\alpha}<0,
\]
where \(D_\alpha>0\) is given in \eqref{eq:gs-control-derivative}.
Since each free solution \(B_i(y)\) is increasing, we may instead
use \(B\) as the independent variable. Denote the resulting net value
and quality by \(\mathcal Y_i(B),\mathcal X_i(B)\). Inverting the
order in Lemma~\ref{prop:oa-beta-free-path} gives the net-value
comparison $\mathcal Y_2(B)<\mathcal Y_1(B)$ and $
 \mathcal X_2(B)<\mathcal X_1(B)$, $0\leq B\leq\alpha_1\bar x$.
To obtain the quality comparison, hold \(B\) fixed, first increase
\(y\) from \(\mathcal Y_2(B)\) to \(\mathcal Y_1(B)\) at
\(\alpha_2\), and then decrease \(\alpha\) to \(\alpha_1\).
All intermediate states are admissible, and each change raises
\(X_\alpha(B,y)\). At \(B=0\), the net-value comparison is the
strict order of the unconstrained lowest net values; at
\(B=\alpha_1\bar x\), only solution 1 has reached the highest type.
In particular, the qualities at the unconstrained lower endpoints,
\(a_i:=\mathcal X_i(0)>0\), satisfy \(a_2<a_1\). These comparisons
require no additional slope or concavity restriction.

\emph{2. Quality at a fixed type.}
To compare qualities at the same type, use quality as the independent
variable and let \(Y_i(x)\) denote net value along the free solution.
With \(\kappa=xC'''/C''\) and
\(z=(y-C'(x))/(xC''(x))\), equations
\eqref{eq:rw-two-state-system} and \eqref{eq:gs-x-log} give
\begin{equation}
\label{eq:mc-quality-coordinate}
\begin{aligned}
 Y_i'(x)&=F_{\alpha_i}(x,Y_i(x);\lambda(Y_i(x))),
 &Y_i(\bar x)&=\bar\nu,\\
 F_\alpha(x,y;\lambda)
 &=zC''(x)\frac{1+\alpha+\kappa(x)z}
                     {1+\alpha\lambda+(1-\lambda)z}.
\end{aligned}
\end{equation}
The last argument of \(F_\alpha\) is the value of \(\lambda\), held
fixed when we differentiate with respect to \(\alpha\). Along an
admissible solution, \(z=\alpha-B/x\in[0,\alpha]\), so the
denominator is \(1+z+\lambda(\alpha-z)\geq1\). At fixed
\((x,y,\lambda)\),
\[
 \partial_\alpha F_\alpha
 =\frac{zC''(x)(1+z)}{[1+\alpha\lambda+(1-\lambda)z]^2}
   \left[1-\lambda\left(1+\frac{\kappa z}{1+z}\right)\right].
\]
Under private values, this derivative is
\(zC''(x)/(1+z)>0\) below the highest type, with no upper bound
on \(\kappa\) needed. Under \eqref{eq:mc-slope}, it is positive
because \(\lambda(1+\bar\kappa/2)<1\) and
\(z/(1+z)\leq1/2\).

Both solutions are defined on $[a_1,\bar x]$. Along solution 1,
$z\leq\alpha_1$, so the comparison of $F_\alpha$ applies for every
$\alpha\in[\alpha_1,\alpha_2]$. On the compact region between the
two solutions, $0\leq z\leq\alpha_2$, and
$F_{\alpha_2}(x,y;\lambda(y))$ is Lipschitz in $y$. Indeed,
$\lambda$ is $C^1$, quality is bounded away from zero, and the
denominator is bounded away from zero.

Set \(r=\bar x-x\), so that \(r\) increases as we move from the
highest quality toward the lowest, and put
\(w(r)=Y_1(x)-Y_2(x)\). Subtracting the two differential equations
gives \(w'=A(r)w+f(r)\), where \(A\) is bounded by the preceding
Lipschitz property, \(w(0)=0\), and
\[
 f(r)=F_{\alpha_2}(x,Y_1(x);\lambda(Y_1(x)))
      -F_{\alpha_1}(x,Y_1(x);\lambda(Y_1(x)))>0,\quad r>0.
\]
The integrating-factor formula therefore gives
\(Y_2(x)<Y_1(x)\) for \(a_1\leq x<\bar x\). Now fix a net value
\(y\) in the interior of both free supports. The quality
\(Y_1^{-1}(y)\) lies in the domain of both solutions, and
\(Y_2(Y_1^{-1}(y))<y\). Since \(Y_2\) is increasing, market 2
assigns a higher quality to a buyer with net value \(y\).
Proposition~\ref{prop:oa-floor-conditioning} shows that conditioning
leaves quality at each retained type unchanged. Thus the fixed-type
comparison holds at every common feasible minimum, under either
private values or (S).

\emph{3. Distributions, rents, and payoffs.}
The support formula gives
\(\ell_i=\max\{\underline\theta,\ell_i^0\}\), with
\(\ell_2^0<\ell_1^0\). On the continuous support, define
\[
 d_i(\theta)=\nu(\theta)-C'(x_i^V(\theta))
            =\frac{1-G_i(\theta)}{g_i(\theta)},\qquad
 t_i(\theta)=\int_{\ell_i}^{\theta}\frac{ds}{d_i(s)}.
\]
By Step 2 and strict increase of \(C'\), virtual net value is higher
in market 2, so \(d_2<d_1\) on the common interior. The hazard,
which is the reciprocal of \(d_i\), is therefore higher in market 2.
Since there is no atom at the lowest type and
\(\dot\theta=d_i(\theta)\), the function \(t_i\) gives the
transformed rank \(t=-\log(1-u)\) of each type. The lower endpoint
is weakly smaller in market 2, and its integrand is strictly larger
on the common interior. Hence \(t_2(\theta)>t_1(\theta)\) for
\(\ell_1<\theta\leq\bar\theta\); the improper integrals at the
highest type are finite. In particular, \(T_2>T_1\) and
\(b_2>b_1\), where \(T_i=t_i(\bar\theta)\) and
\(b_i=1-e^{-T_i}\).

Inverting these rank maps gives \(Q_2<Q_1\) on \((0,b_2)\).
Both quantiles equal \(\bar\theta\) on \([b_2,1]\), and the
comparison at rank zero is strict when
\(\underline\theta<\ell_1^0\). For rents, the envelope formula gives
\[
 U_2(\theta)-U_1(\theta)
 =\int_{\ell_2}^{\ell_1}x_2^V(s)\,ds
  +\int_{\ell_1}^{\theta}[x_2^V(s)-x_1^V(s)]\,ds\geq0,
\]
with strict inequality for \(\theta>\ell_1\), since the second
integrand is strictly positive on the common interior. Finally, the
different masses at the highest type imply that the markets differ.
Proposition~\ref{prop:main-beta-comparative} then gives the strict
payoff comparisons. This proves parts (i)--(v) of the private-value
proposition and part (i) of its extension.

\emph{4. Quality by rank.}
Assume \(m''\leq0\), as under private values. Then
\(\rho'(y)=-m''(\tau(y))/\rho(y)\geq0\). Along free solution
\(i\), the rate of change of \(B\) with respect to transformed rank is
\[
 v_i(B)=B+\rho(\mathcal Y_i(B))[\mathcal X_i(B)-B].
\]
Step 1 orders both net values and qualities at a common \(B\).
Since \(\rho\) is nondecreasing and \(\mathcal X_i(B)\geq B\),
these orders give \(0<v_2(B)<v_1(B)\) on
\([0,\alpha_1\bar x]\). Thus \(s_2(B)>s_1(B)\) for \(B>0\),
where \(s_i(B)=\int_0^B dA/v_i(A)\) is the transformed rank at
which solution \(i\) reaches \(B\). Solution 2 also traverses
\((\alpha_1\bar x,\alpha_2\bar x]\), so \(T_2^0>T_1^0\), where
\(T_i^0=s_i(\alpha_i\bar x)\).

At a common transformed rank \(0<t\leq T_1^0\), inversion gives
\(B_2(t)<B_1(t)\). Combining this inequality with Step 1 and the
strict increase of quality along solution 1 gives
\[
 \mathcal X_2(B_2(t))<\mathcal X_1(B_2(t))
                         <\mathcal X_1(B_1(t)),
\]
and the same argument orders net values. At rank zero, use the
comparisons at \(B=0\) from Step 1. After \(T_1^0\), solution 1
has reached the highest type, while solution 2 remains below it until
\(T_2^0\). Applying the increasing functions \(C'\) and \(\tau\)
to quality and net value gives the orders of virtual net values and types.
If \(\underline\theta\leq\ell_2^0\), neither free market is
truncated, so these comparisons prove the rank result without (S).

Under private values or (S), suppose instead that
\(\underline\theta>\ell_1^0\). Both markets then begin at the
same type, in the interior of both free supports. Step 2 gives higher
quality and virtual net value in market 2 at rank zero. Since
\(b_2>b_1\), on \((b_1,b_2)\) market 1 has already reached its highest
type while market 2 is still separating, so the inequalities are
reversed. Continuity implies that both schedules cross. In the
remaining case, \(\ell_2^0<\underline\theta\leq\ell_1^0\), only
market 2 is strictly truncated; no complete order of qualities by rank
is claimed. This proves the remaining assertions.
\end{proof}

\paragraph{Original private-value units.}
The normalization in Appendix~\ref{app:payoff-normalization} translates these
comparisons back to the original units. If \(v^0=K'_+(0)\) and
\(d=\bar v-v^0>0\), it maps \(v\mapsto(v-v^0)/d\) and
\(\phi\mapsto(\phi-v^0)/d\), divides inverse hazards by \(d\),
and multiplies hazards by \(d\). Physical quality, ranks, and masses
are unchanged, and both payoffs are divided by \(d\). Since
\(d>0\), every comparison retains its sign, including those
involving atoms and binding minimum types.

\paragraph{A quantitative substitute for concavity.}
The rank comparison used concavity of \(m\) only to establish
\(v_2(B)<v_1(B)\). For
\(v_\alpha(B,y)=B+\rho(y)[X_\alpha(B,y)-B]\),
\[
 \partial_\alpha v_\alpha<0,\qquad
 \partial_yv_\alpha
 =\frac{\rho(y)^2-m''(\tau(y))D_\alpha(X_\alpha-B)}
        {\rho(y)D_\alpha}.
\]
Thus the same argument applies if
\(m''(\tau(y))D_\alpha(x,B)(x-B)\leq\rho(y)^2\) at every
admissible comparison state. Under this condition, increasing \(y\)
from \(\mathcal Y_2(B)\) to \(\mathcal Y_1(B)\) at fixed
\(\alpha_2\) weakly raises the rate of change of \(B\), and then
decreasing \(\alpha\) to \(\alpha_1\) strictly raises it.
A sufficient condition stated directly in terms of the primitives is
\[
 L_m(2+\bar\kappa)M_C\leq(1-\bar p)^2,\qquad
 M_C=\sup_{0<x\leq\bar x}xC''(x),\qquad
L_m=\max_{0\leq\theta\leq\bar\theta}\max\{m''(\theta),0\},
\]
with finite \(\bar\kappa,M_C\). To see sufficiency, use
\(x-B\leq x\), \(D_\alpha\leq(2+\bar\kappa)C''\), and
\(\rho\geq1-\bar p\). This condition can replace (R); together
with (S), it gives the full extension.

The curvature term in (S) cannot generally be dropped.
Take \(m(\theta)=s\theta\) and an admissible \(C\) with
\(C'''>0\). When \(s=1/2\), we have \(\lambda=1\) and
\(\partial_\alpha F_\alpha=-\kappa z^2C''/(1+\alpha)^2<0\)
below the highest type. The comparison in Step 2 is therefore
reversed: market 2 assigns lower quality at a fixed type. Hold
\(C,\bar x,\bar\nu\) fixed and set
\(\bar\theta=\bar\nu/(1-s)\). Continuous dependence of the
solutions to \eqref{eq:mc-quality-coordinate}, and of their inverses,
preserves the strict comparison at a common interior net value for
some \(s<1/2\) sufficiently close to one half. Hence
\(m'<1/2\) alone does not suffice for general \(C\).

\paragraph{Higher expected surplus with lower individual rents.}
\begin{proof}
Suppose \(\bar r=\nu'(\bar\theta)=1-m'(\bar\theta)<1/2\).
We first compare qualities just below the highest type, then choose
a minimum type that retains only types for which this comparison
holds. On free solution \(i\), \eqref{eq:gs-x-log} and
\eqref{eq:gs-top-boundary} give the derivative of quality as the
separating interval approaches its upper endpoint:
\[
 \dot x(T_i^{0}-)
 =\bar x\frac{\bar r+\alpha_i(1-\bar r)}{1+\alpha_i}
 =:c_i>0.
\]
Since \(\dot\theta=d_i^0(\theta)\),
\[
 \frac{d}{d\theta}[d_i^0(\theta)^2]
 =2\{\nu'(\theta)d_i^0(\theta)-C''(x)\dot x\}
 \longrightarrow-2C''(\bar x)c_i.
\]
These limits follow from Proposition~\ref{prop:rw-endpoints} and
the continuous quality equation. Integrating up to
\(d_i^0(\bar\theta-)=0\) gives
\[
 \lim_{\theta\uparrow\bar\theta}
 \frac{d_i^0(\theta)^2}{\bar\theta-\theta}=2C''(\bar x)c_i,
 \qquad
 c_2-c_1=
 \frac{\bar x(1-2\bar r)(\alpha_2-\alpha_1)}
      {(1+\alpha_1)(1+\alpha_2)}>0.
\]
Thus \(d_2^0>d_1^0\) on some common upper interval
\((\eta,\bar\theta)\). Since \(d_i^0\) is actual net value minus
virtual net value, quality is lower in market 2 throughout this
interval. Choose
\(\theta^\dagger\in(\max\{0,\eta,\ell_1^0,\ell_2^0\},\bar\theta)\).
For every common feasible minimum above this threshold, conditioning
leaves the allocation curves unchanged and makes both selected
minima equal to \(\underline\theta\). The envelope formula then gives
\[
 U_2(\theta)-U_1(\theta)
 =\int_{\underline\theta}^{\theta}[x_2^V(s)-x_1^V(s)]\,ds<0,
 \quad \underline\theta<\theta\leq\bar\theta.
\]
Rents are zero at the lowest type in both markets, and both assign
quality \(\bar x\) to the highest type.

To compare distributions, use the absence of an atom at the lowest
type to write
\[
 1-G_i(\theta)=
 \exp\left\{-\int_{\underline\theta}^{\theta}
                    \frac{ds}{d_i(s)}\right\},
 \qquad \underline\theta\leq\theta<\bar\theta.
\]
The inequality \(d_2>d_1\) makes the integral smaller in market 2,
so \(G_2<G_1\) on the common interior. Taking left limits at the
highest type gives
\[
 G_i(\{\bar\theta\})
 =1-G_i(\bar\theta-)
 =\exp\left\{-\int_{\underline\theta}^{\bar\theta}
                    \frac{ds}{d_i(s)}\right\}.
\]
The integrals up to \(\bar\theta\) are finite and strictly ordered,
so market 2 has the larger atom at the highest type. In particular,
the markets differ, and Proposition~\ref{prop:main-beta-comparative}
gives the strict payoff comparisons. Since rent is strictly lower
at the highest type, to which \(G_1\) assigns positive mass, we also
have \(\int U_2\,dG_1<\int U_1\,dG_1\). The threshold may depend
on the fixed pair of weights. When \(m'(\bar\theta)=1/2\), the
leading-order comparison is inconclusive.
\end{proof}

\subsection{Raising the minimum feasible type}
\label{app:paper-eligibility}

Fix the primitives and $\beta\in[0,1)$. Raising the minimum feasible
type leaves the market unchanged unless it raises the selected lower
endpoint from $\ell_1$ to $\ell_2$. When it does,
Proposition~\ref{prop:oa-floor-conditioning} identifies market~2 as
market~1 conditional on $\theta\geq\ell_2$. Every retained type
receives the same quality. This proves the distribution and allocation
claims of Proposition~\ref{prop:main-eligibility}. We now prove its
payoff claims, distinguishing the change in each retained buyer's rent
from the change in average rent across the two populations.

\begin{proof}[Completion of Proposition~\ref{prop:main-eligibility}]
Assume $\ell_2>\ell_1$, and put $a=G_1(\ell_2-)\in(0,b_1)$.
Along market~1, define $p(u)=\pi(\phi_1(u))$ and
$s(u)=\nu(Q_1(u))x_1(u)-C(x_1(u))$. Their averages are seller
profit and total surplus, respectively. We first show that both
functions are strictly increasing below $b_1$, the rank at which the
top atom begins, and constant thereafter.

For $p$, this follows from the strict increase in virtual value and
$\pi'=q>0$. For $s$, both net value and quality increase below
$b_1$, and quality remains below first best. To verify that these
changes raise total surplus, write $y=\nu(Q_1)$. If
$0\leq u<v<b_1$, then $y(v)>y(u)$, $x_1(u)>0$, and
$C'(z)\leq\phi_1(v)\leq y(v)$ for
$x_1(u)\leq z\leq x_1(v)$. Therefore
$s(v)-s(u)\geq[y(v)-y(u)]x_1(u)>0$.
Both functions are constant on the top atom because type and quality
are constant there. By the rank identity in
Proposition~\ref{prop:oa-floor-conditioning}, market~2's profit and
total surplus are $(1-a)^{-1}\int_a^1p$ and
$(1-a)^{-1}\int_a^1s$, whereas market~1 averages each function over
all ranks. Since $0<a<b_1$, the removed ranks contribute strictly less
than some of the retained ranks. Both averages therefore rise strictly.

To compare consumer surplus, we use the fact that market~2 remains
feasible under the original, lower minimum type. Optimality of
market~1 gives
$\CS(G_1)+\beta\Pi(G_1)\geq\CS(G_2)+\beta\Pi(G_2)$.
If $\beta>0$, the strict increase in profit implies a strict fall in
consumer surplus. If $\beta=0$, equality would make the two distinct
markets consumer surplus maximizers, contrary to uniqueness. Thus,
consumer surplus falls in this case as well.

Finally, consider the rent and payment of a retained type. The seller
sets the lowest realized type's rent to zero. Raising that type to
$\ell_2$ removes the rent it received in market~1. Since retained
qualities are unchanged, it removes the same amount from every
retained buyer's rent. Formally, the envelope formula
\eqref{eq:gs-envelope} gives
$U_j(\theta)=\int_{\ell_j}^{\theta}x_j(z)\,dz$.
Subtracting and using
$t_j=\theta x_j-U_j$ yields, for every retained type,
\begin{equation*}
 U_2(\theta)=U_1(\theta)-U_1(\ell_2),\qquad
 t_2(\theta)=t_1(\theta)+U_1(\ell_2),\qquad U_1(\ell_2)>0.
\end{equation*}
The amount $U_1(\ell_2)$ is strictly positive because quality at the
lowest type is positive and $\ell_2>\ell_1$. The identities hold at
the cutoff and also at the top atom. They compare the same buyer
across markets. The aggregate comparisons above also account for the
change in the distribution of buyers.
\end{proof}

\subsection{Value--cost co-movement}
\label{app:paper-value-cost}

We compare two environments, $L$ and $H$, holding fixed the quality
cost $C$, the attainable net-value interval $[0,\bar\nu]$, and the
minimum feasible net value $\underline y$. For a buyer with net
value $y$, write $\tau_i(y)=\nu_i^{-1}(y)$ for gross willingness to
pay and $k_i(y)=m_i(\tau_i(y))=\tau_i(y)-y$ for the type-dependent
unit cost. The derivative $\lambda_i(y)=k_i'(y)$ measures the increase
in this cost associated with a unit increase in net value. Thus,
\[
 \tau_i'(y)=1+\lambda_i(y),\qquad
 \rho_i(y)=\frac{1}{1+\lambda_i(y)}
          =\nu_i'(\tau_i(y)).
\]
Stronger co-movement means $\lambda_H(y)\geq\lambda_L(y)$ pointwise,
which gives $\rho_H\leq\rho_L$. Since $\tau_i(0)=k_i(0)=0$,
integration also gives
\begin{equation}
 \tau_H(y)\geq\tau_L(y),\qquad k_H(y)\geq k_L(y)
 \quad(0\leq y\leq\bar\nu).
 \label{eq:oa-loading-level-order}
\end{equation}

We first implement the same virtual-net-value schedule in both
environments. Quality by rank and profit are then unchanged, and
comparing the actual types that support the schedule will give the
fixed-profit frontier comparison. We then allow both markets to be
reoptimized at the same profit weight and compare buyers at a common
net value or a common gross willingness to pay.

\paragraph{The same screening schedule in two environments.}
Let $Y_i^\phi=\nu_i(Q_i^\phi)$ denote the actual net value that supports
schedule $\phi$ in environment $i$. The next lemma shows that a
schedule feasible in $L$ remains feasible in $H$, where it is supported
by weakly higher net values.

\begin{appendixlemma}[A common screening schedule and the fixed-profit frontier]
\label{prop:oa-loading-profile}
Under Assumptions~\ref{ass:main-nu}--\ref{ass:main-C}, let $\phi$ be
an $L$-feasible regular schedule that equals $\bar\nu$ on a
positive-length top interval, or any $L$-feasible bounded nonnegative
schedule $\phi\in\Phi$. It induces an $H$-feasible market,
and, at every rank,
\begin{align*}
 Y_H^\phi&\geq Y_L^\phi,\qquad Q_H^\phi\geq Q_L^\phi,
 & k_H(Y_H^\phi)&\geq k_L(Y_L^\phi),
 \\
 x_H^\phi&=x_L^\phi=q(\phi),\qquad \Pi_H(\phi)=\Pi_L(\phi),
 & \CS_H(\phi)&\geq\CS_L(\phi).
\end{align*}
Consequently, with $\Gamma_i$ as in
Proposition~\ref{thm:main-dependence},
\begin{equation}
 \Gamma_H(p)\geq\Gamma_L(p)
 \quad\text{for every profit attainable in }L.
 \label{eq:oa-loading-frontier-order}
\end{equation}

\end{appendixlemma}
\begin{proof}
Begin with a schedule equal to $\bar\nu$ on an interval of highest
ranks. Write $t=-\log(1-u)$, and let $T$ be the transformed rank at
which that interval begins. The tail equation becomes
\begin{equation*}
 \dot Y_i=\rho_i(Y_i)\{Y_i-\widehat\phi(t)\},\qquad
 Y_i(T)=\bar\nu,\qquad \widehat\phi(t)=\phi(1-e^{-t}).
\end{equation*}
To compare the solutions, trace them backward from their common
endpoint. Since $Y_L\geq\widehat\phi$ and $\rho_H\leq\rho_L$,
the derivative in reverse time for environment $H$, evaluated at the
$L$ solution, is weakly larger:
$-\rho_H(Y_L)(Y_L-\widehat\phi)\geq
-\rho_L(Y_L)(Y_L-\widehat\phi)$.
Scalar comparison therefore gives $Y_H\geq Y_L$.

This comparison also applies when $\phi$ is unbounded below at rank
zero. Indeed, $\rho_i$ is $C^1$ on the compact net-value interval,
so the right-hand side of the differential equation is Lipschitz in
its state, with a bound proportional to $1+|\widehat\phi|$.
The interval of highest ranks has positive length, so $T$ is finite.
Regular schedules are integrable in rank, and the change of variables
$dt=du/(1-u)$ introduces a factor of at most $e^T$ before that interval.
The bound is therefore integrable on $[0,T]$. Applying Gronwall's
inequality to $(Y_L-Y_H)_+$ in reverse time gives the same order. The $L$
solution is a lower bound for the $H$ solution, and the constant
$\bar\nu$ is an upper bound, so the latter remains in $[0,\bar\nu]$.
Moreover, $Y_H\geq Y_L\geq\widehat\phi$ implies that
$Q_H=\tau_H(Y_H)$ is nondecreasing. It solves the tail equation and
is therefore the required regular quantile. Since $Y_H(0)\geq Y_L(0)\geq\underline y$, the $H$ market
satisfies the same minimum-net-value restriction.

Equation~\eqref{eq:oa-loading-level-order}, together with the
monotonicity of $\tau_H,k_H$, gives the gross type and cost results.
Quality and profit depend only on the common schedule $\phi$.
Finally, $\CS_i(\phi)=\int_0^1(Y_i^\phi-\phi)q(\phi)\,du$ gives
the consumer surplus order.

Now take a bounded nonnegative schedule $\phi\in\Phi$, which need
not equal $\bar\nu$ on any interval. Set $\phi_n=\phi$ below
$1-1/n$ and $\phi_n=\bar\nu$ above it. These profiles are monotone,
dominate $\phi$, remain feasible for the same minimum net value in
$L$, and converge to it in $L^1$. Apply the preceding argument to each profile and then use
Proposition~\ref{lem:gs-compact-continuous} in both environments.
Quantiles converge at every $u<1$, including zero, and payoffs
converge. The feasibility and payoff results therefore pass to
the limit, as do the net-value, gross type, and cost comparisons at
these ranks. At rank one, the limit in the tail equation gives
$Y_i^\phi(1)=\phi(1)$, which establishes the remaining endpoint
comparisons.

It remains to compare maximum consumer surplus at a fixed profit.
Fix any profit $p$ attainable in $L$. Compactness and payoff
continuity, from Proposition~\ref{lem:app-compact-density}, ensure
that consumer surplus is maximized among $L$-markets with profit
$p$. Regularization preserves profit and feasibility and weakly
raises consumer surplus; because the original market is a maximizer,
surplus must be unchanged. Truncation then gives a profile in
$\Phi$ with payoff $(\Gamma_L(p),p)$. Its implementation in $H$
has the same profit and weakly higher consumer surplus, proving
\eqref{eq:oa-loading-frontier-order}. This argument includes both
profit endpoints.
\end{proof}

\paragraph{Reoptimized markets.}
Let both institutions now optimize at the same weight
$\beta\in[0,1)$, under Assumption~\ref{ass:main-curvature}.
Their virtual-value schedules may differ. We compare the free
solutions, before the common minimum net value truncates them, using
quality as the independent variable. In these coordinates the two
solutions have the same top endpoint, and the differential equation
is nonsingular there.

Put $\alpha=1-\beta\in(0,1]$. The first-order condition
\eqref{eq:gs-state-control} expresses actual net value in terms of
quality and the marginal-welfare variable $B$:
\begin{equation*}
 \mathcal Y_\alpha(x,B)=C'(x)+C''(x)(\alpha x-B).
\end{equation*}
Its derivative with respect to quality is $D_\alpha$ in
\eqref{eq:gs-control-derivative}, which is positive on
$0\leq B\leq\alpha x$ by
Lemma~\ref{lem:gs-state-control-monotone}. Its derivative with
respect to $B$ is $-C''(x)<0$.

\begin{appendixproposition}[Reoptimized treatment and selection]
\label{prop:oa-loading-treatment}
Maintain Assumptions~\ref{ass:main-nu}--\ref{ass:main-curvature} and
$0\leq\beta<1$. The optimal $H$ market has a weakly higher selected
lowest net value and gross type. At a common retained net value, its
quality and virtual value are weakly lower, and its actual--virtual
difference and deadweight loss weakly higher. When both lower endpoints
are freely chosen, quality at the selected lowest type is weakly higher
in $H$.

More precisely, let $B_i^x(x)$ denote the free, zero-floor solution
in quality coordinates, put $y_i(x)=\mathcal Y_\alpha(x,B_i^x(x))$,
and denote its unconstrained lower endpoint, where $B=0$, by
$(x_i^0,y_i^0)$. Then
\begin{equation*}
 B_H^x(x)\leq B_L^x(x),\qquad y_H(x)\geq y_L(x),\qquad
 x_H^0\geq x_L^0,\quad y_H^0\geq y_L^0,
\end{equation*}
where the solution inequalities hold on the common free quality range.
At a common retained net value, $B_H(y)\leq B_L(y)$. At a common gross
willingness to pay in both realized supports, quality, information rent, and total surplus are
weakly lower in $H$.
\end{appendixproposition}
\begin{proof}
We first compare the free solutions at a common quality. The net
values assigned that quality may differ, so we will subsequently
invert this comparison to study treatment at a common net value.
Divide \eqref{eq:main-two-state-B} by \eqref{eq:gs-x-log} along
the free solution of
Proposition~\ref{prop:app-stopped-characterization}:
\begin{equation*}
 \frac{dB_i}{dx}=\mathcal G_i(x,B_i)
 :=\frac{D_\alpha(x,B_i)}{C''(x)}
 \frac{x+\lambda_i(\mathcal Y_\alpha(x,B_i))B_i}
 {(1+\alpha)x+[\lambda_i(\mathcal Y_\alpha(x,B_i))-1]B_i},
 \qquad B_i(\bar x)=\alpha\bar x.
\end{equation*}
The denominator is at least $(1+\alpha\lambda)x>0$ when
$\lambda<1$; when $\lambda\geq1$, both terms in the denominator
are nonnegative and the first is positive. At a common admissible
state,
\begin{equation*}
 \frac{\partial}{\partial\lambda}
 \frac{x+\lambda B}{(1+\alpha)x+(\lambda-1)B}
 =\frac{B(\alpha x-B)}{[(1+\alpha)x+(\lambda-1)B]^2}\geq0.
\end{equation*}
It follows that $\mathcal G_H\geq\mathcal G_L$ at every common
admissible $(x,B)$. The denominator remains positive at the common
top state, so this equation, unlike the equation in value coordinates,
is nonsingular there.

On compact intervals of positive qualities, $\mathcal G_i(x,B)$
is locally Lipschitz in $B$: the function $\lambda_i$ is $C^1$,
and the denominator is positive. Scalar comparison in the reverse
coordinate $\bar x-x$ therefore gives $B_H^x\leq B_L^x$ until
one of the solutions first reaches zero. Both solutions, and the
states between them, satisfy $0\leq B\leq\alpha x$ and
$0\leq\mathcal Y_\alpha(x,B)\leq\bar\nu$, so all states used
in this comparison are admissible. Since $\mathcal Y_\alpha$
decreases in $B$, we also have $y_H(x)\geq y_L(x)$.

These inequalities order the lower endpoints as well. Both $B_i^x$
increase strictly with quality. If $x_H^0<x_L^0$, then at $x_L^0$
the $H$ solution would be positive but no larger than the $L$
solution, which is zero. Hence $x_H^0\geq x_L^0$. At a free
lower endpoint, $y_i^0=\mathcal Y_\alpha(x_i^0,0)$, and this
expression is strictly increasing in $x_i^0$. Thus, $y_H^0\geq y_L^0$.
The support formula and conditioning give
\begin{equation}
 \ell_i^y=\max\{\underline y,y_i^0\},\qquad
 \ell_i=\tau_i(\ell_i^y),\qquad
 \ell_H^y\geq\ell_L^y,\quad\ell_H\geq\ell_L.
 \label{eq:app-loading-selected-support}
\end{equation}
The support formulas also apply when $\underline y=y_i^0$,
where the constraint binds with zero multiplier.
The ordering $x_H^0\geq x_L^0$ compares qualities at the
lowest selected types when both markets start at their
free endpoints.

Fix a net value $y$ present in both markets, and set
$x=x_H(y)$. Since $x\geq x_H^0\geq x_L^0$, both free
solutions are defined at this quality. We use the free
solutions because $x$ may be below the lowest quality
retained in $L$ after truncation. We showed that
$y_L(x)\leq y_H(x)=y$.
Since $y_L$ is strictly increasing, inversion gives
\begin{equation*}
 x_H(y)\leq x_L(y),\qquad
 B_H(y)=B_H^x(x_H(y))\leq B_L^x(x_H(y))\leq B_L(y).
\end{equation*}
Since $\phi_i(y)=C'(x_i(y))$, it follows that
$\phi_H(y)\leq\phi_L(y)$ and
$y-\phi_H(y)\geq y-\phi_L(y)$. Deadweight loss is
$\mathcal D(y,x)=\pi(y)-yx+C(x)$. Below first best,
$\partial_x\mathcal D=C'(x)-y<0$, so the lower quality in $H$
gives weakly higher deadweight loss. Continuity extends these weak
comparisons to both endpoints.

Finally, fix a gross willingness to pay $\theta$ in both realized
supports. Inverting \eqref{eq:oa-loading-level-order} gives
$y_H:=\nu_H(\theta)\leq\nu_L(\theta)=:y_L$ on the common gross
domain. Since $\theta$ is retained in both markets,
$y_L\geq y_H\geq\ell_H^y\geq\ell_L^y$, so both net values
belong to both retained schedules. Monotonicity and the comparison
at a common net value yield
$x_H(\theta)=x_H(y_H)\leq x_H(y_L)\leq x_L(y_L)=x_L(\theta)$.
The interval $[\ell_H,\theta]$ lies in both gross-type supports.
Hence the utility envelope and
\eqref{eq:app-loading-selected-support} imply
\begin{equation*}
 U_H(\theta)=\int_{\ell_H}^{\theta}x_H(z)\,dz
 \leq\int_{\ell_H}^{\theta}x_L(z)\,dz
 \leq\int_{\ell_L}^{\theta}x_L(z)\,dz=U_L(\theta).
\end{equation*}
Writing $x_i=x_i(\theta)$, we also have
$y_Hx_H-C(x_H)\leq y_Lx_H-C(x_H)\leq y_Lx_L-C(x_L)$.
The first inequality uses $y_H\leq y_L$ and $x_H\geq0$; the
second uses $0\leq x_H\leq x_L\leq q(y_L)$, on which total
surplus at net value $y_L$ is increasing in quality. This proves
the total-surplus comparison without differentiating the reoptimized
schedule.
\end{proof}

\paragraph{Completion of Proposition~\ref{thm:main-dependence}.}
Lemma~\ref{prop:oa-loading-profile} proves (i) and (iv) under
Assumptions~\ref{ass:main-nu}--\ref{ass:main-C}.
Proposition~\ref{prop:oa-loading-treatment} proves (ii)--(iii) under
the original curvature condition. For the private-value comparisons,
take private values as environment $L$, so that
$\lambda_L\equiv0\leq\lambda_H$.

\section{Costly Attenuation}
\label{app:paper-attenuation}

We begin by showing that an increase in the raw adjusted revenue curve
weakly raises types at every rank.

\begin{appendixlemma}[Raw adjusted revenue and stochastic dominance]
\label{lem:attenuation-raw-order}
Maintain Assumption~\ref{ass:main-nu}. For any two feasible markets,
$\widehat R_F\geq\widehat R_G$ implies
$Q_F^+(u)\geq Q_G^+(u)$ for every $u<1$.
Consequently, $F\succeq_{\mathrm{st}}G$ and
$F\succeq_{\mathrm{icx}}G$. Moreover, for
$T=\delta_{\bar\theta}$,
$\widehat R_T(u)=(1-u)\bar\nu\geq\widehat R_G(u)$ for every feasible
$G$ and every $u<1$.
\end{appendixlemma}

\begin{proof}
Extend $m$ constantly outside $[0,\bar\theta]$, as in the proof of
Proposition~\ref{lem:gs-regularization}. Recall the reconstruction map
\eqref{eq:app-tail-operator},
\[
 (T_Rf)(u)=\frac{R(u)+\int_u^1m(f(s))\,ds}{1-u},\qquad u<1,
\]
which takes bounded Borel functions to bounded Borel functions whenever
$\sup_{u<1}|R(u)|/(1-u)<\infty$. It is a contraction, since
$\|T_Rf-T_Rh\|_\infty
 \leq(1-\underline r)\|f-h\|_\infty$.
Here $1-\underline r<1$. By
Lemma~\ref{lem:gs-concavification-regularity}, every raw adjusted
revenue curve satisfies the required bound. Moreover, $Q_G$ and
$Q_G^+$ agree almost everywhere, so \eqref{eq:reduction-adjusted-raw} gives
$T_{\widehat R_G}Q_G^+=Q_G^+$. Thus, the right-continuous quantile is
the unique fixed point associated with the raw curve, including for
markets with atoms or gaps in their support.

Suppose $\widehat R_F\geq\widehat R_G$. Then
$T_{\widehat R_F}Q_G^+\geq Q_G^+$. Because $m$ is nondecreasing,
$T_{\widehat R_F}$ preserves pointwise order. Repeatedly applying
this map to $Q_G^+$ therefore gives an increasing sequence of
functions. The contraction property makes this sequence converge
uniformly to $Q_F^+$, proving the quantile comparison. Taking
one-sided limits gives the same result under the remaining
endpoint conventions. Drawing both types at the same uniform rank
now gives first-order stochastic dominance and
\[
 S_F(z)=\int_0^1(Q_F(u)-z)_+\,du
 \geq\int_0^1(Q_G(u)-z)_+\,du=S_G(z),
 \qquad z\in\mathbb R,
\]
which proves increasing-convex dominance.

Finally, monotonicity of the quantile gives
$Q_G(s)\geq Q_G^+(u)$ for almost every $s>u$. Since $m$ is
nondecreasing,
\[
 \frac{\widehat R_G(u)}{1-u}
 =Q_G^+(u)-\frac{1}{1-u}\int_u^1m(Q_G(s))\,ds
 \leq\nu(Q_G^+(u))\leq\bar\nu.
\]
The market concentrated at the highest type attains equality at
every rank, proving the last claim.
\end{proof}

\begin{proof}[Proof of Proposition~\ref{prop:costly-attenuation}]
\emph{1. Existence and reduction.}
Write $T=\delta_{\bar\theta}$ and
$\widehat W_\beta(G)=\CS(G)+\beta\Pi(G)-\mathcal K(G)$.
The preceding lemma and the cost restriction give
\begin{equation}
 \widehat R_F\geq\widehat R_G
 \quad\Longrightarrow\quad
 Q_F^+\geq Q_G^+,
 \qquad \mathcal K(F)\leq\mathcal K(G).
 \label{eq:app-attenuation-order-cost}
\end{equation}
In particular, $T$ minimizes cost. Subtracting $\mathcal K(T)$
therefore makes all costs nonnegative and the cost of $T$ zero,
without changing the optimal markets. The market space is weakly
compact, and both screening payoffs are weakly continuous by
Proposition~\ref{lem:app-compact-density}. Since $\mathcal K$ is
lower semicontinuous, $\widehat W_\beta$ is upper semicontinuous
and attains its maximum.

Fix any optimizer $G$. Its regularization $\widetilde G$ satisfies
$\widehat R_{\widetilde G}=R_G
 =\cav(\widehat R_G)\geq\widehat R_G$.
Proposition~\ref{lem:gs-regularization} shows that this change
preserves profit and weakly raises consumer surplus, while
\eqref{eq:app-attenuation-order-cost} shows that it weakly lowers
cost. Thus $\widetilde G$ is also optimal. Neither consumer surplus
nor cost can change strictly, since either change would strictly
improve the objective.

Next, replace any negative virtual net values of $\widetilde G$
by zero. This weakly raises every upper-tail integral of its profile,
and hence its raw adjusted revenue curve. By
Lemma~\ref{lem:gs-truncation}, both payoffs are unchanged.
Cost cannot rise, so the resulting market is again optimal and has
a profile $\phi\in\Phi_{\underline\theta}$.

\emph{2. Positive quality and regularity of every optimizer.}
First, $\phi$ cannot be identically zero. A zero profile is infeasible
when the minimum feasible type is positive; when that minimum is
zero, it induces $\delta_0$. Choose $\gamma\in(0,\bar\nu)$ and
$b\in(0,1)$, and consider the regular profile equal to $\gamma$
below $b$ and to $\bar\nu$ above it. Its raw adjusted revenue is
nonnegative, so it weakly exceeds $\widehat R_{\delta_0}=0$ and
its cost is no higher. Its profit is nonnegative, and
Lemma~\ref{lem:gs-strict-wedge} gives strictly positive consumer
surplus. It therefore improves the objective for every $\beta\geq0$.

For a nonzero $\phi$, put $M=Q_\phi(0)$ and $v_0=\nu(M)$.
Equation~\eqref{eq:gs-profile-tail} gives
\[
 M=\int_0^1\{\phi(u)+m(Q_\phi(u))\}\,du
 \geq\int_0^1\phi(u)\,du>0,
\]
so $v_0>0$. Suppose a positive mass
$a=|\{u:\phi(u)=0\}|$ receives zero quality. Since $\phi$ is
nondecreasing and nonnegative, these buyers occupy an initial
interval of ranks. Choose $0<\varepsilon<a v_0$ and set
$\psi(u)=\max\{\phi(u),\varepsilon\}$ and $E_\varepsilon=\{u:\phi(u)<\varepsilon\}$.
The profile $\psi$ remains monotone and within its bounds, since
$\varepsilon<\bar\nu$. It weakly raises the raw adjusted revenue
curve and the induced quantiles, so it preserves feasibility,
weakly raises profit, and weakly lowers cost.

To compare consumer surplus, use \eqref{eq:gs-CS-formula}.
On $E_\varepsilon$, virtual value is raised to $\varepsilon$.
Since induced types weakly rise, the change in the consumer surplus
integrand is at least
\[
 [\nu(Q_\phi)-\varepsilon]
     [q(\varepsilon)-q(\phi)]
       -(\varepsilon-\phi)q(\phi).
\]
At ranks where $\phi=0$, this
expression is at least $(v_0-\varepsilon)q(\varepsilon)$, since
$\nu(Q_\phi)\geq v_0$ and $q(0)=0$. On the remaining ranks in
$E_\varepsilon$, its first term is nonnegative and its second term
is at least $-\varepsilon q(\varepsilon)$. Outside
$E_\varepsilon$, quality is unchanged and types weakly rise, so
the integrand cannot fall. Integrating these bounds gives
\begin{equation*}
 \CS(G_\psi)-\CS(G_\phi)
 \geq q(\varepsilon)
       [a v_0-\varepsilon|E_\varepsilon|]
 \geq q(\varepsilon)(a v_0-\varepsilon)>0.
\end{equation*}
Profit has weakly increased and cost has weakly fallen. Since
$\beta\geq0$, the strict consumer surplus gain contradicts
optimality under $\widehat W_\beta$.

The truncated profile is therefore positive almost everywhere.
Regularization and truncation preserve the original market's
allocation by rank, so $q(\phi_G)>0$ almost everywhere as well.
Recall that regularization left consumer surplus unchanged.
Its equality condition \eqref{eq:gs-reg-equality} then gives
$(Q_{\widetilde G}-Q_G)q(\phi_G)=0$ almost everywhere.
Since quality is positive almost everywhere, this implies
$Q_{\widetilde G}=Q_G$ almost everywhere and hence
$G=\widetilde G$. The original optimizer was arbitrary, so every
optimizer is regular and assigns zero quality to no positive mass
of buyers. This proves part~(i).

\emph{3. The structure of optimal markets.}
For $\beta\geq1$, Proposition~\ref{prop:rw-existence-top} shows
that $T$ uniquely maximizes the objective without attenuation costs.
It also minimizes cost, so it remains uniquely optimal when cost is
subtracted. This proves part~(ii).

Now, fix $\beta\in[0,1)$ and let $\phi$ be the profile of an
optimal market different from $T$. By Step 2, the market
is regular and $\phi$ is positive almost everywhere.
The profile must equal $\bar\nu$ on an interval of highest
ranks. Otherwise, the upward perturbation in
Lemma~\ref{lem:rw-top-atom} would strictly increase
$J_\beta$ and weakly raise the raw adjusted revenue curve.
Cost could not increase, so the perturbation would improve
the objective. Thus $\phi=\bar\nu$ on $[b,1]$ for some
$b\in(0,1)$, where $b>0$ because the market is not $T$.
Lemma~\ref{lem:gs-strict-wedge} then gives connected
support, no atoms below $\bar\theta$, and an atom at
$\bar\theta$ receiving first-best quality.

Also, $\phi$ cannot be constant on a nondegenerate interval
below $b$. The perturbation in
Lemma~\ref{lem:attenuation-no-bunching-tilt} strictly
increases $J_\beta$ and, by
\eqref{eq:attenuation-balanced-tail}, weakly raises every
upper-tail integral. It therefore weakly raises the raw
adjusted revenue curve and cannot increase cost.
This would again improve the objective. Hence $\phi$
is strictly increasing on $(0,b)$, and the allocation
strictly separates interior types.

\emph{4. The cutoff.}
For every $G\ne T$,
\[
 \Pi(G)\leq\TS(G)
 \leq\int\pi(\nu(\theta))\,dG(\theta)
 <\pi(\bar\nu).
\]
The last inequality is strict because $G$ puts positive mass on types
$\theta<\bar\theta$, where $\pi(\nu(\theta))<\pi(\bar\nu)$.
Thus $d(G):=\pi(\bar\nu)-\Pi(G)>0$, and
$\widehat W_\beta(G)-\widehat W_\beta(T)
 =\CS(G)-\mathcal K(G)-\beta d(G)$.
Define
\begin{equation*}
 \beta^{\mathcal K}
 :=\max\left\{0,\sup_{G\ne T}
       \frac{\CS(G)-\mathcal K(G)}
            {\pi(\bar\nu)-\Pi(G)}\right\}.
\end{equation*}
Each ratio is strictly below one, since
$\pi(\bar\nu)-\TS(G)+\mathcal K(G)>0$. Their supremum can equal
one, so $\beta^{\mathcal K}\in[0,1]$. The objective comparison
above shows that $T$ is optimal when
$\beta\geq\beta^{\mathcal K}$. If the inequality is strict,
$T$ strictly beats every other market and is therefore uniquely
optimal. At the cutoff, a market other than $T$ ties with it if
and only if its ratio equals $\beta^{\mathcal K}$. If the cutoff
is one, part~(ii) gives uniqueness there as well. This proves
part~(iii).

\emph{5. Expected attenuation costs.}
Suppose $\mathcal K(G)=\E_G[k(\theta)]$ as in part~(iv).
The cost normalization and order imply $k\geq0$ and
$k(\bar\theta)=0$. Fix $\beta\in[0,1)$. We will construct a
market that strictly improves on $T$. Choose a fixed
$\eta\in(0,\bar\nu)$ small enough that
$a_\eta:=\eta q(\bar\nu-\eta)
       -\beta\{\pi(\bar\nu)-\pi(\bar\nu-\eta)\}>0$.
Such a choice is possible because $\pi'=q$ and
$a_\eta/\eta\to(1-\beta)\bar x>0$ as $\eta\downarrow0$.
For small $\varepsilon>0$, let
\[
 \lambda_\varepsilon=\frac{\varepsilon}{\eta+\varepsilon},
 \qquad
 G_\varepsilon=\lambda_\varepsilon\delta_{\bar\theta-\varepsilon}
                +(1-\lambda_\varepsilon)T,
 \qquad
 z_\varepsilon=\nu(\bar\theta-\varepsilon)-\eta.
\]
Choose $\varepsilon<\bar\theta-\underline\theta$ and small
enough that $0<z_\varepsilon<\bar\nu$. This market moves a mass
$\lambda_\varepsilon$ to a type just below the highest type.
The two-type calculation in
Proposition~\ref{prop:rw-existence-top} gives
\[
 \CS(G_\varepsilon)
   =(1-\lambda_\varepsilon)\varepsilon q(z_\varepsilon)
   =\lambda_\varepsilon\eta q(z_\varepsilon),
 \qquad
 \Pi(G_\varepsilon)
   =\lambda_\varepsilon\pi(z_\varepsilon)
      +(1-\lambda_\varepsilon)\pi(\bar\nu).
\]
Consequently,
\begin{equation*}
 \widehat W_\beta(G_\varepsilon)-\widehat W_\beta(T)
 =\lambda_\varepsilon\bigl[
    \eta q(z_\varepsilon)
      -\beta\{\pi(\bar\nu)-\pi(z_\varepsilon)\}
    -k(\bar\theta-\varepsilon)\bigr].
\end{equation*}
As $\varepsilon$ decreases to zero, the expression in brackets
converges to $a_\eta>0$, by continuity of $\nu,q,\pi$ and of $k$
at $\bar\theta$. Since $\lambda_\varepsilon>0$, the new market
strictly improves on $T$ for sufficiently small $\varepsilon$.
Thus $T$ is not optimal for any $\beta<1$. Together with
part~(ii), this proves $\beta^{\mathcal K}=1$.
\end{proof}

\paragraph{The expected-cost examples.}
Let $k$ be finite, nonincreasing, and lower semicontinuous on the
compact type interval. Monotonicity bounds $k$ between its finite
endpoint values. Hence $G\mapsto\E_G[k(\theta)]$ is finite and
weakly lower semicontinuous. First-order stochastic dominance
weakly lowers this expectation, so
Lemma~\ref{lem:attenuation-raw-order} implies
\eqref{eq:costly-attenuation-order}. No concavity of $k$ is needed.
The quadratic example need not satisfy the stronger restriction
based on increasing-convex order. Indeed, for $a<b$ in the type
interval, $\frac12\delta_a+\frac12\delta_b$ dominates
$\delta_{(a+b)/2}$ in that order, but strict convexity of the
individual cost makes its expectation strictly larger.

For \eqref{eq:costly-attenuation-example}, substitute
$S_G(z)=\E_G[(\theta-z)_+]$ and change the order of integration
to obtain $\mathcal K_\omega(G)=\E_G[k_\omega(\theta)]$, where
\[
 k_\omega(\theta)
 =\int_{\underline\theta}^{\bar\theta}\omega(z)
       \bigl[(\bar\theta-z)-(\theta-z)_+\bigr]\,dz.
\]
The interchange is valid because the integrand is nonnegative
and bounded by $(\bar\theta-\underline\theta)\omega(z)$, with
$\omega$ integrable. The function $k_\omega$ is nonnegative and
nonincreasing, and vanishes at $\bar\theta$. Since the positive-part
function is Lipschitz with constant one, $k_\omega$ is Lipschitz
with constant at most $\int\omega(z)\,dz$. It therefore satisfies
part~(iv). Continuity at the highest type is an additional
restriction: lower semicontinuity alone permits a downward jump
of $k$ at $\bar\theta$.

\paragraph{Coexistence can occur.}
Fix $\beta_0\in[0,1)$ and let
$k_0:=U_{\underline\theta}(\beta_0)-\beta_0\pi(\bar\nu)>0$,
where $U$ is the value without attenuation costs in
\eqref{eq:model-support-problem}. Positivity follows from
Proposition~\ref{prop:rw-existence-top}. Set
$\mathcal K(G)=k_0\mathbf1_{\{G\ne T\}}$.
This cost is finite and weakly lower semicontinuous because
$\{T\}$ is closed. It also satisfies
\eqref{eq:costly-attenuation-order}. A violation could only occur
with $G=T$ and $F\ne T$, but raw adjusted revenue dominance would
then require $F$ to dominate the highest-type market, which the
quantile comparison rules out.

At $\beta_0$, any optimizer $G_0\ne T$ of the problem without
attenuation costs satisfies
\[
 \widehat W_{\beta_0}(G_0)
 =U_{\underline\theta}(\beta_0)-k_0
 =\beta_0\pi(\bar\nu)
 =\widehat W_{\beta_0}(T),
\]
and no market has a larger value. Every market other than $T$ has
strictly smaller profit than $T$, as shown in the cutoff argument.
Thus, for every admissible $\beta<\beta_0$, $G_0$ strictly beats
$T$; for every $\beta>\beta_0$, $T$ strictly beats every other
market. Hence $\beta^{\mathcal K}=\beta_0$, with coexistence at
the cutoff.

\section{Payoff-preserving normalizations}
\label{app:payoff-normalization}

We verify the two transformations described in
Section~\ref{subsec:normalization-applications}: first the representation
of primitive payoffs, and then the normalization at the highest type
that does not trade under the first best.

\paragraph{The common-index payoff representation.}
On compact primitive type and allocation intervals, suppose
$v_{ay}(a,y)=z(a)k(y)$ with continuous strictly positive $z,k$.
Retain the primitive function $X$ with $X'=k$, let $a_0,y_0$ be
the respective lower endpoints, and define
\par\noindent\begin{minipage}{\linewidth}
\[
 X_0=X(y_0),\qquad Z(a)=\int_{a_0}^az(s)\,ds,
 \qquad h(y)=v(a_0,y)-v(a_0,y_0),
\]
\[
 g(a)=v(a,y_0)-Z(a)X_0.
\]
\end{minipage}\par\noindent
Integrating the cross-partial in both variables gives
$v(a,y)-g(a)=Z(a)X(y)+h(y)$. Since $z$ and $k$ are strictly
positive, both $Z$ and $X$ are strictly increasing. We can therefore
use them to relabel type and allocation. With the transfer adjustment
\[
 \vartheta=Z(a),\qquad x=X(y),\qquad \widetilde t=\tau-h(y),
 \qquad
 \widetilde m=\mu\circ Z^{-1},\quad
 \widetilde C=(K-h)\circ X^{-1}
\]
the payoffs become
\[
 \begin{aligned}
 v(a,y)-\tau-g(a)
 &=\vartheta x-\widetilde t,\\
 \tau-\mu(a)X(y)-K(y)
 &=\widetilde t-\widetilde m(\vartheta)x-\widetilde C(x).
 \end{aligned}
\]
The amount $h(y)$ deducted from each transfer depends
only on the contract's allocation, so the adjustment
does not require observing the buyer's type. For each
type $a$, the transformation subtracts $g(a)$ from the
utility of every contract. Subtracting it from her
outside option as well preserves incentive compatibility
and participation. The second identity shows that the
principal's payoff from each contract is unchanged.

For the zero allocation specification $X(y_0)=K(y_0)=0$, the
definitions give $g(a)=v(a,y_0)$ and $\widetilde C(0)=0$.

\paragraph{Affine normalization at the no-trade threshold.}
Suppose the payoffs already are $\vartheta x-\widetilde t$ and
$\widetilde t-\widetilde m(\vartheta)x-\widetilde C(x)$, where
$\widetilde C(0)=0$, $\widetilde C$ is differentiable and strictly
convex, and $\widetilde\nu(\vartheta)=\vartheta-\widetilde m(\vartheta)$
is strictly increasing. Let $c_0=\widetilde C'(0)$ and suppose
$\widetilde\nu(\vartheta^0)=c_0$, with
$\vartheta^0\leq\underline\vartheta<\bar\vartheta$ and primitives
defined on $[\vartheta^0,\bar\vartheta]$.
For any $d>0$, set
\[
 \theta=\frac{\vartheta-\vartheta^0}{d},\qquad
 t=\frac{\widetilde t-\vartheta^0x}{d},
\]
\[
 m(\theta)=\frac{\widetilde m(\vartheta^0+d\theta)
                         -\widetilde m(\vartheta^0)}{d},
 \qquad
 C(x)=\frac{\widetilde C(x)-c_0x}{d}.
\]
The identity $\vartheta^0=\widetilde m(\vartheta^0)+c_0$ gives
\[
 \begin{aligned}
 \vartheta x-\widetilde t&=d(\theta x-t),\\
 \widetilde t-\widetilde m(\vartheta)x-\widetilde C(x)
 &=d\{t-m(\theta)x-C(x)\}.
 \end{aligned}
\]
These identities hold for every contract, including
deviations. Division by $d>0$ therefore preserves
incentive compatibility, participation, and seller
optimality. Expected consumer surplus and profit are
each divided by $d$.

The increasing relabeling of types preserves ranks
and atom masses, and transforms the quantile function
to $(Q-\vartheta^0)/d$. Physical allocations remain
unchanged. The definitions also give
\par\noindent\begin{minipage}{\linewidth}
\[
 m(0)=C(0)=C'(0)=0,\qquad
 \nu(\theta)=\frac{\widetilde\nu(\vartheta^0+d\theta)-c_0}{d},
\]
\[
 (\underline\theta,\bar\theta)
 =\left(\frac{\underline\vartheta-\vartheta^0}{d},
         \frac{\bar\vartheta-\vartheta^0}{d}\right).
\]
\end{minipage}\par\noindent

To identify the no-trade threshold, strict convexity gives
$\widetilde C(x)>c_0x$ for $x>0$. If
$\widetilde\nu\leq c_0$, every positive quality therefore yields
negative net trade surplus, so first-best quality is zero. If
$\widetilde\nu>c_0$, the right derivative of net trade surplus at
zero is positive, so some positive quality is profitable.
$\widetilde\nu$ being strictly increasing then makes $\vartheta^0$ the
highest no-trade type. Under the stated restriction
$\vartheta^0\leq\underline\vartheta$, this type is feasible
when $\underline\theta=0$.

Finally, at corresponding types, wherever the derivatives exist,
\[
 m'=\widetilde m',\qquad m''=d\widetilde m'',\qquad
 \nu'=\widetilde\nu',\qquad
 C'=\frac{\widetilde C'-c_0}{d},\quad
 C''=\frac{\widetilde C''}{d},\quad
 C'''=\frac{\widetilde C'''}{d}.
\]
Thus the slope restrictions are preserved, as is the curvature
ratio $-xC'''/C''$. Convexity gives $C\geq0$, and nondecreasing
$\widetilde m$ gives $m\geq0$ on nonnegative normalized types.
The top-quality condition is preserved as well:
$\widetilde C'(\bar x)=\widetilde\nu(\bar\vartheta)$ implies
$C'(\bar x)=\bar\nu$.

\end{document}